\documentclass{article}
\usepackage{latexsym}
\usepackage{amssymb}
\usepackage{amsmath,amsthm}
\usepackage[numbers,sort&compress]{natbib}
\usepackage{color}
\usepackage{amsfonts,amssymb}
\usepackage{mathrsfs}
\usepackage{dsfont}
\usepackage{graphicx}
\usepackage{subfigure}
\usepackage{float}
\usepackage{ulem}
\usepackage{color}
\usepackage{appendix}
\newcommand {\dif}{\mathrm{d}}

\newtheorem{prop}{Proposition}
\theoremstyle{definition}
\newtheorem{theorem}{Theorem}[section]
\newtheorem{lemma}{Lemma}[section]
\newtheorem{remark}{Remark}[section]
\begin{document}
	\title{Integrable discretizations for a generalized sine-Gordon	equation and the reductions to the sine-Gordon equation and the short pulse equation}
	\author{Han-Han Sheng$^\dag$, Bao-Feng Feng$^\ddag$ and Guo-Fu Yu$^\dag$\footnote{Corresponding author. Email address: gfyu@sjtu.edu.cn}\\
		$^\dag$ School of Mathematical Sciences, CMA-Shanghai, Shanghai Jiao Tong University, \\
		Shanghai 200240, P.R.\ China \\
		$^\ddag$ School of Mathematical and Statistical Sciences,\\
		The University of Texas Rio Grande Valley, Edinburg, Texas 78541, USA}
	\date{}
	\maketitle

	\begin{abstract}
	{In this paper, we propose fully discrete analogues of a  generalized sine-Gordon (gsG) equation	$u_{t x}=\left(1+\nu \partial_x^2\right) \sin u$. The bilinear equations of the discrete KP hierarchy and the proper definition of discrete hodograph transformations are the keys to the construction. Then we derive semi-discrete analogues of the gsG equation from the fully discrete gsG equation by taking the temporal parameter $b\rightarrow0$. 	Especially,  one full-discrete gsG equation is reduced to a semi-discrete gsG equation in the case of $\nu=-1$ (Feng {\it et al. Numer. Algorithms} 2023). Furthermore, $N$-soliton solutions to the semi- and fully discrete analogues of the gsG equation in the determinant form are constructed. Dynamics of one- and two-soliton solutions for the discrete gsG equations are discussed with plots. We also investigate the reductions to the sine-Gordon (sG) equation and the short pulse (SP) equation. By introducing an important parameter $c$, we demonstrate that the gsG equation reduces to the sG equation and the SP equation, and the discrete gsG equation reduces to the discrete sG equation and the discrete SP equation, respectively, in the appropriate scaling limit. The limiting forms of the $N$-soliton solutions to the gsG equation also correspond to those of the sG equation and the SP equation.}
	\end{abstract}
	\section{Introduction}
	In this paper, we are concerned with the integrable discretizations of a generalized sine-Gordon (gsG) equation
	\begin{align}
	u_{t x}=\left(1+\nu \partial_x^2\right) \sin u,\label{gsg}
	\end{align}
	where $u = u(x, t)$ is a scalar-valued function, $\nu$ is a real parameter which can be normalized into $\nu=\pm 1$, $\partial_x^2$ and the
	subscripts $t$ and $x$ appended to $u$ denote partial differentiation. The gsG equation \eqref{gsg} was first derived by Fokas in 1995 using bi-Hamiltonian methods \cite{gsg}. For $\nu=-1$, the integrability was established by the Lax pair, and the initial value problem for decaying initial data was solved by the inverse scattering method \cite{gsg-fl}. Eigenfunctions of the Lax pair and traveling-wave solutions were obtained through the Riemann-Hilbert formalism \cite{gsg-fl}. A variety of solutions, including kinks, loop solitons, and breathers, were recognized  from the general soliton solution in parametric form constructed by Hirota's  bilinear method  \cite{gsg-1}. The gsG equation \eqref{gsg} with $\nu=1$ was solved by Hirota's bilinear method \cite{gsg1}. It should be commented here that the structure of the solutions to equation \eqref{gsg} with $\nu=1$ is significantly different from that with $\nu=-1$, as it does not admit multi-valued solutions like loop solitons \cite{gsg-1,gsg1}.	
	Quite recently, we constructed several semi-discrete analogues of the gsG equation with $\nu=-1$ and presented the determinant formulae of $N$-soliton solutions both for the gsG equation and the semi-discrete gsg equation in \cite{gsg-F}.

	The study of discrete integrable systems, which is connected to many other disciplines such as quantum field theory, numerical algorithms, random matrices, and orthogonal and biorthogonal polynomials, has recently received a lot of interest \cite{Disbook}. Compared to continuous integrable systems, there are far fewer examples of discrete integrable systems and analytical tools available for studying them. On the other hand, it is widely believed that discrete integrable systems are more fundamental and universal than their continuous counterparts. The authors have conducted a substantial amount of research in finding integrable discretizations of soliton equations, including the short pulse (SP) equation \cite{dis-sp,Feng1}, the (2+1)-dimensional Zakharov equation \cite{Yu}, the Camassa-Holm equation \cite{Feng2,d-CH3}, the Degasperis-Proceli equation \cite{dDP}, and the modified Camassa-Holm equation \cite{dmCH} via Hirota's bilinear method. Building upon the compatibility between an integrable system and its B{\" a}cklund transformation, a systematic procedure was proposed for obtaining discrete versions of integrable PDEs using Hirota's bilinear method \cite{backlund}.
	
	It was demonstrated that the gsG equation with $\nu=\pm1$ reduces to the SP equation in the short wave limit and to the sine-Gordon (sG) equation in the long wave limit \cite{gsg-1,gsg1}. Thus, the gsG equation is an interesting soliton equation that lies between the SP equation and the sG equation. Since the semi-discrete and fully discrete sG equation has been well known \cite{sm-sg,fd-sg1,fd-sg2,fd-sg3}, and the integrable discretization of the SP equation was recently proposed by one of the authors in \cite{dis-sp}, it would be quite an interesting problem to  construct the integrable discretizations of the gsG equation  \eqref{gsg}, which is indeed the main motivation of our present paper.
	
	Another challenging problem is looking for the reductions from the gsG equation to the sG equation and the SP equation, both in the continuous case and in the discrete case. Although the reductions from the gsG equation to the sG equation and the SP equation were proposed in \cite{gsg-1,gsg1}, the  reductions in the discrete case differ substantially from those in the continuous case. The problem lies in the fact that we not only apply scaling transformations to the original variables in the equation, i.e., $u$, $x$, and $t$, but also transform the new variables after hodograph transformation and parameters in the $\tau$-function such as $y$, $\tau$, and $p$ in \cite{gsg-1,gsg1}. It is difficult to find the correspondence in the discrete case. To obtain the reductions, we introduce an important parameter $c$, which can be viewed as the coefficient of the B{\"a}clund transformation between two sets of bilinear equations of the two-dimensional Toda-lattice (2DTL) equation.
	
	In this paper, we construct integrable fully-discrete analogues of the gsG equation \eqref{gsg} with $\nu=\pm1$ from two sets of discrete bilinear 2DTL equations, and the corresponding determinant solutions are obtained. In addition, the semi-discrete gsG equation with $\nu=\pm1$ is constructed from the fully discrete gsG, of which the semi-discrete gsG equation with $\nu=-1$ agrees with our previous result in \cite{gsg-F}.
		Moreover, the connections of the discrete gsG equation to the discrete sG equation and the discrete SP equation are clarified by appropriate but different scaling limits.
	
	The remainder of the paper is organized as follows. In Section \ref{sec2}, we review the bilinear equations
	and determinant solutions of the gsG equation with $\nu = 1$, which can be reduced from the bilinear equations of 2DTL and their  B{\"a}cklund transformation. We demonstrate that equation \eqref{gsg} reduces to the sG equation and the SP equation with the scaling transformation on the original variables and the corresponding limits of $c$. The limiting forms of the $N$-soliton solution also correspond to the known
	{solutions of the sG and SP equations}. In Section \ref{sec3}, starting with two sets of bilinear discrete 2DTL equations, we derive a fully integrable discrete analog of the gsG equation \eqref{gsg} with $\nu=\pm1$ and present its $N$-soliton solutions. Two conserved quantities of the fully discrete gsG equation are obtained. In Section \ref{sec4}, we propose the semi-discrete gsG equation from two approaches: the fully discrete gsG equation and the semi-discrete 2DTL equation, respectively. Reductions to the corresponding discrete analogues of the sG equation and the SP equation are also investigated. In Section \ref{fig}, we present soliton solutions to the semi- and fully discrete gsG equation and investigate their properties, focusing mainly on one- and two-soliton solutions. Section \ref{sec6} is devoted to a brief summary and discussion. Some detailed proofs are given in Appendices.
	
	\section{From the 2DTL equation and its B{\"a}cklund transformation to the gsG equation with $\nu=1$}\label{sec2}
	In this section, we show that the bilinear equations and the multisoliton solution to the gsG equation \eqref{gsg} with $\nu=1$ given by Matsuno \cite{gsg1} can be generated from two sets of bilinear equations of the 2DTL equation and its B{\"a}cklund transformation between them and its determinant solution through a series of reductions and transformations, including the hodograph transformation and dependent variable transformation.
	\subsection{A brief review of the gsG equation with $\nu=1$}
	Firstly, we give a brief review of the results in \cite{gsg1} about the bilinear form of the gsG equation \eqref{gsg} with $\nu=1$
	\begin{align}
	u_{t x}=\left(1+\partial_x^2\right) \sin u.\label{gsg1}
	\end{align}
	Through the new dependent variable {\it r} in accordance with the relation
	\begin{align}
	r^2=1-u_x^2,\qquad (0<r<1),\label{r}
	\end{align}
	the gsG equation can be rewritten as
	\begin{align}
	r_t-(r\cos u)_x=0,\label{con-law}
	\end{align}
	which is exactly a conservation law of \eqref{gsg1}. Then we define the hodograph transformation $(x,t)\rightarrow(y,\tau)$ by
	\begin{align}
	\dif y = \lambda r \dif x + \lambda r \cos u \dif t, \qquad \dif \tau =\lambda^{-1} \dif  t,\label{hodo}
	\end{align}
	where $\lambda\in\mathbb{R}$ is a constant. The derivatives for $x$ and $t$ are then rewritten in terms of $y$ and $\tau$ as
	\begin{align}
	\frac{\partial}{\partial x}=\lambda r \frac{\partial}{\partial y}, \quad \frac{\partial}{\partial t}=\lambda^{-1}\frac{\partial}{\partial \tau}+\lambda r \cos u \frac{\partial}{\partial y} .
	\end{align}
	With the new variables $y$ and $\tau$, \eqref{r} and \eqref{con-law} are recast into the form
	\begin{align}
	& r^2=1-\lambda^2r^2 u_y^2,\label{r1} \\
	& \lambda^{-1}\left(\frac{1}{r}\right)_\tau+\lambda(\cos u)_y=0,\label{con-law1}
	\end{align}
	respectively. Further reduction is possible if one defines the variable $\varphi$ by
	\begin{align}
	\lambda u_y=\sinh \varphi, \quad \varphi=\varphi(y, \tau).\label{def-phi}
	\end{align}
	It follows from \eqref{r1} and \eqref{def-phi} that
	\begin{align}
	\frac{1}{r}=\cosh \varphi .\label{r-phi}
	\end{align}
	Substituting \eqref{def-phi} and \eqref{r-phi} into equation \eqref{con-law1}, we find
	\begin{align}
	\varphi_\tau=\lambda\sin u .\label{u-phi}
	\end{align}
	Let $\sigma$ and $\sigma^{\prime}$ be solutions of sG equation
	\begin{align}
	& \sigma_{\tau y}=\sin \sigma, \quad \sigma=\sigma(\tau, y),\label{sg-1} \\
	& \sigma_{\tau y}^{\prime}=\sin \sigma^{\prime}, \quad \sigma^{\prime}=\sigma^{\prime}(\tau, y) .\label{sg-2}
	\end{align}
	Then we put
	\begin{align}
	&u=\frac{1}{2}(\sigma+\sigma'),\\
	&\varphi=\frac{1}{2\mathrm{i}}(\sigma-\sigma').
	\end{align}
	In terms of $\sigma$ and $\sigma^{\prime}$, equations \eqref{def-phi} and \eqref{u-phi} can be written as
	\begin{align}
	&\frac{1}{2}(\sigma+\sigma')_y=\lambda^{-1}\sinh\left(\frac{1}{2\mathrm{i}}(\sigma-\sigma')\right)=-\mathrm{i}\lambda^{-1}\sin\frac{1}{2}(\sigma-\sigma'),\label{non-bt1}\\
	&\frac{1}{2}(\sigma-\sigma')_\tau=\mathrm{i}\lambda\sin\frac{1}{2}(\sigma+\sigma').\label{non-bt2}
	\end{align}
	Introducing the dependent variable transformation
	\begin{align}
	&\sigma=2\mathrm{i}\ln\frac{\bar{g}}{f},\\
	&\sigma'=2\mathrm{i}\ln\frac{\bar{f}}{g},
	\end{align}
	where $\bar{f}$ and $\bar{g}$ denote the complex conjugate of $f$ and $g$, respectively. From \eqref{sg-1} and \eqref{sg-2}, one can obtain
	\begin{align}
	&D_\tau D_yf\cdot f=\frac{1}{2}(f^2-\bar{g}^2),\\
	&D_\tau D_y g\cdot g=\frac{1}{2}(g^2-\bar{f}^2),
	\end{align}
	and their complex conjugates. And from \eqref{non-bt1}-\eqref{non-bt2}, we have
	\begin{align}
	&\mathrm{i}\left( \ln\frac{\bar{f}\bar{g}}{fg}\right)_y=\frac{1}{2\lambda}\left(\frac{g\bar{g}}{f\bar{f}}-\frac{f\bar{f}}{g\bar{g}}\right),\label{bt-new1}\\
	&\left( \ln\frac{g\bar{g}}{f\bar{f}}\right)_\tau=\frac{\mathrm{i}\lambda}{2}\left(\frac{\bar{f}\bar{g}}{fg}-\frac{fg}{\bar{f}\bar{g}}\right).\label{bt-new2}
	\end{align}
	By using
	\begin{align}
	\frac{D_\tau f\cdot g}{fg}=\left(\ln\frac{f}{g}\right)_\tau,
	\end{align}
	Eqs. \eqref{bt-new1} and \eqref{bt-new2} can be represented as
	\begin{align}
	&\mathrm{i}\lambda D_yf\cdot \bar{f}-\frac{1}{2}(f\bar{f}-g\bar{g})=0,\label{bl1}\\
	&\mathrm{i}\lambda D_yg\cdot \bar{g} +\frac{1}{2}(g\bar{g}-f\bar{f})=0,\label{bl2}\\
	&\mathrm{i}\lambda^{-1}D_\tau f\cdot g+\frac{1}{2}(fg-\bar{f}\bar{g})=0,\label{bl3}\\
	&\mathrm{i}\lambda^{-1}D_\tau \bar{g}\cdot \bar{f}+\frac{1}{2}(\bar{f}\bar{g}-fg)=0,\label{bl4}
	\end{align}
	here Eq. \eqref{bl4} is the complex conjugate of Eq. \eqref{bl3}. In addition, the definition of $u$ and $\varphi$ can be expressed as
	\begin{align}
	&u=\mathrm{i}\ln\frac{\bar{f}\bar{g}}{fg},\label{u}\\
	&\varphi=\ln\frac{g\bar{g}}{f\bar{f}}.\label{phi}
	\end{align}
	
	\subsection{From the 2DTL equation and its B{\"a}cklund transformation to the bilinear form of the gsG equation \eqref{gsg1}}
	We give two sets of bilinear equations of the 2DTL equation with $\tau$-functions $\tau_n$ and $\tau_n^{\prime}$ and a B{\"a}cklund transformation(BT) between them, respectively,
	\begin{align}
	& \left(\frac{1}{2} D_{x_{-1}} D_{x_1}-1\right) \tau_n \cdot \tau_n=-\tau_{n+1} \tau_{n-1},\label{2dtl1} \\
	& \left(\frac{1}{2} D_{x_{-1}} D_{x_1}-1\right) \tau_n^{\prime} \cdot \tau_n^{\prime}=-\tau_{n+1}^{\prime} \tau_{n-1}^{\prime},\label{2dtl2}\\
	& \left(c^{-1}D_{x_{-1}}-1\right) \tau_n \cdot \tau_n^{\prime}+\tau_{n+1} \tau_{n-1}^{\prime}=0, \label{bt1}\\
	& \left(cD_{x_{1}}-1\right) \tau_{n+1} \cdot \tau_n^{\prime}+\tau_n \tau_{n+1}^{\prime}=0.\label{bt2}
	\end{align}
	Here $c$ is a constant. As mentioned in \cite{ohta-rToda}, we can apply the BT recursively and denote the $\tau$-functions of the $m$-th 2DTL equation as $\tau_{n,m}$. Then we rewrite Eqs. \eqref{2dtl1}-\eqref{bt2} as
	\begin{align}
	&\left(\frac{1}{2} D_{x_{-1}} D_{x_1}-1\right) \tau_{n,m} \cdot \tau_{n,m}=-\tau_{n+1,m} \tau_{n-1,m},\label{2dtl}\\
	& \left(c^{-1}D_{x_{-1}}-1\right) \tau_{n,m} \cdot \tau_{n,m+1}+\tau_{n+1,m} \tau_{n-1,m+1}=0, \label{BT1}\\
	& \left(cD_{x_{1}}-1\right) \tau_{n+1,m} \cdot \tau_{n,m+1}+\tau_{n,m}\tau_{n+1,m+1}=0, \label{BT2}
	\end{align}
	where  we have the following correspondence: $\tau_n  \equiv \tau_{n,m}$, $\tau^{\prime}_n \equiv  \tau_{n,m+1}$. The above equations \eqref{2dtl}-\eqref{BT2} have exact solutions in Casorati determinant form with arbitrary parameters as follows. 
	\begin{lemma}\label{Casorati}
		The bilinear equations \eqref{2dtl}-\eqref{BT2} have the following Casorati-type determinant solution
		\begin{align}
		\tau_{n,m}(x_{-1},x_1)& =\left|\begin{array}{cccc}
		\phi_n^{(1)}(m) & \phi_{n+1}^{(1)}(m) & \cdots & \phi_{n+N-1}^{(1)}(m) \\
		\phi_n^{(2)}(m) & \phi_{n+1}^{(2)}(m) & \cdots & \phi_{n+N-1}^{(2)}(m) \\
		\cdots & \cdots & \cdots & \cdots \\
		\phi_n^{(N)}(m) & \phi_{n+1}^{(N)}(m) & \cdots & \phi_{n+N-1}^{(N)}(m)
		\end{array}\right|,
		\end{align}
		where
		\begin{align}
		& \phi_n^{(i)}(m)=c_ip_i^n\left(1-cp_i\right)^m e^{\xi_i}+d_iq_i^n\left(1-cq_i\right)^m e^{\eta_i},\\
		& \xi_i=p_i x_1+p_i^{-1} x_{-1}+\xi_{i 0}, \quad \eta_i=q_i x_1+q_i^{-1} x_{-1}+\eta_{i 0}.
		\end{align}
		Here, $c_i,d_i,p_i, q_i, \xi_{i0}$ and $\eta_{i0}$ are the arbitrary parameters that can take either real or complex values.
	\end{lemma}
	
	In addition to the Casorati determinant solution, the $\tau$-functions can also be expressed by the Gram-type determinant, which is given by the following lemma.
	\begin{lemma}\label{Gram}
		The following Gram-type determinants satisfy bilinear equations \eqref{2dtl}-\eqref{BT2}
		\begin{align}
		\tau_{n,m}(x_{-1},x_1)=\left|m_{ij}^{n,m}\right|_{N\times N}=\left|c_{ij}+\frac{1}{p_i+q_j}\left(-\frac{p_i}{q_j}\right)^n\left(
		\frac{1-cp_i}{1+cq_j}
		\right)^me^{\xi_i+\eta_j} \right|_{N\times N},
		\end{align}
		where
		\begin{align}
		\xi_i=p_i x_1+p_i^{-1} x_{-1}+\xi_{i 0}, \quad \eta_j=q_j x_1+q_j^{-1} x_{-1}+\eta_{j 0}
		\end{align}
		Here, $p_i, q_j, \xi_{i 0}$ and $\eta_{j 0}$ are the arbitrary parameters that can take either real or complex values and $N\in \mathbb{R}$.
	\end{lemma}
	
	Next we are ready to obtain $N$-soliton solution of the gsG equation \eqref{gsg1}. Firstly, we set $c=\lambda\mathrm{i}\  (\lambda\in \mathbb{R})$.
	$$$$
	{\large\bf Case 1: for real $p_i$,\,$q_i$}
	
	We impose restrictions on the parameters
	\begin{align}
	p_i=-q_i,\,d_{i}=\mathrm{i}\left({\frac{1-cp_i}{1-cq_i}}\right)^{\frac{1}{2}}c_i,\label{rest-cas}
	\end{align}
	for Casorati-type solution or
	\begin{align}
	p_i=q_i,\,c_{ij}=\mathrm{i}\left({\frac{1-cp_i}{1+cq_j}}\right)^{\frac{1}{2}}\delta_{ij},\label{rest-gram}
	\end{align}
	for Gram-type solution. In addition, we take variable transformations $y=2x_1,\,\tau=2x_{-1}$.
	As a result, for Casorati-type solution we have
	\begin{align*}
	\phi_{n+1}^{(i)}(0)&=c_ip_i^{n+1}e^{\xi_i}+d_iq_i^{n+1}e^{\eta_i}\\
	&=c_ip_i^{n+1}e^{\xi_i}\left[1+\mathrm{i}(-1)^{n+1}\left({\frac{1-cp_i}{1-cq_i}}\right)^{\frac{1}{2}}e^{\eta_i-\xi_i}\right]\\
	&=c_ip_i^{n+1}e^{\xi_i}\left[1-\mathrm{i}(-1)^{n}\left({\frac{1-\mathrm{i}\lambda p_i}{1+\mathrm{i}\lambda p_i}}\right)^{\frac{1}{2}}e^{-p_iy-\tau/p_i+\eta_{i0}-\xi_{i0}}\right],
	\end{align*}
	and
	\begin{align*}
	\phi_n^{(i)}(1)&=c_ip_i^{n}(1-cp_i)e^{\xi_i}+d_iq_i^{n}(1-cq_i)e^{\eta_i}\\
	&=c_ip_i^{n}(1-cp_i)e^\xi_i\left[1+\mathrm{i}(-1)^n\left(\frac{1+\mathrm{i}\lambda p_i}{1-\mathrm{i}\lambda p_i}\right)^{\frac{1}{2}}e^{-p_iy-\tau/p_i+\eta_{i0}-\xi_{i0}}\right].
	\end{align*}
	
	Thus we know
	\begin{align}
	\phi^{(i)}_{n+1}(0)\Bumpeq\bar{\phi}^{(i)}_n(1),\,\phi^{(i)}_{n+1}(1)\Bumpeq\bar{\phi}^{(i)}_n(0),
	\end{align}
	which imply the relations
	\begin{align}
	\tau_{n+1,0}\Bumpeq\bar{\tau}_{n,1},\qquad \tau_{n+1,1}\Bumpeq\bar{\tau}_{n,0}.\label{red-1}
	\end{align}
	Relations \eqref{red-1} means both $\tau_{n}(0)$ and $\tau_{n}(1)$ are 2-period sequences. Here $\Bumpeq$ means two $\tau$ functions are equivalent up to a constant multiple and $\bar{\tau}$ denotes complex conjugate of $\tau$. The same result appears for the relations between $\tau$ functions with Gram-type determinant form, which we omit here.
	In this case, the kink and anti-kink solutions are obtained.
	
	$$$$
	{\large\bf Case 2: for complex $p_i$,\,$q_i$}
	
	We impose restrictions on the parameters of $\tau$ functions
	\begin{align}
	p_i=-q_i,\,d_i=\mathrm{i}\left({\frac{1-cp_i}{1-cq_i}}\right)^{\frac{1}{2}}c_i,\,p_{2i-1}=\bar{p}_{2i},\,N=2M,\label{rest-cas2}
	\end{align}
	for Casorati-type solution or
	\begin{align}
	p_i=q_i,\,c_{ij}=\mathrm{i}\left({\frac{1-cp_i}{1+cq_j}}\right)^{\frac{1}{2}}\delta_{ij},\ p_{2i-1}=\bar{p}_{2i},\,N=2M,\label{rest-gram2}
	\end{align}
	for Gram-type solution. By taking $y=2x_1,\,\tau=2x_{-1}$,
	from Case 1, we know that
	\begin{align*}
	&\phi_{n+1}^{(i)}(0)=p_i^{n+1}e^{\xi_i}\left[1-\mathrm{i}(-1)^{n}\sqrt{\frac{1-\mathrm{i}\lambda p_i}{1+\mathrm{i}\lambda p_i}}e^{-p_iy-\tau/p_i+\eta_{i0}-\xi_{i0}}\right],\\
	&\phi_{n}^{(j)}(1)=p_j^{n}(1-cp_j)e^{\xi_j}\left[1+\mathrm{i}(-1)^n\sqrt{\frac{1+\mathrm{i}\lambda p_j}{1-\mathrm{i}\lambda p_j}}e^{-p_jy-\tau/p_j+\eta_{j0}-\xi_{j0}}\right].
	\end{align*}
	Thus we can obtain
	\begin{align}
	\phi_{n+1}^{(2i-1)}(0)\Bumpeq\bar{\phi}_n^{(2i)}(1),\,\phi_{n+1}^{(2i)}(0)\Bumpeq\bar{\phi}_{n}^{(2i-1)}(1),
	\end{align}
	and
	\begin{align}
	\phi_{n+1}^{(2i-1)}(1)\Bumpeq\bar{\phi}_n^{(2i)}(0),\,\phi_{n+1}^{(2i)}(1)\Bumpeq\bar{\phi}_{n}^{(2i-1)}(0),
	\end{align}
	which also correspond to the relation \eqref{red-1}. In this case, the breather solutions are obtained.
	\begin{remark}
		Reductions from \eqref{2dtl}-\eqref{BT2} to the bilinear form of the gsG equation with $\nu=1$ differ from the case with $\nu=-1$ we proposed in \cite{gsg-F}, in which $\tau_{n+1,0}\Bumpeq\bar{\tau}_{n,0},\  \tau_{n+1,1}\Bumpeq\bar{\tau}_{n,1}$. Actually, reductions we introduce here are relatively rare in other equations.
	\end{remark}
	Naturally, we can get kink-breather solutions by mixing Case 1 and Case 2. Moreover, by substitution
	\begin{align}
	\tau_{00}=f,\,\tau_{01}=g,\label{fg1}
	\end{align}
	equations \eqref{2dtl}-\eqref{BT2} can be recast into
	\begin{align}
	&D_\tau D_yf\cdot f=\frac{1}{2}(f^2-\bar{g}^2),\label{Bl1}\\
	&D_\tau D_y g\cdot g=\frac{1}{2}(g^2-\bar{f}^2),\label{Bl2}\\
	&\mathrm{i}\lambda D_yf\cdot \bar{f}-\frac{1}{2}(f\bar{f}-g\bar{g})=0,\label{Bl3}\\
	&\mathrm{i}\lambda D_yg\cdot \bar{g} +\frac{1}{2}(g\bar{g}-f\bar{f})=0,\label{Bl4}\\
	&\mathrm{i}\lambda^{-1}D_\tau f\cdot g+\frac{1}{2}(fg-\bar{f}\bar{g})=0,\label{Bl5}\\
	&\mathrm{i}\lambda^{-1}D_\tau \bar{g}\cdot \bar{f}+\frac{1}{2}(\bar{f}\bar{g}-fg)=0,\label{Bl6}
	\end{align}
	which are nothing but the bilinear equations of the gsG equation \eqref{gsg1}. From equation \eqref{Bl3}-\eqref{Bl6}, we have
	\begin{align}
	&\mathrm{i}\lambda\left(\ln \frac{f}{\bar{f}}\right)_y=\frac{1}{2}-\frac{1}{2}\frac{g\bar{g}}{f\bar{f}},\\
	&\mathrm{i}\lambda\left(\ln \frac{g}{\bar{g}}\right)_y=-\frac{1}{2}+\frac{1}{2}\frac{f\bar{f}}{g\bar{g}},\\
	&\mathrm{i}\lambda^{-1}\left(\ln \frac{f}{g}\right)_\tau=-\frac{1}{2}+\frac{1}{2}\frac{\bar{f}\bar{g}}{fg},\\
	&\mathrm{i}\lambda^{-1}\left(\ln \frac{\bar{f}}{\bar{g}}\right)_\tau=\frac{1}{2}-\frac{1}{2}\frac{fg}{\bar{f}\bar{g}},
	\end{align}
	that lead to
	\begin{align}
	&x_y=\frac{1}{\lambda r}=\lambda^{-1}\cosh\varphi=\frac{1}{2\lambda}\left(\frac{g\bar{g}}{f\bar{f}}+\frac{f\bar{f}}{g\bar{g}}\right)=\lambda^{-1}+\mathrm{i}\left(\ln\frac{\bar{f}g}{f\bar{g}}\right)_y,\\
	&x_\tau=-\lambda^2rx_y\cos u =-\lambda\cos u=-\frac{\lambda}{2}\left(\frac{\bar{f}\bar{g}}{fg}+\frac{fg}{\bar{f}\bar{g}}\right)=-\lambda+\mathrm{i}\left(\ln\frac{\bar{f}g}{f\bar{g}}\right)_\tau.
	\end{align}
	Thus we obtain the expression for $x$ by tau functions
	\begin{align}
	x=\lambda^{-1}y-\lambda\tau+\mathrm{i}\ln\frac{\bar{f}g}{f\bar{g}},
	\end{align}
	Summarizing the above results, the determinant ($N$-soliton) solution of the gsG equation is given by the following theorem.
	\begin{theorem}\label{det-sol}
		The parametric form for the soliton and breather solution of the gsG equation \eqref{gsg1} is
		\begin{align}
		u=\mathrm{i}\ln\frac{\bar{f}\bar{g}}{fg},
		\end{align}
		\begin{align}
		x=\lambda^{-1}y-\lambda\tau+\mathrm{i}\ln\frac{\bar{f}g}{f\bar{g}},\ t=\lambda\tau,\ \lambda\in\mathbb{R},
		\end{align}
		where $f$ and $g$ are the determinants given by \eqref{fg1} and $\tau_{n,m}$ can be written either as a Casorati-type determinant
		\begin{align}
		\tau_{n,m}& =\left|\begin{array}{cccc}
		\phi_n^{(1)}(m) & \phi_{n+1}^{(1)}(m) & \cdots & \phi_{n+N-1}^{(1)}(m) \\
		\phi_n^{(2)}(m) & \phi_{n+1}^{(2)}(m) & \cdots & \phi_{n+N-1}^{(2)}(m) \\
		\cdots & \cdots & \cdots & \cdots \\
		\phi_n^{(N)}(m) & \phi_{n+1}^{(N)}(m) & \cdots & \phi_{n+N-1}^{(N)}(m)
		\end{array}\right|,
		\end{align}
		where
		\begin{align}
		\phi_n^{(i)}(m)=p_i^n\left(1-\mathrm{i}\lambda p_i\right)^m e^{\frac{p_i}{2}y+\frac{1}{2p_i}\tau+\xi_{i0}}+\mathrm{i}\sqrt{\frac{1-\mathrm{i}\lambda p_i}{1+\mathrm{i}\lambda p_i}}(-p_i)^n\left(1+\mathrm{i}\lambda p_i\right)^me^{-\frac{p_i}{2}y-\frac{1}{2p_i}\tau+\eta_{i0}},\quad i=1,2,\cdots,N,
		\end{align}
		\begin{align}
		&p_{2i-1}=\bar{p}_{2i},\quad i=1,2,\cdots,M,\\
		&p_i\in \mathbb{R},\,i=2M+1,2M+2,\cdots,N,
		\end{align}
		or a Gram-type determinant
		\begin{align}
		\tau_{n,m}=\left|m_{ij}^{n,m}\right|_{N\times N}=\left|\mathrm{i}\left({\frac{1-\mathrm{i}\lambda p_i}{1+\mathrm{i}\lambda p_j}}\right)^{\frac{1}{2}}\delta_{ij}+\frac{1}{p_i+p_j}\left(-\frac{p_i}{p_j}\right)^n\left(
		\frac{1-\mathrm{i}\lambda p_i}{1+\mathrm{i}\lambda p_j}
		\right)^me^{\xi_i+\eta_j} \right|_{N\times N},
		\end{align}
		where
		\begin{align}
		&\xi_i=\frac{p_i}{2}y+\frac{1}{2p_i}\tau+\xi_{i0}, \quad \eta_j=\frac{p_j}{2}y+\frac{1}{2p_j}\tau+\eta_{j 0},\\
		&p_{2i-1}=\bar{p}_{2i},\quad i=1,2,\cdots,M,\\
		&p_i\in \mathbb{R},\,i=2M+1,2M+2,\cdots,N.
		\end{align}
	\end{theorem}
	\begin{remark}\label{nu-1sol}
		Similar to the deduction of the gsG equation with $\nu=1$, one can generalize the Theorem 1 in \cite{gsg-F}. The parametric form for the $N$-soliton solution of the gsG equation \eqref{gsg} with $\nu=-1$ is
		\begin{align}
		u=\mathrm{i}\ln\frac{\bar{f}\bar{g}}{fg},
		\end{align}
		\begin{align}
		x=c^{-1}y+c\tau+\ln\frac{\bar{g}g}{\bar{f}f},\ t=c\tau,\ c\in\mathbb{R},
		\end{align}
		\begin{align}
		f=\tau_{00},\ g=\tau_{01},
		\end{align}
		and $\tau_{n,m}$ can be written as a Casorati-type determinant
		\begin{align}
		\tau_{n,m}& =\left|\begin{array}{cccc}
		\phi_n^{(1)}(m) & \phi_{n+1}^{(1)}(m) & \cdots & \phi_{n+N-1}^{(1)}(m) \\
		\phi_n^{(2)}(m) & \phi_{n+1}^{(2)}(m) & \cdots & \phi_{n+N-1}^{(2)}(m) \\
		\cdots & \cdots & \cdots & \cdots \\
		\phi_n^{(N)}(m) & \phi_{n+1}^{(N)}(m) & \cdots & \phi_{n+N-1}^{(N)}(m)
		\end{array}\right|,
		\end{align}
		with
		\begin{align}
		\phi_n^{(i)}(m)=p_i^n\left(1-cp_i\right)^m e^{\frac{p_i}{2}y+\frac{1}{2p_i}\tau+\xi_{i0}}+\mathrm{i}(-p_i)^n\left(1+cp_i\right)^me^{-\frac{p_i}{2}y-\frac{1}{2p_i}\tau+\eta_{i0}},\quad i=1,2,\cdots,N.
		\end{align}
		Obviously, the result in this remark is equivalent to that in Theorem 1 of \cite{gsg-F} when $c=1$.
	\end{remark}
	\begin{remark}
		The determinant solutions we obtained above are consistent with the solutions given in \cite{gsg1}.
	\end{remark}
	
	\subsection{Reduction to the sG and SP equation}
	
	In \cite{gsg-1} and \cite{gsg1}, Matsuno demonstrated that the gsG equation is reduced to the sG equation in the long wave limit and to the SP equation in the short wave limit. Here we introduce the scaling parameter $c$ (or $\lambda$) in the hodograph transformation and the $\tau$-function, and give another kind of reduction.
	
	\subsubsection{Reduction to the sine-Gordon equation}
	
	The sG equation
	\begin{align}
	u_{xt}=\sin u,\ \quad u \equiv u(x, t) \in \mathbb{R},\  \quad(x, t) \in \mathbb{R}^2,
	\end{align}
	is a fundamental model in the integrable system, which appears in a number of disciplines of physics including magnetic flux propagation \cite{sg-phy1,sg-phy2}, one-dimensional classical field theory \cite{sg-phy3,sg-phy4}, and nonlinear optics \cite{sg-phy5}. In this part, we take into account the reductions from the gsG equation with $\nu=\pm1$ to the sG equation in the continuous case through some scaling transformations.
	\\
	\\
	{\bf (I) From the gsG equation with $\nu=-1$ to the sG equation}. In this part of reduction, we take $c\in\mathbb{R}$ as a small parameter. The matrix elements of the $\tau$-function for the gsG equation we proposed in \cite{gsg-F} can be written as
	\begin{align}
	\phi_n^{(i)}(1)&=p_i^n(1-cp_i)e^{\frac{p_i}{2}y+\frac{1}{2p_i}\tau+\xi_{i0}}+\mathrm{i}(-p_i)^n\left(1+cp_i\right)e^{-\frac{p_i}{2}y-\frac{1}{2p_i}\tau+\eta_{i0}}\nonumber\\
	&= \phi_n^{(i)}(0)-c\phi_{n+1}^{(i)}(0),
	\end{align}
	which means
	\begin{align}
	g=\tau_0(1)=f+\mathnormal{O}(c),\ \bar{g}=\tau_1(1)=\bar{f}+\mathnormal{O}(c).
	\end{align}
	Then we can rewrite the dependent variable transformations and introduce the scaling transformation
	\begin{align}
	&u=\mathrm{i}\ln\frac{\bar{f}\bar{g}}{fg}=2\mathrm{i}\ln\frac{\bar{f}}{f}+\mathnormal{O}(c),\label{con-u-sg}\\
	&\phi=\mathrm{i}\ln\frac{\bar{f}g}{f\bar{g}}=\mathnormal{O}(c),\label{con-phi-sg}\\
	&\hat{x}=cx=y+c^2\tau+c\ln\frac{g\bar{g}}{f\bar{f}}=y+\mathnormal{O}(c^2),\label{con-x-sg}\\
	&\hat{t}=c^{-1}t=\tau.\label{con-t-sg}
	\end{align}
	In addition, the gsG equation \eqref{gsg} with $\nu=-1$ can be recast into
	\begin{align}
	u_{\hat{x}\hat{t}}=(1-c^2\partial_{\hat{x}}^2)\sin u.\label{c-sg}
	\end{align}
	With the scaling limit $c\rightarrow 0$, equation \eqref{c-sg} becomes
	\begin{align}
	u_{\hat{x}\hat{t}}=\sin u,
	\end{align}
	which is the well-known sG equation. And the dependent variable transformation, as well as the $\tau$-function $f$, also reduces to the usual form of the $N$-soliton
	solutions of the sG equation \cite{sg1,sg2,sg3}.
	
	It should be point out that the transformation \eqref{con-u-sg}-\eqref{con-t-sg} are agree with the scaled variables introducing by Matsuno in \cite{gsg-1}. The difference is that we introduce the parameter $c$ in the hodograph transformation, thus the $\tau$-function can be transformed naturally without the scaling of variables $y,\ \tau$ and the parameter $p$.
	\\
	\\
	{\bf (II) From the gsG equation with $\nu=1$ to the sG equation}. In this part of reduction, we take $\lambda\in\mathbb{R}$ as a small parameter.
	Similar to the case $\nu=-1$, recall that
	\begin{align}
	\phi_n^{(i)}(1)&=p_i^n\left(1-\mathrm{i}\lambda p_i\right) e^{\frac{p_i}{2}y+\frac{1}{2p_i}\tau+\xi_{i0}}+\mathrm{i}\sqrt{\frac{1-\mathrm{i}\lambda p_i}{1+\mathrm{i}\lambda p_i}}(-p_i)^n\left(1+\mathrm{i}\lambda p_i\right)e^{-\frac{p_i}{2}y-\frac{1}{2p_i}\tau+\eta_{i0}},\nonumber\\
	&=p_i^ne^{\frac{p_i}{2}y+\frac{1}{2p_i}\tau+\xi_{i0}}+\mathrm{i}(-p_i)^ne^{-\frac{p_i}{2}y-\frac{1}{2p_i}\tau+\eta_{i0}}+\mathnormal{O}(\lambda),\nonumber\\
	&=\hat{\phi}_n^{(i)}(0)+\mathnormal{O}(\lambda),\\
	\phi_n^{(i)}(0)&=p_i^ne^{\frac{p_i}{2}y+\frac{1}{2p_i}\tau+\xi_{i0}}+\mathrm{i}\sqrt{\frac{1-\mathrm{i}\lambda p_i}{1+\mathrm{i}\lambda p_i}}(-p_i)^ne^{-\frac{p_i}{2}y-\frac{1}{2p_i}\tau+\eta_{i0}},\nonumber\\
	&=\hat{\phi}_n^{(i)}(0)+\mathnormal{O}(\lambda),
	\end{align}
	which lead to
	\begin{align}
	f=\hat{f}+\mathnormal{O}(\lambda),\ g=\hat{f}+\mathnormal{O}(\lambda),\ \bar{f}=\bar{\hat{f}}+\mathnormal{O}(\lambda),\
	\bar{g}=\bar{\hat{f}}+\mathnormal{O}(\lambda)
	\end{align}
	here $\hat{f}$ is the $\tau$-function of the sG equation. Thus we have
	\begin{align}
	u=2\mathrm{i}\ln\frac{\bar{\hat{f}}}{\hat{f}}+\mathnormal{O}(\lambda),\ \varphi=\mathnormal{O}(\lambda),\ \hat{x}=\lambda x=y+\mathnormal{O}(\lambda^2),\ \hat{t}=\lambda^{-1}t=\tau,
	\end{align}
	And the gsG equation \eqref{gsg} with $\nu=1$ becomes
	\begin{align}
	u_{\hat{x}\hat{t}}=(1+\lambda^2\partial_{\hat{x}}^2)\sin u.
	\end{align}
	The gsG equation is converted into the sG equation with the scaling limit $\lambda\rightarrow 0$, as well as its solutions. Here those transformations are agree with transformations introduced in \cite{gsg1}.
	
	\subsubsection{Reduction to the short pulse equation}
	The short pulse (SP) equation
	\begin{align}
	u_{x t}=u-\frac{\sigma}{6}\left(u^3\right)_{x x},\label{sp}
	\end{align}
	was derived to describe the propagation of ultra-short optical pulses in nonlinear media by Sch{\"a}fer and Wayne when $\sigma=-1$ \cite{sp1}. Here, the real-valued function $u=u(x,t)$ represent the magnitude of the electric field, and the subscripts $t$ and $x$ signify partial differentiation. The SP equation has also been developed as an integrable differential equation linked to pseudospherical surfaces outside of the context of nonlinear optics \cite{sp2}. When $\sigma=1$, equation \eqref{sp}
	\begin{align}
	u_{x t}=u-\frac{1}{6}\left(u^3\right)_{x x},\label{sp1}
	\end{align}
	was shown to model the evolution of ultra-short pulses in the band gap of nonlinear metamaterials \cite{1sp}.
	Here, we show that applying the right scaling limit and variable transformations results in the gsG equation with $\nu=\pm1$ being reduced to the SP equation \eqref{sp} with $\sigma=\pm1$ in the continuous case.
	\\
	\\
	{\bf (I) From the gsG equation with $\nu=-1$ to the SP equation with $\sigma=-1$.} In this part of reduction, we take $c$ as a big parameter (or $\epsilon=\frac{1}{c}$ small). The matrix elements of the $\tau$-function for the gsG equation with $\nu=-1$ in \cite{gsg-F} are
	\begin{align}
	\phi_n^{(i)}(1)&=p_i^n(1-cp_i)e^{\frac{p_i}{2}y+\frac{1}{2p_i}\tau+\xi_{i0}}+\mathrm{i}(-p_i)^n\left(1+cp_i\right)e^{-\frac{p_i}{2}y-\frac{1}{2p_i}\tau+\eta_{i0}}\nonumber\\
	&\propto\phi_{n+1}^{(i)}(0)-2\epsilon \frac{\partial}{\partial \tau}\phi_{n+1}^{(i)}(0),
	\end{align}
	from which we can obtain
	\begin{align}
	&g=\tau_0(1)\propto\bar{f}-2\epsilon \frac{\partial}{\partial \tau}\bar{f}+\mathnormal{O}(\epsilon^2),\\
	&\bar{g}\propto {f}-2\epsilon \frac{\partial}{\partial \tau} {f}+\mathnormal{O}(\epsilon^2).
	\end{align}
	Then we rewrite the dependent variable transformations and introduce new variable as
	\begin{align}
	&\epsilon \hat{u}=u=\mathrm{i}\ln\frac{\bar{f}\bar{g}}{{f}{g}}=\mathrm{i}\ln \frac{\bar{f}(f-2\epsilon\frac{\partial}{\partial \tau}{f})}{{f}(\bar{f}-2\epsilon\frac{\partial}{\partial \tau}\bar{f})}+\mathnormal{O}(\epsilon^2),\label{con-u-sp}\\
	&\phi=\mathrm{i}\ln\frac{\bar{f}{g}}{{f}\bar{g}}=\mathrm{i}\ln \frac{\bar{f}(\bar{f}-2\epsilon\frac{\partial}{\partial \tau}\bar{f})}{{f}(f-2\epsilon\frac{\partial}{\partial \tau}{f})}+\mathnormal{O}(\epsilon^2),\label{con-phi-sp}\\
	&\epsilon \hat{x}=x-t=\epsilon y+\ln\frac{g\bar{g}}{f\bar{f}}=\epsilon y+\ln\frac{(\bar{f}-2\epsilon\frac{\partial}{\partial \tau}\bar{f})(f-2\epsilon\frac{\partial}{\partial \tau}{f})}{f\bar{f}}+\mathnormal{O}(\epsilon^2),\label{con-x-sp}\\
	&\hat{t}=\epsilon t=\tau.\label{con-t-sp}
	\end{align}
	The gsG equation \eqref{gsg} with $\nu=-1$ can be recast into
	\begin{align}
	\epsilon \hat{u}_{\hat{t} \hat{x}}=\epsilon\left(\hat{u}+\frac{1}{6}\left(\hat{u}^3\right)_{\hat{x} \hat{x}}\right)+\mathnormal{O}\left(\epsilon^3\right).\label{con-sp}
	\end{align}
	Dividing both sides of \eqref{con-sp} by $\epsilon$ and taking $\epsilon\rightarrow0$ in \eqref{con-u-sp}-\eqref{con-sp}, we arrive at
	\begin{align}
	\hat{u}=2\mathrm{i}\left(\ln\frac{\bar{f}}{f}\right)_\tau,\ \phi=2\mathrm{i}\ln\frac{\bar{f}}{f}&,\ \hat{x}=y-2(\ln(f\bar{f}))_\tau,\ \hat{t}=\tau,\\
	\hat{u}_{\hat{t} \hat{x}}=\hat{u}&+\frac{1}{6}\left(\hat{u}^3\right)_{\hat{x} \hat{x}}.
	\end{align}
	Here the SP equation \eqref{sp} is derived and its parametric representation $\tau$-function are equivalent to those in \cite{dis-sp}.
	\\
	\\
	{\bf (II) From the gsG equation with $\nu=1$ to the SP equation with $\sigma=1$.} Similar to the case with $\nu=-1$, we take $\epsilon=\frac{1}{\lambda}$ as a small parameter. The $\tau$-function in Theorem \ref{det-sol} can be written as
	\begin{align}
	\phi_n^{(i)}(0)&=p_i^n e^{\frac{p_i}{2}y+\frac{1}{2p_i}\tau+\xi_{i0}}+\mathrm{i}\sqrt{\frac{1-\mathrm{i}\lambda p_i}{1+\mathrm{i}\lambda p_i}}(-p_i)^ne^{-\frac{p_i}{2}y-\frac{1}{2p_i}\tau+\eta_{i0}},\nonumber\\
	&\propto\hat{\phi}_{n+1}^{(i)}-{\mathrm{i}\epsilon}\frac{\partial}{\partial \tau}\hat{\phi}_{n+1}^{(i)}+\mathnormal{O}(\epsilon^2),
	\end{align}
	\begin{align}
	\phi_n^{(i)}(1)&=p_i^n\left(1-\mathrm{i}\lambda p_i\right) e^{\frac{p_i}{2}y+\frac{1}{2p_i}\tau+\xi_{i0}}+\mathrm{i}\sqrt{\frac{1-\mathrm{i}\lambda p_i}{1+\mathrm{i}\lambda p_i}}(-p_i)^n\left(1+\mathrm{i}\lambda p_i\right)e^{-\frac{p_i}{2}y-\frac{1}{2p_i}\tau+\eta_{i0}},\nonumber\\
	&\propto \hat{\phi}_{n}^{(i)}-{\mathrm{i}\epsilon}\frac{\partial}{\partial \tau}\hat{\phi}_{n}^{(i)}+\mathnormal{O}(\epsilon^2),
	\end{align}
	if we define
	\begin{align}
	\hat{\phi}_{n}^{(i)}=p_i^{n} e^{\frac{p_i}{2}y+\frac{1}{2p_i}\tau+\xi_{i0}}+(-p_i)^{n}e^{-\frac{p_i}{2}y-\frac{1}{2p_i}\tau+\eta_{i0}},
	\end{align}
	and
	\begin{align}
	\hat{f} =\left|\begin{array}{cccc}
	\hat{\phi}_{n+1}^{(1)} & \hat{\phi}_{n+2}^{(1)} & \cdots & \hat{\phi}_{n+N}^{(1)} \\
	\hat{\phi}_{n+1}^{(2)} & \hat{\phi}_{n+2}^{(2)} & \cdots & \hat{\phi}_{n+N}^{(2)} \\
	\cdots & \cdots & \cdots & \cdots \\
	\hat{\phi}_{n+1}^{(N)} & \hat{\phi}_{n+2}^{(N)} & \cdots & \hat{\phi}_{n+N}^{(N)}
	\end{array}\right|,\ \hat{g}=\left|\begin{array}{cccc}
	\hat{\phi}_n^{(1)} & \hat{\phi}_{n+1}^{(1)} & \cdots & \hat{\phi}_{n+N-1}^{(1)} \\
	\hat{\phi}_n^{(2)} & \hat{\phi}_{n+1}^{(2)} & \cdots & \hat{\phi}_{n+N-1}^{(2)} \\
	\cdots & \cdots & \cdots & \cdots \\
	\hat{\phi}_n^{(N)} & \hat{\phi}_{n+1}^{(N)} & \cdots & \hat{\phi}_{n+N-1}^{(N)}
	\end{array}\right|.
	\end{align}
	Furthermore, one obtains
	\begin{align}
	&f=\hat{f}-{\mathrm{i}\epsilon}\frac{\partial}{\partial \tau}\hat{f}+\mathnormal{O}(\epsilon^2),\ \bar{f}=\hat{f}+{\mathrm{i}\epsilon}\frac{\partial}{\partial \tau}\hat{f}+\mathnormal{O}(\epsilon^2),\\ &g=\hat{g}+{\mathrm{i}\epsilon}\frac{\partial}{\partial \tau}\hat{g}+\mathnormal{O}(\epsilon^2),\ \bar{g}=\hat{g}-{\mathrm{i}\epsilon}\frac{\partial}{\partial \tau}\hat{g}+\mathnormal{O}(\epsilon^2).
	\end{align}
	Here we introduce new variables as
	\begin{align}
	&\epsilon \hat{u}=u=\mathrm{i}\ln\frac{\bar{f}\bar{g}}{{f}{g}}=\mathrm{i}\ln \frac{(\hat{f}+{\mathrm{i}\epsilon}\hat{f}_\tau)(\hat{g}-{\mathrm{i}\epsilon}\hat{g}_\tau)}{(\hat{f}-{\mathrm{i}\epsilon}\hat{f}_\tau)(\hat{g}+{\mathrm{i}\epsilon}\hat{g}_\tau)}+\mathnormal{O}(\epsilon^2),\label{con-1u-sp}\\
	&\varphi=\ln\frac{\bar{g}g}{f\bar{f}}=\ln\frac{(\hat{g}-{\mathrm{i}\epsilon}\hat{g}_\tau)(\hat{g}+{\mathrm{i}\epsilon}\hat{g}_\tau)}{(\hat{f}-{\mathrm{i}\epsilon}\hat{f}_\tau)(\hat{f}+{\mathrm{i}\epsilon}\hat{f}_\tau)}+\mathnormal{O}(\epsilon^2),\label{con-1phi-sp}\\
	&{\epsilon \hat{x}=x+t=\epsilon y+\mathrm{i}\ln\frac{(\hat{f}+{\mathrm{i}\epsilon}\hat{f}_\tau)(\hat{g}+{\mathrm{i}\epsilon}\hat{g}_\tau)}{(\hat{f}-{\mathrm{i}\epsilon}{f}_\tau)(\hat{g}-{\mathrm{i}\epsilon}\hat{g}_\tau)}+\mathnormal{O}(\epsilon^2)},\label{con-1x-sp}\\
	&\hat{t}=\tau,\label{con-1t-sp}
	\end{align}
	then the scaling limit $\epsilon\rightarrow 0$ leads to
	\begin{align}
	\hat{u}=2\left(\ln\frac{\hat{g}}{\hat{f}}\right)_\tau,\
	\varphi=2\ln\frac{\hat{g}}{\hat{f}},\ \hat{x}=y-2(\ln\hat{f}\hat{g})_\tau,\ \hat{t}=\tau,
	\end{align}
	and
	\begin{align}
	\hat{u}_{\hat{t}\hat{x}}= \hat{u}-\frac{1}{6}(\hat{u}^3)_{\hat{x}\hat{x}}.
	\end{align}
	Note that the $N$-soliton solutions of the SP equation with $\sigma=1$ exhibits the singular nature since $u$ diverges when $|x|\rightarrow \infty$ as shown in \cite{gsg1}.	
	\section{Integrable fully discretization of the gsG equation}\label{sec3}
	To construct a fully discrete analogue of the gsG equation, we introduce two discrete variables, $k$ and $l$, which correspond to the discrete spartial and time variables, respectively. We start with the following fully discrete bilinear equations.
	
	\begin{align}
	&\tau_n(k,l+1,m)\tau_n(k,l,m-1)-bc\tau_{n+1}(k,l,m-1)\tau_{n-1}(k,l+1,m)\notag\\
	&=(1-bc)\tau_n(k,l,m)\tau_n(k,l+1,m-1),\label{dis-2dtl2}
	\end{align}
	\begin{align}
			& \tau_{n+1}(k,l,m-1)\tau_{n}(k+1,l,m)-ac^{-1}\tau_{n+1}(k+1,l,m)\tau_{n}(k,l,m-1)\notag\\
			&=(1-ac^{-1})\tau_{n}(k,l,m)\tau_{n+1}(k+1,l,m-1).
			\label{dis-2dtl3}
			\end{align}		
	Here $n$, $k$, $l$, $m$ are integers, $a$, $b$, $c$ are parameters.
	\begin{prop}\label{fd-2dtl-sol}
		The bilinear equations \eqref{dis-2dtl2}-\eqref{dis-2dtl3} admit either the following Gram-type determinant solutions
		$$
		\tau_n(k, l,m)=\left|m_{ij}^{n}(k, l,m)\right|_{1 \leqslant i, j \leqslant N},
		$$
		with
		\begin{align}
		m_{ij}^{(n)}(k, l,m)&=c_{i j}+\frac{d_{ij}}{p_i+q_j}\left(-\frac{p_i}{q_j}\right)^{n}\left(\frac{1-a p_i}{1+a q_j}\right)^{-k}\left(\frac{1-b p_i^{-1}}{1+b q_j^{-1}}\right)^{-l}\left(\frac{1-c p_i}{1+c q_j}\right)^{m},
		\end{align}
		or the Casorati-type determinant solutions
		\begin{align}
		\tau_n(k, l,m)=\left|\phi_{(n+j-1)}^{(i)}(k, l,m)\right|_{1 \leqslant i, j \leqslant N},
		\end{align}
		with
		\begin{align*}
		\phi_{(n)}^{(i)}(k, l,m)=c_ip_i^n\left(1-a p_i\right)^{-k}\left({1-b p_i^{-1}}\right)^{-l}\left(1-cp_i\right)^m +d_iq_i^n\left(1-a q_i\right)^{-k} \left({1+b q_i^{-1}}\right)^{-l}\left(1-cq_i\right)^m.
		\end{align*}
	\end{prop}

		\begin{proof}
			Here we give a proof for the Gram-type determinant solution. The proof for the Casorati-type determinant solution is similar.
			The discrete Kadomtsev-Petviashvili (dKP) equation was proposed independently by Hirota \cite{Hirota-1981} and Miwa \cite{Miwa-1982}  in early 1980s so it is also called Hirota-Miwa (HW) equation.  Discrete KP hierarchy is an infinite number of bilinear equations with $(k_i, k_j, k_m)$ taken from $(k_1, k_2, k_3, \cdots)$, among which, let us choose two triples: $(k_1, k_3, k_4)$, $(k_1, k_2, k_4)$
			so that the following two bilinear equations follow
			
			\begin{eqnarray}
			&& (a^{-1}_{3}-a^{-1}_{4})\tau (k_{1}+1,k_{2},k_{3}, k_4)\tau (k_{1},k_{2},k_{3}+1,k_4+1) \nonumber
			\\
			&&+(a^{-1}_{4}-a^{-1}_{1})\tau (k_{1},k_{2},k_{3}+1,k_4)\tau (k_{1}+1,k_{2},k_{3},k_4+1) \nonumber
			\\
			&&+(a^{-1}_{1}-a^{-1}_{3}) \tau (k_{1},k_{2},k_{3},k_4+1)\tau (k_{1}+1,k_{2},k_{3}+1,k_4)=0\,,
			\label{dKP1}
			\end{eqnarray}	
			\begin{eqnarray}
			&& (a^{-1}_{2}-a^{-1}_{4})\tau (k_{1}+1,k_{2},k_{3}, k_4)\tau (k_{1},k_{2}+1,k_{3},k_4+1) \nonumber
			\\
			&&+(a^{-1}_{4}-a^{-1}_{1})\tau (k_{1},k_{2}+1,k_{3},k_4)\tau (k_{1}+1,k_{2},k_{3},k_4+1) \nonumber
			\\
			&&+(a^{-1}_{1}-a^{-1}_{2}) \tau (k_{1},k_{2},k_{3},k_4+1)\tau (k_{1}+1,k_{2}+1,k_{3},k_4)=0\,.
			\label{dKP2}
			\end{eqnarray}	
			As is shown by Ohta {\it et al.} \cite{OHTI_JPSJ},  above two bilinear equations admit Gram type solution
			\begin{align}
			\tau(k_1, \cdots, k_4)=\left|c_{i j}+\frac{1}{p_i+q_j} \prod_{l=1}^4 \left(\frac{1-a_l p_i}{1+a_l q_j}\right)^{-k_l} \right|_{1 \leqslant i, j \leqslant N}.
			\end{align}
			By taking $a_1 \to \infty$, and redefined $a_2=a$, $a_3=b^{-1}$, $a_4=c$, $k_2=k$, $k_3=l, k_4=-m$, $k_1+k_3=-n-1$ and $\tau({k_1, \cdots, k_4})=\tau_n(k,l,m)$,  we have the bilinear equation \eqref{dis-2dtl2} together with its Gram-type determinant solution
			from \eqref{dKP1}.
			
			On the other hand, by taking the same limit $a_1 \to \infty$, and redefined $a_2=a$, $a_3=b^{-1}$, $a_4=c$, $k_2=k, k_3=l$, $k_4=-m$, $k_1=-n-1$ and $\tau({k_1, \cdots, k_4})=\tau_n(k,l,m)$,  we have the bilinear equation \eqref{dis-2dtl3} together with its Gram-type determinant solution
			from \eqref{dKP2}. The proof is complete.
		\end{proof}
		\begin{remark}
			Eq. \eqref{dis-2dtl2} is actually the fully discrete 2DTL equation, while  eq. \eqref{dis-2dtl3}
			is discrete modified KP equation. As shown from above proof, they are equivalent to discrete KP equation via reparameterization. As shown in this section, the discrete analog of the gsG equation is constructed from the combination of  \eqref{dis-2dtl2} and\eqref{dis-2dtl3}.
		\end{remark}

	By applying a 2-reduction condition: $q_i=p_i$ for Gram-type determinant solution, or $q_i=-p_i$ for Casorati-type, we have $\tau_n \Bumpeq \tau_{n+2}$. Here $\Bumpeq$ means two $\tau$-functions are equivalent up to a constant multiple. In addition, by defining
	\begin{align}
	f_k^l=\tau_0(k,l,0),\,\tilde{f}_k^l=\tau_1(k,l,0),\,g_k^l=\tau_0(k,l,1),\,\tilde{g}_k^l=\tau_1(k,l,1),\label{dis-tau-def}
	\end{align}
	we can obtain the following equations
	\begin{align}
	f_k^lg_k^{l+1}-bc\tilde{f}_k^l\tilde{g}_k^{l+1}&=(1-bc)f_k^{l+1}g_k^l,\label{fulldis5}\\
	\tilde{f}_k^l\tilde{g}_k^{l+1}-bc{f}_k^l{g}_k^{l+1}&=(1-bc)\tilde{f}_k^{l+1}\tilde{g}_k^l,\label{fulldis6}\\
	f_k^l\tilde{g}_{k+1}^l-ac^{-1}\tilde{f}_k^lg_{k+1}^l&=(1-ac^{-1})f_{k+1}^l\tilde{g}_k^l,\label{fulldis7}\\
	\tilde{f}_k^l{g}_{k+1}^l-ac^{-1}{f}_k^l\tilde{g}_{k+1}^l&=(1-ac^{-1})\tilde{f}_{k+1}^l{g}_k^l.\label{fulldis8}
	\end{align}
	Introducing four intermediate variable transformations
	\begin{align}
	\sigma_k^l=2\mathrm{i}\ln\frac{\tilde{f}_k^l}{f_k^l},\,\tilde{\sigma}_k^{l}=2\mathrm{i}\ln\frac{\tilde{g}_k^l}{g_k^l},
	\end{align}
	\begin{align}
	\theta_k^l=\ln\frac{f_k^l}{g_k^l},\ \tilde{\theta}_k^l=\ln\frac{\tilde{f}_k^l}{\tilde{g}_k^l},
	\end{align}
	and then dividing \eqref{fulldis5} and \eqref{fulldis6} by $\tilde{f}_k^l\tilde{g}_k^{l+1}$ and $\tilde{f}_k^{l+1}\tilde{g}_k^l$, respectively, lead to
	\begin{align}
	&e^{\frac{\mathrm{i}}{2}(\sigma_k^l+\tilde{\sigma}_k^{l+1})}-bc=(1-bc)\frac{\tilde{f}_k^{l+1}\tilde{g}_k^l}{\tilde{f}_k^l\tilde{g}_k^{l+1}}e^{\frac{\mathrm{i}}{2}(\sigma_k^{l+1}+\tilde{\sigma}_k^l)},\\
	&\frac{\tilde{f}_k^l\tilde{g}_k^{l+1}}{\tilde{f}_k^{l+1}\tilde{g}_k^l}\left(1-bce^{\frac{\mathrm{i}}{2}(\sigma_k^l+\tilde{\sigma}_k^{l+1})}\right)=1-bc.
	\end{align}
	
	Eliminating $\frac{\tilde{f}_k^{l+1}\tilde{g}_k^l}{\tilde{f}_k^l\tilde{g}_k^{l+1}}$, we get
	\begin{align}
	\frac{1}{bc}\sin\frac{\sigma_{k}^{l+1}-\tilde{\sigma}_{k}^{l+1}-\sigma_k^l+\tilde{\sigma}_k^l}{4}=\sin\frac{\sigma_{k}^{l+1}+\tilde{\sigma}_{k}^{l+1}+\sigma_k^l+\tilde{\sigma}_k^l}{4}.\label{dis-gsg2}
	\end{align}
	
	Meanwhile, if we dividing \eqref{fulldis5} and \eqref{fulldis6} by ${f}_k^l{g}_k^{{l+1}}$ and $\tilde{f}_k^{l}\tilde{g}_k^{l+1}$, respectively, we know that
	\begin{align}
	&1-bc\frac{\tilde{f}_k^{l}\tilde{g}_k^{l+1}}{f_k^{l}g_k^{l+1}}=(1-bc)e^{\theta_k^{l+1}-\theta_k^l},\\
	&1-bc\frac{f_k^lg_k^{l+1}}{\tilde{f}_k^{l}\tilde{g}_k^{l+1}}=(1-bc)e^{\tilde{\theta}_k^{l+1}-\tilde{\theta}_k^l},
	\end{align}
	which can be transformed into
	\begin{align}
	\cosh\frac{\theta_k^{l+1}-\theta_k^l+\tilde{\theta}_k^{l+1}-\tilde{\theta}_k^l}{2}-bc\sinh\frac{\theta_k^{l+1}-\theta_k^l+\tilde{\theta}_k^{l+1}-\tilde{\theta}_k^l}{2}=\cosh\frac{\theta_k^{l+1}-\theta_k^{l}+\tilde{\theta}_k^{l}-\tilde{\theta}_k^{l+1}}{2}.\label{dis-gsg3}
	\end{align}
	Similarly, from \eqref{fulldis7} and \eqref{fulldis8}, we know that
	\begin{align}
	&ac^{-1}\sin\frac{\sigma_{k+1}^l-\tilde{\sigma}_{k+1}^l+\sigma_{k}^l-\tilde{\sigma}_{k}^l}{4}=\sin\frac{\sigma_{k+1}^l+\tilde{\sigma}_{k+1}^l-\sigma_{k}^l-\tilde{\sigma}_{k}^l}{4},\label{dis-gsg4}\\
	&ac^{-1}\cosh\frac{\theta_{k+1}^{l}-\theta_k^l+\tilde{\theta}_{k+1}^{l}-\tilde{\theta}_k^l}{2}-\sinh\frac{\theta_{k+1}^{l}-\theta_k^l+\tilde{\theta}_{k+1}^{l}-\tilde{\theta}_k^l}{2}=ac^{-1}\cosh\frac{\theta_{k+1}^{l}+\theta_k^{l}-\tilde{\theta}_k^{l}-\tilde{\theta}_{k+1}^{l}}{2}.\label{dis-gsg5}
	\end{align}
	\subsection{Fully discretization of the gsG
		equation with $\nu=-1$}
	
	By choosing particular values in phase constants
	\begin{align}
	c_{ij}=\mathrm{i}\delta_{ij},\label{fd-res-1}
	\end{align}
	for Gram-type solution, or
	\begin{align}
	d_i=\mathrm{i}c_i,\label{fd-res-2}
	\end{align}
	for Casorati-type solution, we can make $\tau_n$ and $\tau_{n+1}$ complex
	conjugate to each other, which means
	\begin{align}
	\tilde{f}_{k}^l\Bumpeq\bar{f}_k^l,\ \tilde{g}_{k}^l\Bumpeq\bar{g}_k^l.
	\end{align}
	Here $\bar{f}$ denotes complex conjugate of $f$. Then, similar to the continuous case, we introduce discrete hodograph transformation and dependent variable transformation
	\begin{align}
	u_k^l=\frac{1}{2}(\sigma_k^l+\tilde{\sigma}_k^{l})&=\mathrm{i}\ln\frac{\bar{f}_k^l\bar{g}_k^l}{{f}_k^l{g}_k^l},\label{fulldis-u}\\
	\phi_k^l=\frac{1}{2}(\sigma_k^l-\tilde{\sigma}_k^{l})&=\mathrm{i}\ln\frac{\bar{f}_k^l{g}_k^l}{{f}_k^l\bar{g}_k^l}.\label{fulldis-phi}\\
	x_k^l=2kac^{-1}+2lbc-\theta_k^l-\tilde{\theta}_k^l&=2kac^{-1}+2lbc+\ln\frac{g_k^l\bar{g}_k^l}{f_k^l\bar{f}_k^l}.\label{fulldis-x}
	\end{align}
	The fully discrete analogue of the gsG equation with $\nu=-1$ is given by the following theorem:
	\begin{theorem}\label{fd-gsg-1}
		The fully discrete analogue of the gsG equation with $\nu=-1$ is of the form
		\begin{align}
		&\frac{1}{b}\sin\frac{u_{k+1}^{l+1}-u_{k+1}^l-u_k^{l+1}+u_k^l}{4}=\Delta_k^l
		\sin \frac{u_{k+1}^{l+1}+u_{k+1}^l+u_k^{l+1}+u_k^l}{4},\label{dis-gsG1}
		\end{align}
		\begin{align}
		&({b^2c^2-1})\sinh\frac{x_{k+1}^{l+1}-x_{k+1}^l+x_k^{l+1}-x_k^l-4bc+4\chi_1}{2}\sinh\frac{x_{k+1}^{l+1}-x_{k+1}^l-x_k^{l+1}+x_k^l}{2}\notag\\
		&=-b^2c^2\sin \frac{u_{k+1}^{l+1}+u_{k+1}^l+u_k^{l+1}+u_k^l}{2} \sin\frac{u_{k+1}^{l+1}+u_{k+1}^l-u_k^{l+1}-u_k^l}{2},\label{Dis-gsG1}
		\end{align}
		with
		\begin{align}
		&\Delta_k^l=\frac{\sqrt{c^2-a^2}\sinh \frac{x_{k+1}^{l+1}-x_k^{l+1}+x_{k+1}^l-x_k^l-4ac^{-1}+4\chi_2}{4}}{\sqrt{b^2c^2-1}\sinh\frac{x_{k+1}^{l+1}-x_{k+1}^l+x_k^{l+1}-x_k^l-4bc+4\chi_1}{4}},\label{dis-gsG2}\\
		&\sinh\chi_1=\frac{1}{\sqrt{b^2c^2-1}},\ \sinh\chi_2=\frac{a}{\sqrt{c^2-a^2}}.\label{Dis-gsG2}
		\end{align}
		Here $u_k^l,\ x_k^l$ are defined in \eqref{fulldis-u} and \eqref{fulldis-x}. The $\tau$-functions $f_k^l$ and $g_k^l$ are defined in \eqref{dis-tau-def} with restriction \eqref{fd-res-1} for Gram-type determinants, or with restriction \eqref{fd-res-2} for Casorati-type determinants in Proposition \ref{fd-2dtl-sol}. Moreover, two conserved quantities in the full-discrete gsG equation read as
		\begin{align}
		&I_k^l=({c^2-a^2})\sinh^2\frac{x_{k+1}^l-x_k^l-2ac^{-1}+2\chi_2}{2}+c^{2}\sin^2\frac{u_{k+1}^l-u_k^l}{2}=a^2,\label{cons2}\\
		&J_k^l=\frac{{b^2c^2-1}}{b^2}\sinh^2\frac{x_k^{l+1}-x_k^l-2bc+2\chi_1}{2}+c^{2}\sin^2\frac{u_k^l+u_k^{l+1}}{2}=\frac{1}{b^2}.\label{cons1}
		\end{align}
	\end{theorem}
	
	\begin{proof}
		Note that we can rewrite \eqref{dis-gsg2}, \eqref{dis-gsg3}, \eqref{dis-gsg4} and \eqref{dis-gsg5} as
		
		\begin{align}
		\frac{1}{b}\sin\frac{\phi_k^{l+1}-\phi_{k}^l}{2}&=c\sin\frac{u_k^{l+1}+u_{k}^l}{2},\label{Dis-gsg2}\\
		\frac{1}{b}\cos\frac{\phi_k^{l+1}-\phi_k^l}{2}&=\frac{\sqrt{b^2c^2-1}}{b}\sinh\frac{x_k^{l+1}-x_k^l-2bc+2\chi_1}{2},\ \sinh\chi_1=\frac{1}{\sqrt{b^2c^2-1}},\label{Dis-gsg3}\\
		{a}\sin\frac{\phi_k^l+\phi_{k+1}^l}{2}&={c}\sin\frac{u_{k+1}^l-u_k^l}{2},\label{Dis-gsg4}\\
		{a}\cos\frac{\phi_k^l+\phi_{k+1}^l}{2}&=\sqrt{c^2-a^2}\sinh\frac{x_{k+1}^l-x_k^l-2ac^{-1}+2\chi_2}{2},\ \sinh\chi_2=\frac{a}{\sqrt{c^2-a^2}}.\label{Dis-gsg5}
		\end{align}
		By making a shift of $k\rightarrow k+1$ in \eqref{Dis-gsg2}, then adding and subtracting with \eqref{Dis-gsg2}, we obtain, respectively
		\begin{align}
		&\sin\frac{\phi_{k+1}^{l+1}-\phi_{k+1}^l+\phi_k^{l+1}-\phi_k^l}{4}\cos \frac{\phi_{k+1}^{l+1}-\phi_{k+1}^l-\phi_k^{l+1}+\phi_k^l}{4}\notag\\&=bc\sin \frac{u_{k+1}^{l+1}+u_{k+1}^l+u_k^{l+1}+u_k^l}{4} \cos\frac{u_{k+1}^{l+1}+u_{k+1}^l-u_k^{l+1}-u_k^l}{4},\label{dis-gsg9}
		\end{align}
		\begin{align}
		&\cos\frac{\phi_{k+1}^{l+1}-\phi_{k+1}^l+\phi_k^{l+1}-\phi_k^l}{4}\sin \frac{\phi_{k+1}^{l+1}-\phi_{k+1}^l-\phi_k^{l+1}+\phi_k^l}{4}\notag\\&=bc\cos \frac{u_{k+1}^{l+1}+u_{k+1}^l+u_k^{l+1}+u_k^l}{4} \sin\frac{u_{k+1}^{l+1}+u_{k+1}^l-u_k^{l+1}-u_k^l}{4}.\label{Dis-gsg9}
		\end{align}
		Similarly, from \eqref{Dis-gsg3}-\eqref{Dis-gsg5}, one can obtain
		\begin{align}
		&\sqrt{b^2c^2-1}\sinh\frac{x_{k+1}^{l+1}-x_{k+1}^l+x_k^{l+1}-x_k^l-4bc+4\chi_1}{4}\cosh\frac{x_{k+1}^{l+1}-x_{k+1}^l-x_k^{l+1}+x_k^l}{4}\notag\\
		&=\cos\frac{\phi_{k+1}^{l+1}-\phi_{k+1}^l+\phi_k^{l+1}-\phi_k^l}{4}\cos\frac{\phi_{k+1}^{l+1}-\phi_{k+1}^l-\phi_k^{l+1}+\phi_k^l}{4},\label{dis-gsg6}
		\end{align}
		\begin{align}
		&\sqrt{b^2c^2-1}\cosh\frac{x_{k+1}^{l+1}-x_{k+1}^l+x_k^{l+1}-x_k^l-4bc+4\chi_1}{4}\sinh\frac{x_{k+1}^{l+1}-x_{k+1}^l-x_k^{l+1}+x_k^l}{4}\notag\\
		&=-\sin\frac{\phi_{k+1}^{l+1}-\phi_{k+1}^l+\phi_k^{l+1}-\phi_k^l}{4}\sin\frac{\phi_{k+1}^{l+1}-\phi_{k+1}^l-\phi_k^{l+1}+\phi_k^l}{4},\label{Dis-gsg6}
		\end{align}
		\begin{align}
		&a\sin\frac{\phi_{k+1}^{l+1}+\phi_{k+1}^l+\phi_k^{l+1}+\phi_k^l}{4}\cos \frac{\phi_{k+1}^{l+1}-\phi_{k+1}^l+\phi_k^{l+1}-\phi_k^l}{4}\notag\\
		&=c\sin\frac{u_{k+1}^{l+1}+u_{k+1}^l-u_k^{l+1}-u_k^l}{4}\cos \frac{u_{k+1}^{l+1}-u_{k+1}^l-u_k^{l+1}+u_k^l}{4} ,\label{dis-gsg10}
		\end{align}
		\begin{align}
		&a\cos\frac{\phi_{k+1}^{l+1}+\phi_{k+1}^l+\phi_k^{l+1}+\phi_k^l}{4}\sin \frac{\phi_{k+1}^{l+1}-\phi_{k+1}^l+\phi_k^{l+1}-\phi_k^l}{4}\notag\\
		&=c\cos\frac{u_{k+1}^{l+1}+u_{k+1}^l-u_k^{l+1}-u_k^l}{4}\sin \frac{u_{k+1}^{l+1}-u_{k+1}^l-u_k^{l+1}+u_k^l}{4} ,\label{Dis-gsg10}
		\end{align}
		and
		\begin{align}
		&\sqrt{c^2-a^2}\sinh \frac{x_{k+1}^{l+1}-x_k^{l+1}+x_{k+1}^l-x_k^l-4ac^{-1}+4\chi_2}{4}\cosh\frac{x_{k+1}^{l+1}-x_{k+1}^l-x_k^{l+1}+x_k^l}{4}\notag\\
		&=a\cos\frac{\phi_{k+1}^{l+1}+\phi_{k+1}^l+\phi_k^{l+1}+\phi_k^l}{4}\cos\frac{\phi_{k+1}^{l+1}-\phi_{k+1}^l+\phi_k^{l+1}-\phi_k^l}{4},\label{dis-gsg7}
		\end{align}
		\begin{align}
		&\sqrt{c^2-a^2}\cosh \frac{x_{k+1}^{l+1}-x_k^{l+1}+x_{k+1}^l-x_k^l-4ac^{-1}+4\chi_2}{4}\sinh\frac{x_{k+1}^{l+1}-x_{k+1}^l-x_k^{l+1}+x_k^l}{4}\notag\\
		&=-a\sin\frac{\phi_{k+1}^{l+1}+\phi_{k+1}^l+\phi_k^{l+1}+\phi_k^l}{4}\sin\frac{\phi_{k+1}^{l+1}-\phi_{k+1}^l+\phi_k^{l+1}-\phi_k^l}{4},\label{Dis-gsg7}
		\end{align}
		respectively. Thus, eqs. \eqref{dis-gsg9} and \eqref{Dis-gsg10} give
		\begin{align}
		&\frac{1}{a b} \sin \frac{u_{k+1}^{l+1}-u_{k+1}^l-u_k^{l+1}+u_k^l}{4}\cos\frac{\phi_{k+1}^{l+1}-\phi_{k+1}^l-\phi_k^{l+1}+\phi_k^l}{4}\notag \\
		&=\sin \frac{u_{k+1}^{l+1}+u_{k+1}^l+u_k^{l+1}+u_k^l}{4}\cos\frac{\phi_{k+1}^{l+1}+\phi_{k+1}^l+\phi_k^{l+1}+\phi_k^l}{4}.\label{dis-gsg1}
		\end{align}
		Eqs. \eqref{dis-gsg6} and \eqref{dis-gsg7} lead to
		\begin{align}
		&a\sqrt{b^2c^2-1}\sinh\frac{x_{k+1}^{l+1}-x_{k+1}^l+x_k^{l+1}-x_k^l-4bc+4\chi_1}{4}\cos\frac{\phi_{k+1}^{l+1}+\phi_{k+1}^l+\phi_k^{l+1}+\phi_k^l}{4}\notag\\
		&=\sqrt{c^2-a^2}\sinh \frac{x_{k+1}^{l+1}-x_k^{l+1}+x_{k+1}^l-x_k^l-4ac^{-1}+4\chi_2}{4} \cos\frac{\phi_{k+1}^{l+1}-\phi_{k+1}^l-\phi_k^{l+1}+\phi_k^l}{4}.\label{dis-gsg8}
		\end{align}
		 A substitution of \eqref{dis-gsg8} into \eqref{dis-gsg1} leads to
		\begin{align}
		\frac{1}{b}\sin\frac{u_{k+1}^{l+1}-u_{k+1}^l-u_k^{l+1}+u_k^l}{4}=\Delta_k^l
		\sin \frac{u_{k+1}^{l+1}+u_{k+1}^l+u_k^{l+1}+u_k^l}{4},\label{fdis-gsG1}
		\end{align}
		with
		\begin{align}
		\Delta_k^l=\frac{\sqrt{c^2-a^2}\sinh \frac{x_{k+1}^{l+1}-x_k^{l+1}+x_{k+1}^l-x_k^l-4ac^{-1}+4\chi_2}{4}}{\sqrt{b^2c^2-1}\sinh\frac{x_{k+1}^{l+1}-x_{k+1}^l+x_k^{l+1}-x_k^l-4bc+4\chi_1}{4}}.\label{fdis-gsG2}
		\end{align}
		On the other hand, by multiplying  \eqref{dis-gsg9} and \eqref{Dis-gsg9}, \eqref{dis-gsg6} and \eqref{Dis-gsg6}, we obtain, respectively
		\begin{align}
		&b^2c^2\sin \frac{u_{k+1}^{l+1}+u_{k+1}^l+u_k^{l+1}+u_k^l}{2} \sin\frac{u_{k+1}^{l+1}+u_{k+1}^l-u_k^{l+1}-u_k^l}{2}\notag\\&=\sin\frac{\phi_{k+1}^{l+1}-\phi_{k+1}^l+\phi_k^{l+1}-\phi_k^l}{2}\sin \frac{\phi_{k+1}^{l+1}-\phi_{k+1}^l-\phi_k^{l+1}+\phi_k^l}{2},\label{dis-gsg11}
		\end{align}
		\begin{align}
		&({b^2c^2-1})\sinh\frac{x_{k+1}^{l+1}-x_{k+1}^l+x_k^{l+1}-x_k^l-4bc+4\chi_1}{2}\sinh\frac{x_{k+1}^{l+1}-x_{k+1}^l-x_k^{l+1}+x_k^l}{2}\notag\\
		&=-\sin\frac{\phi_{k+1}^{l+1}-\phi_{k+1}^l+\phi_k^{l+1}-\phi_k^l}{2}\sin\frac{\phi_{k+1}^{l+1}-\phi_{k+1}^l-\phi_k^{l+1}+\phi_k^l}{2},\label{Dis-gsg11}
		\end{align}
		which leads to exactly \eqref{Dis-gsG1} by eliminating the right side of the equations.
		Meanwhile, from \eqref{Dis-gsg2}-\eqref{Dis-gsg5}, we have
		\begin{align}
		&J_k^l=\frac{{b^2c^2-1}}{b^2}\sinh^2\frac{x_k^{l+1}-x_k^l-2bc+2\chi_1}{2}+c^{2}\sin^2\frac{u_k^l+u_k^{l+1}}{2}=\frac{1}{b^2},\label{dcons1}\\
		&I_k^l=({c^2-a^2})\sinh^2\frac{x_{k+1}^l-x_k^l-2ac^{-1}+2\chi_2}{2}+c^{2}\sin^2\frac{u_{k+1}^l-u_k^l}{2}=a^2.\label{dcons2}
		\end{align}
		Here $a^2$ and $\frac{1}{b^2}$ are constants, thus equations \eqref{dcons1} and \eqref{dcons2} actually give conserved quantities.
		The proof is complete.
	\end{proof}
	
	\subsection{Fully discretization of the gsG equation with $\nu=1$}
	
	In order to construct the fully discrete analogue of the gsG equation with $\nu=1$, we choose the restriction as
	\begin{align}
	c=\lambda\mathrm{i},\,c_{ij}=\mathrm{i}\sqrt{\frac{1-cp_i}{1+cq_j}}\delta_{ij},\label{fd-res1}
	\end{align}
	for Gram-type solution, or
	\begin{align}
	c=\lambda\mathrm{i},\,d_{i}=\mathrm{i} {\sqrt{\frac{1-cp_i}{1-cq_i}}}
	c_i,
	\label{fd-res2}
	\end{align}
	for Casorati-type, which implies
	\begin{align}
	\tilde{f}_k^l\Bumpeq\bar{g}_k^l,\ \tilde{g}_k^l\Bumpeq\bar{f}_k^l.
	\end{align}
	Next we introduce dependent variable transformations
	\begin{align}
	&u_k^l=\frac{1}{2}(\sigma_k^l+\tilde{\sigma}_k^{l})=\mathrm{i}\ln\frac{\bar{f}_k^l\bar{g}_k^l}{{f}_k^l{g}_k^l},\label{fulldis-1u}\\
	&\varphi_k^l=\frac{1}{2\mathrm{i}}(\sigma_k^l-\tilde{\sigma}_k^{l})=\ln\frac{\bar{g}_k^l{g}_k^l}{{f}_k^l\bar{f}_k^l},\label{fulldis-1phi}
	\end{align}
	and
	\begin{align}
	&\tilde{x}_k^l=2ka\lambda^{-1}-2lb\lambda-\mathrm{i}(\theta_k^l+\tilde{\theta}_k^l)=2ka\lambda^{-1}-2lb\lambda+\mathrm{i}\ln\frac{\bar{f}_k^lg_k^l}{f_k^l\bar{g}_k^l},\label{fulldis-1x}\\
	&\tilde{t}_k^l=2lb\lambda.
	\end{align}
	We can construct the fully discrete analogue of the gsG equation with $\nu =1$ through the following theorem.
	\begin{theorem}\label{fdis-nu1}
		The fully discrete analogue of the gsG equation with $\nu=1$ is of the form
		\begin{align}
		&\frac{1}{b} \sin \frac{u_{k+1}^{l+1}-u_{k+1}^l-u_k^{l+1}+u_k^l}{4}
		=\tilde{\Delta}_k^l\sin \frac{u_{k+1}^{l+1}+u_{k+1}^l+u_k^{l+1}+u_k^l}{4},\label{dis-1gsG1}
		\end{align}
		\begin{align}
		&({1+b^2\lambda^2})\sin \frac{\tilde{x}_{k+1}^{l+1}-\tilde{x}_{k+1}^l+\tilde{x}_k^{l+1}-\tilde{x}_k^l+4b\lambda+4\omega_1}{2}\sin\frac{\tilde{x}_{k+1}^{l+1}-\tilde{x}_{k+1}^l-\tilde{x}_k^{l+1}+\tilde{x}_k^l}{2}\notag\\
		&=b^2\lambda^2\sin \frac{u_{k+1}^{l+1}+u_{k+1}^l+u_k^{l+1}+u_{k}^l}{2}
		\sin \frac{u_{k+1}^{l+1}+u_{k+1}^l-u_k^{l+1}-u_{k}^l}{2},
		\end{align}
		with
		\begin{align}
		&	\tilde{\Delta}_k^l=\frac{\sqrt{a^2+\lambda^2}\sin\frac{\tilde{x}_{k+1}^{l+1}-\tilde{x}_k^{l+1}+\tilde{x}_{k+1}^l-\tilde{x}_k^l-4a\lambda^{-1}+4\omega_2}{4}}{\sqrt{1+b^2\lambda^2}\sin \frac{\tilde{x}_{k+1}^{l+1}-\tilde{x}_{k+1}^l+\tilde{x}_k^{l+1}-\tilde{x}_k^l+4b\lambda+4\omega_1}{4}},\label{dis-1gsG2}\\
		&\sin\omega_1=\frac{1}{\sqrt{b^2\lambda^2+1}},\ \sin\omega_2=\frac{a}{\sqrt{a^2+\lambda^2}}.
		\end{align}
		Here $u_k^l,\ x_k^l$ are defined in \eqref{fulldis-1u} and \eqref{fulldis-1x}. The $\tau$-functions $f_k^l$ and $g_k^l$ are defined in \eqref{dis-tau-def} with restriction \eqref{fd-res1} for Gram-type determinants, or with restriction \eqref{fd-res2} for Casorati-type determinants in Proposition \ref{fd-2dtl-sol}. Moreover, there are two conserved quantities in the full-discrete gsG equation read as
		\begin{align}
		&\tilde{I}_k^l=\left(a\cos\frac{\tilde{x}_{k+1}^l-\tilde{x}_k^l-2a\lambda^{-1}}{2}+\lambda\sin\frac{\tilde{x}_{k+1}^l-\tilde{x}_k^l-2a\lambda^{-1}}{2}\right)^2-\lambda^2\sin^2\frac{u_{k+1}^l-u_k^l}{2}=a^2,\\
		&\tilde{J}_k^l=\left(\frac{1}{\lambda}\cos\frac{\tilde{x}_k^{l+1}-\tilde{x}_k^l+2b\lambda}{2}+\lambda\sin\frac{\tilde{x}_k^{l+1}-\tilde{x}_k^l+2b\lambda}{2}\right)^2-\lambda^2\sin^2\frac{u_k^{l+1}+u_{k}^l}{2}=\frac{1}{b^2}.
		\end{align}
	\end{theorem}
	\begin{remark}
		Fully discrete analogues of the generalized sG equation with $\nu=1$ can also be transformed into the case with $\nu=-1$ through
		\begin{align}
		\phi_k^l=\mathrm{i}\varphi_{k}^l,\ \tilde{x}_k^l=\mathrm{i}x_k^l,\ \tilde{t}=-\mathrm{i}t,\
		c=\lambda\mathrm{i}.
		\end{align}
	\end{remark}
	The detail proof of Theorem \ref{fdis-nu1} is similar to the case with $\nu=-1$, which is given in Appendix \ref{theo-fdis-nu1}.
	
	\subsection{Reduction to the discrete sG and discrete SP equation}
	
	In this section, we mainly investigate the reduction from the discrete gsG equation to the discrete sG equation \cite{fd-sg1,fd-sg2,fd-sg3} and the discrete SP equation \cite{dis-sp}. The $\tau$-functions, as well as the variable transformations of the discrete sG equation and the discrete SP equation, can be derived from those of the discrete gsG equation. The parameter $c$ (or $\lambda$) also plays an important role in the reduction of the discrete case.
	
	\subsubsection{Reduction to the discrete sG equation}
	{\bf (I) From the discrete gsG equation with $\nu=-1$ to the discrete sG equation.}
	In this part of reduction, we take $c\in\mathbb{R}$ as a small parameter. Then the elements of the $\tau$-function in Theorem \ref{fd-gsg-1} can be written as
	\begin{align}
	m_{ij}^{(n)}(k, l,1)&=\mathrm{i}\delta_{i j}+\frac{1}{p_i+p_j}\left(-\frac{p_i}{p_j}\right)^{n}\left(\frac{1-a p_i}{1+a p_j}\right)^{-k}\left(\frac{1-b p_i^{-1}}{1+b p_j^{-1}}\right)^{-l}\left(\frac{1-c p_i}{1+c p_j}\right)\notag\\
	&=\mathrm{i}\delta_{i j}+\frac{1}{p_i+p_j}\left(-\frac{p_i}{p_j}\right)^{n}\left(\frac{1-a p_i}{1+a p_j}\right)^{-k}\left(\frac{1-b p_i^{-1}}{1+b p_j^{-1}}\right)^{-l}(1-(p_i+p_j)c)+\mathnormal{O} (c^2)\notag\\
	&=m_{ij}^{(n)}(k, l,0)-(p_i+p_j)m_{ij}^{(n)}(k, l,0)c+\mathnormal{O}(c^2).
	\end{align}
	We can rewrite the dependent variable transformations and define new variable as
	\begin{align}
	&u_k^l=\mathrm{i}\ln\frac{\bar{f}_k^l\bar{g}_k^l}{{f}_k^l{g}_k^l}=2\mathrm{i}\ln \frac{\bar{f}_k^l}{{f}_k^l}+\mathnormal{O}(c),\label{fulldis-u-sg}\\
	&\phi_k^l=\mathrm{i}\ln\frac{\bar{f}_k^l{g}_k^l}{{f}_k^l\bar{g}_k^l}=\mathrm{i}\ln \frac{\bar{f}_k^l{f}_k^l}{f_k^l\bar{f}_k^l}+\mathnormal{O}(c)=\mathnormal{O}(c),\label{fulldis-phi-sg}\\
	&\hat{x}_k^l=cx_k^l=2ka+2lbc^2+c\ln\frac{g_k^l\bar{g}_k^l}{f_k^l\bar{f}_k^l}=2ka+\mathnormal{O}(c^2),\label{fulldis-x-sg}\\
	&\hat{t}_k^l=c^{-1}{t_k^l}=2lb.\label{fulldis-t-sg}
	\end{align}
	In the limit of $c\rightarrow0$, \eqref{fulldis-u-sg}-\eqref{fulldis-x-sg} lead to
	\begin{align}
	&u_k^l=2\mathrm{i}\ln \frac{\bar{f}_k^l}{{f}_k^l},\ \phi_k^l=0,\ \hat{x}_k^l=2ka.\label{fulldis-x-sg1}
	\end{align}
	The definition of $u_k^l$ corresponds to the dependent transformation of the discrete sG equation. Substituting \eqref{fulldis-u-sg}-\eqref{fulldis-x-sg} into \eqref{dis-gsG1} and \eqref{dis-gsG2} and taking $c\rightarrow0$, one can obtain
	\begin{align}
	\frac{1}{b}\sin\frac{u_{k+1}^{l+1}-u_{k+1}^l-u_k^{l+1}+u_k^l}{4}=a
	\sin \frac{u_{k+1}^{l+1}+u_{k+1}^l+u_k^{l+1}+u_k^l}{4},\label{dis-sg}
	\end{align}
	which is just the discrete sG equation \cite{fd-sg1,fd-sg2,fd-sg3}.
	\\
	\\
	{\bf (II) {From the discrete gsG equation with $\nu=1$ to the discrete sG equation.}}
	In this part of reduction, we take $\lambda\in\mathbb{R}$ as a small parameter. Here we have
	\begin{align}
	m_{ij}^{(n)}(k, l,1)&=\mathrm{i}\sqrt{\frac{1-\mathrm{i}\lambda p_i}{1+\mathrm{i}\lambda q_j}}\delta_{ij}+\frac{1}{p_i+p_j}\left(-\frac{p_i}{p_j}\right)^{n}\left(\frac{1-a p_i}{1+a p_j}\right)^{-k}\left(\frac{1-b p_i^{-1}}{1+b p_j^{-1}}\right)^{-l}\left(\frac{1-\mathrm{i}\lambda p_i}{1+\mathrm{i}\lambda p_j}\right)\notag\\
	&=\mathrm{i}\delta_{i j}+\frac{1}{p_i+p_j}\left(-\frac{p_i}{p_j}\right)^{n}\left(\frac{1-a p_i}{1+a p_j}\right)^{-k}\left(\frac{1-b p_i^{-1}}{1+b p_j^{-1}}\right)^{-l}+\mathnormal{O} (\lambda)\notag\\
	&=\hat{m}_{ij}^{(n)}(k, l,0)+\mathnormal{O}(\lambda),\\
	m_{ij}^{(n)}(k, l,0)&=\mathrm{i}\sqrt{\frac{1-\mathrm{i}\lambda p_i}{1+\mathrm{i}\lambda q_j}}\delta_{ij}+\frac{1}{p_i+p_j}\left(-\frac{p_i}{p_j}\right)^{n}\left(\frac{1-a p_i}{1+a p_j}\right)^{-k}\left(\frac{1-b p_i^{-1}}{1+b p_j^{-1}}\right)^{-l}\notag\\
	&=\hat{m}_{ij}^{(n)}(k, l,0)+\mathnormal{O}(\lambda).
	\end{align}
	We can rewrite the dependent variable transformations and define new variable as
	\begin{align}
	u_k^l=2\mathrm{i}\ln\frac{\bar{\hat{f}}_k^l}{{\hat{f}}_k^l}+\mathnormal{O}(\lambda),\
	\varphi_k^l=\mathnormal{O}(\lambda),\
	\hat{x}_k^l=\lambda\tilde{x}_k^l=2ka+\mathnormal{O}(\lambda^2),\
	\hat{t}_k^l=\frac{\tilde{t}_k^l}{\lambda}=2lb.\label{fulldis-1t-sg}
	\end{align}
	In the limit of $\lambda\rightarrow0$, eqs. \eqref{fulldis-1t-sg} and \eqref{dis-1gsG1} converge to
	\begin{align}
	u_k^l=2\mathrm{i}\ln\frac{\bar{\hat{f}}_k^l}{{\hat{f}}_k^l},\
	\varphi_k^l=0,\
	&\hat{x}_k^l=2ka,\ \hat{t}_k^l=2lb,\\
	\frac{1}{ab} \sin \frac{u_{k+1}^{l+1}-u_{k+1}^l-u_k^{l+1}+u_k^l}{4}
	&=\sin \frac{u_{k+1}^{l+1}+u_{k+1}^l+u_k^{l+1}+u_k^l}{4},\label{dis-1sG1}
	\end{align}
	respectively. The latter is exactly the fully discrete sG equation.
	
	\subsubsection{Reduction to the discrete SP equation}
	
	{\bf (I) From the discrete gsG equation with $\nu=-1$ to the discrete SP equation with $\sigma=-1$.}
	In this part of reduction, we take $c$ as a big parameter (or $\epsilon=\frac{1}{c}$ small). Let us introduce a new auxiliary parameter $s$, and redefine the matrix elements of the $\tau$ function as
	\begin{align}
	\tilde{m}_{ij}^{(n)}(k,l,m)=\mathrm{i}\delta_{i j}+\frac{1}{p_i+p_j}\left(-\frac{p_i}{p_j}\right)^{n}\left(\frac{1-a p_i}{1+a p_j}\right)^{-k}\left(\frac{1-b p_i^{-1}}{1+b p_j^{-1}}\right)^{-l}\left(\frac{1-c p_i}{1+c p_j}\right)^{m}e^{\left(\frac{1}{2p_i}+\frac{1}{2p_j}\right)s}.\label{new-tau}
	\end{align}
	It is obviously that $	\tau_n(k, l,m)=\left|\tilde{m}_{ij}^{n}(k, l,m)\right|_{1 \leqslant i, j \leqslant N}$ is still the solution of \eqref{dis-2dtl2}-\eqref{dis-2dtl3}.
	If we take $c$ as a big parameter (or $\epsilon=\frac{1}{c}$ small), the elements of the $\tau$ function can be written as
	\begin{align}
	\tilde{m}_{ij}^{(n)}(k, l,1)&=\mathrm{i}\delta_{i j}+\frac{1}{p_i+p_j}\left(-\frac{p_i}{p_j}\right)^{n}\left(\frac{1-a p_i}{1+a p_j}\right)^{-k}\left(\frac{1-b p_i^{-1}}{1+b p_j^{-1}}\right)^{-l}\left(\frac{1-c p_i}{1+c p_j}\right)^{1}e^{\left(\frac{1}{2p_i}+\frac{1}{2p_j}\right)s}\notag\\
	&= \mathrm{i}\delta_{i j}+\frac{1}{p_i+p_j}\left(-\frac{p_i}{p_j}\right)^{n}\left(\frac{1-a p_i}{1+a p_j}\right)^{-k}\left(\frac{1-b p_i^{-1}}{1+b p_j^{-1}}\right)^{-l}e^{\left(\frac{1}{2p_i}+\frac{1}{2p_j}\right)s}\left(-\frac{p_i}{p_j}+\frac{p_i+p_j}{p_j^2}\epsilon\right)+\mathnormal{O} (\epsilon^2)\notag\\
	&=m_{ij}^{(n+1)}(k, l,0)-2\epsilon\frac{\dif}{\dif s}m_{ij}^{(n+1)}(k, l,0)+\mathnormal{O}(\epsilon^2),
	\end{align}
	from which one can obtain
	\begin{align}
	&g_k^l=\tau_0(k,l,1)=\bar{f}_k^l-2\epsilon\frac{\dif}{\dif s}\bar{f}_k^l+\mathnormal{O}(\epsilon^2),\\
	&\bar{g}_k^l=\tau_1(k,l,1)= {f}_k^l-2\epsilon\frac{\dif}{\dif s}{f}_k^l+\mathnormal{O}(\epsilon^2).
	\end{align}
	We rewrite the dependent variable transformations and introduce new variable as
	\begin{align}
	&\epsilon \hat{u}_k^l=u_k^l=\mathrm{i}\ln\frac{\bar{f}_k^l\bar{g}_k^l}{{f}_k^l{g}_k^l}=\mathrm{i}\ln \frac{\bar{f}_k^l(f_k^l-2\epsilon\frac{\dif}{\dif s}{f}_k^l)}{{f}_k^l(\bar{f}_k^l-2\epsilon\frac{\dif}{\dif s}\bar{f}_k^l)}+\mathnormal{O}(\epsilon^2),\label{fulldis-u-sp}\\
	&\phi_k^l=\mathrm{i}\ln\frac{\bar{f}_k^l{g}_k^l}{{f}_k^l\bar{g}_k^l}=\mathrm{i}\ln \frac{\bar{f}_k^l(\bar{f}_k^l-2\epsilon\frac{\dif}{\dif s}\bar{f}_k^l)}{{f}_k^l(f_k^l-2\epsilon\frac{\dif}{\dif s}{f}_k^l)}+\mathnormal{O}(\epsilon^2),\label{fulldis-phi-sp}\\
	&\epsilon \hat{x}_k^l+\frac{2lb}{\epsilon}=x_k^l=2ka\epsilon+\frac{2lb}{\epsilon}+\ln\frac{(\bar{f}_k^l-2\epsilon\frac{\dif}{\dif s}\bar{f}_k^l)(f_k^l-2\epsilon\frac{\dif}{\dif s}{f}_k^l)}{f_k^l\bar{f}_k^l}+\mathnormal{O}(\epsilon^2),\label{fulldis-x-sp}\\
	&\hat{t}_k^l=\epsilon t_k^l=2lb.\label{fulldis-t-sp}
	\end{align}
	In the limit of $\epsilon\rightarrow0$, \eqref{fulldis-u-sp}-\eqref{fulldis-x-sp} lead to
	\begin{align}
	&\hat{u}_k^l=2\mathrm{i}\left(\ln\frac{\bar{f}_k^l}{f_k^l}\right)_s,\label{fulldis-u-sp1}\\
	&\phi_k^l=2\mathrm{i}\ln\frac{\bar{f}_k^l}{f_k^l},\label{fulldis-phi-sp1}\\
	&\hat{x}_k^l=2ka-2(\ln(f_k^l\bar{f}_k^l))_s.\label{fulldis-x-sp1}
	\end{align}
	Substituting \eqref{fulldis-u-sp}-\eqref{fulldis-x-sp} into \eqref{dis-gsG1}-\eqref{dis-gsG2} and taking $\epsilon\rightarrow0$, one can obtain
	\begin{align}
	&({\hat{x}_{k+1}^{l+1}-\hat{x}_{k+1}^l+\hat{x}_k^{l+1}-\hat{x}_k^l}+\frac{4}{b})({\hat{u}_{k+1}^{l+1}-\hat{u}_{k+1}^l-\hat{u}_k^{l+1}+\hat{u}_k^l})\notag\\
	&=(\hat{x}_{k+1}^{l+1}-\hat{x}_k^{l+1}+\hat{x}_{k+1}^l-\hat{x}_k^l)(\hat{u}_{k+1}^{l+1}+\hat{u}_{k+1}^l+\hat{u}_k^{l+1}+\hat{u}_k^l),\label{dis-sp1}
	\end{align}
	\begin{align}
	&({\hat{x}_{k+1}^{l+1}-\hat{x}_{k+1}^l+\hat{x}_k^{l+1}-\hat{x}_k^l}+\frac{4}{b})({\hat{x}_{k+1}^{l+1}-\hat{x}_{k+1}^l-\hat{x}_k^{l+1}+\hat{x}_k^l})\notag\\
	&=-(\hat{u}_{k+1}^{l+1}+\hat{u}_{k+1}^l+\hat{u}_k^{l+1}+\hat{u}_k^l)({\hat{u}_{k+1}^{l+1}+\hat{u}_{k+1}^l-\hat{u}_k^{l+1}-\hat{u}_k^l}),\label{dis-sp2}
	\end{align}
	which lead to the discrete SP equations. In addition, conserved quantities $I_k^l$ and $J_k^l$ are recast into
	\begin{align}
	&J_k^l=\left(\frac{1}{b}+\frac{\hat{x}_k^{l+1}-\hat{x}_k^l}{2}\right)^2+\left(\frac{\hat{u}_k^l+\hat{u}_k^{l+1}}{2}\right)^2,\\
	&I_k^l=\left(\frac{\hat{x}_{k+1}^l-\hat{x}_k^l}{2}\right)^2+\left(\frac{\hat{u}_{k+1}^l-\hat{u}_k^l}{2}\right)^2,
	\end{align}
	which correspond to the conserved quantities of the discrete SP equation derived in \cite{dis-sp}.
	
	{\bf (II) From the discrete gsG equation with $\nu=1$ to the discrete SP equation with $\sigma=1$.}
	Here we give a discrete analog of the SP equation with $\sigma=1$ through the variable transformation and the scaling limit from the discrete gsG equation with $\nu=1$. Similar to the case with $\nu=-1$, we take $\epsilon=\frac{1}{\lambda}$ as a small parameter. Introducing an auxiliary parameter $s$, and redefining the matrix elements of the $\tau$ function as shown in \eqref{new-tau}, we have
	\begin{align}
	&\tilde{m}_{ij}^{(n)}(k, l,0)\propto\hat{m}_{ij}^{(n+1)}(k, l)-\mathrm{i}\epsilon\frac{\dif}{\dif s}\hat{m}_{ij}^{(n+1)}(k, l)+\mathnormal{O}(\epsilon^2),\\
	&\tilde{m}_{ij}^{(n+1)}(k, l,0)\propto\hat{m}_{ij}^{(n)}(k, l)-\mathrm{i}\epsilon\frac{\dif}{\dif s}\hat{m}_{ij}^{(n)}(k, l)+\mathnormal{O}(\epsilon^2),
	\end{align}
	where
	\begin{align}
	\hat{m}_{ij}^{(n)}(k,l)=\delta_{i j}+\frac{1}{p_i+p_j}\left(-\frac{p_i}{p_j}\right)^{n}\left(\frac{1-a p_i}{1+a p_j}\right)^{-k}\left(\frac{1-b p_i^{-1}}{1+b p_j^{-1}}\right)^{-l}e^{\left(\frac{1}{2p_i}+\frac{1}{2p_j}\right)s}.
	\end{align}
	Thus one can obtain
	\begin{align}
	&f_k^l=\hat{f}_k^l-{\mathrm{i}\epsilon}\frac{\dif}{\dif \tau}\hat{f}_k^l+\mathnormal{O}(\epsilon^2),\ \bar{f}_k^l=\hat{f}_k^l+{\mathrm{i}\epsilon}\frac{\dif}{\dif \tau}\hat{f}_k^l+\mathnormal{O}(\epsilon^2),\\ &g_k^l=\hat{g}_k^l+{\mathrm{i}\epsilon}\frac{\dif}{\dif \tau}\hat{g}_k^l+\mathnormal{O}(\epsilon^2),\ \bar{g}_k^l=\hat{g}_k^l-{\mathrm{i}\epsilon}\frac{\partial}{\partial \tau}\hat{g}_k^l+\mathnormal{O}(\epsilon^2),\\
	&\hat{f}_k^l=\left|\hat{m}_{ij}^{(n+1)}(k, l)\right|_{1 \leqslant i, j \leqslant N},\ \hat{g}_k^l=\left|\hat{m}_{ij}^{(n)}(k, l)\right|_{1 \leqslant i, j \leqslant N}.
	\end{align}
	Similar to the continuou, we rewrite the dependent variable transformations and introduce new variables, then under the limit $\epsilon \rightarrow0$, we have
	\begin{align}
	\hat{u}_k^l=2\left(\ln\frac{\hat{g}_k^l}{\hat{f}_k^l}\right)_s,\ \varphi_k^l=2\ln\frac{\hat{g}_k^l}{\hat{f}_k^l},\ \hat{x}_k^l=2ka-2(\ln(\hat{f}_k^l\hat{g}_k^l))_s,\ \hat{t}_k^l=2lb.\label{fulldis-1x-sp1}
	\end{align}
	Moreover, one can propose a fully discrete analogue of the SP equation with $\sigma=1$
	\begin{align}
	&({\hat{x}_{k+1}^{l+1}-\hat{x}_{k+1}^l+\hat{x}_k^{l+1}-\hat{x}_k^l}+\frac{4}{b})({\hat{u}_{k+1}^{l+1}-\hat{u}_{k+1}^l-\hat{u}_k^{l+1}+\hat{u}_k^l})\notag\\
	&=(\hat{x}_{k+1}^{l+1}-\hat{x}_k^{l+1}+\hat{x}_{k+1}^l-\hat{x}_k^l)(\hat{u}_{k+1}^{l+1}+\hat{u}_{k+1}^l+\hat{u}_k^{l+1}+\hat{u}_k^l),\label{dis-1sp1}
	\end{align}
	\begin{align}
	&({\hat{x}_{k+1}^{l+1}-\hat{x}_{k+1}^l+\hat{x}_k^{l+1}-\hat{x}_k^l}+\frac{4}{b})({\hat{x}_{k+1}^{l+1}-\hat{x}_{k+1}^l-\hat{x}_k^{l+1}+\hat{x}_k^l})\notag\\
	&=(\hat{u}_{k+1}^{l+1}+\hat{u}_{k+1}^l+\hat{u}_k^{l+1}+\hat{u}_k^l)({\hat{u}_{k+1}^{l+1}+\hat{u}_{k+1}^l-\hat{u}_k^{l+1}-\hat{u}_k^l}),\label{dis-1sp2}
	\end{align}
	and conserved quantities $\tilde{I}_k^l$ and $\tilde{J}_k^l$ are recast into
	\begin{align}
	&\tilde{J}_k^l=\left(\frac{1}{b}+\frac{\hat{x}_k^{l+1}-\hat{x}_k^l}{2}\right)^2-\left(\frac{\hat{u}_k^l+\hat{u}_k^{l+1}}{2}\right)^2,\\
	&\tilde{I}_k^l=(\frac{\hat{x}_{k+1}^l-\hat{x}_k^l}{2})^2-\left(\frac{\hat{u}_{k+1}^l-\hat{u}_k^l}{2}\right)^2.
	\end{align}
	Here $N$-soliton solutions of the full-discrete SP equation with $\sigma=1$ also exhibits the singular nature.
	
	\section{The semi-discrete gsG equation}\label{sec4}
	\subsection{From the fully discrete gsG equation to the semi-discrete gsG equation}	
	In this section, we demonstrate that the proposed fully discrete gsG equation with $\nu=-1$ converges to the semi-discrete equation we obtained in \cite{gsg-F} in the continuous limit $b\rightarrow0$. Moreover, we give a semi-discrete gsG equation with $\nu=1$ through the same continuous limit from the fully discrete gsG equation with $\nu=1$.
	
	\subsubsection{The semi-discrete gsG equation with $\nu=-1$}\label{semi-1}
	Recall that
	\begin{align}
	\Delta_k^l&=\frac{\sqrt{c^2-a^2}\sinh \frac{x_{k+1}^{l+1}-x_k^{l+1}+x_{k+1}^l-x_k^l-4ac^{-1}+4\chi_2}{4}}{\sqrt{b^2c^2-1}\sinh\frac{x_{k+1}^{l+1}-x_{k+1}^l+x_k^{l+1}-x_k^l-4bc+4\chi_1}{4}}\notag\\
	&=\frac{\sqrt{c^2-a^2}\sinh \frac{x_{k+1}^{l+1}-x_k^{l+1}+x_{k+1}^l-x_k^l-4ac^{-1}+4\chi_2}{4}}{\cosh \frac{x_{k+1}^{l+1}-x_{k+1}^l+x_k^{l+1}-x_k^l-4bc}{4}+bc\sinh\frac{x_{k+1}^{l+1}-x_{k+1}^l+x_k^{l+1}-x_k^l-4bc}{4}}.
	\end{align}
	Obviously, as $b\rightarrow0$, we have
	\begin{align}
	&\cosh \frac{x_{k+1}^{l+1}-x_{k+1}^l+x_k^{l+1}-x_k^l-4bc}{4}+bc\sinh\frac{x_{k+1}^{l+1}-x_{k+1}^l+x_k^{l+1}-x_k^l-4bc}{4}\rightarrow 1,\\
	&\sqrt{c^2-a^2}\sinh \frac{x_{k+1}^{l+1}-x_k^{l+1}+x_{k+1}^l-x_k^l-4ac^{-1}+4\chi_2}{4}\rightarrow \sqrt{c^2-a^2}\sinh \frac{x_{k+1}-x_k-2ac^{-1}+2\chi_2}{2}.
	\end{align}
	Note the relation \eqref{cons2} holds. Thus we have
	\begin{align}
	\Delta_k^l\rightarrow \sqrt{a^2-c^2\sin^2\frac{u_{k+1}-u_k}{2}}=\frac{\Delta_k}{2}.
	\end{align}
	Furthermore, we can easily verify that
	\begin{align}
	\frac{1}{b}\sin\frac{u_{k+1}^{l+1}-u_{k+1}^l-u_k^{l+1}+u_k^l}{4}\rightarrow \frac{\dif}{\dif \tau} \frac{u_{k+1}-u_k}{2},
	\end{align}
	and
	\begin{align}
	\sin \frac{u_{k+1}^{l+1}+u_{k+1}^l+u_k^{l+1}+u_k^l}{4}\rightarrow \sin \frac{u_{k+1}+u_k}{2}.
	\end{align}
	Thus equation \eqref{dis-gsG1} converges to
	\begin{align}
	&	\frac{\dif }{\dif \tau}\left(u_{k+1}-u_k\right)=\Delta_k\sin\frac{u_{k+1}+u_k}{2},\label{sm-gsg1}\\
	&	\Delta_k=\sqrt{4a^2-4c^2\sin^2\frac{u_{k+1}-u_k}{2}}\label{sm-gsg3},
	\end{align}
	which are actually part of the semi-discrete gsG equation with $\nu=-1$ (see eqs. (3.47) and (3.49) in \cite{gsg-F}).
	Here we used $\frac{f^{l+1}-f^l}{2b}\rightarrow \frac{\dif}{\dif\tau}f$ as $b\rightarrow0$. For equation \eqref{Dis-gsG1}, the continuous limit $b\rightarrow0$ leads to
	\begin{align}
	&({b^2c^2-1})\sinh\frac{x_{k+1}^{l+1}-x_{k+1}^l+x_k^{l+1}-x_k^l-4bc+4\chi_1}{2}\notag\\
	=&({b^2c^2+1})\sinh\frac{x_{k+1}^{l+1}-x_{k+1}^l+x_k^{l+1}-x_k^l-4bc}{2}+2bc\cosh \frac{x_{k+1}^{l+1}-x_{k+1}^l+x_k^{l+1}-x_k^l-4bc}{2}\notag\\
	\sim& ({b^2c^2+1})\left[\left(\frac{\dif}{\dif\tau}(x_k+x_{k+1})-2c\right)b+\mathcal{O}(b^2)\right]+2bc+\mathcal{O}(b^3)\notag\\
	\sim& \frac{\dif}{\dif\tau}(x_k+x_{k+1})b+\mathcal{O}(b^3),
	\end{align}
	\begin{align}
	\sinh\frac{x_{k+1}^{l+1}-x_{k+1}^l-x_k^{l+1}+x_k^l}{2}\sim \frac{\dif}{\dif\tau}(x_{k+1}-x_{k})b+\mathcal{O}(b^3).
	\end{align}
	Note that
	\begin{align}
	&\sin \frac{u_{k+1}^{l+1}+u_{k+1}^l+u_k^{l+1}+u_k^l}{2} \sin\frac{u_{k+1}^{l+1}+u_{k+1}^l-u_k^{l+1}-u_k^l}{2}\notag\\
	=&\frac{1}{2}\left(\cos(u_k^{l+1}+u_k^l)-\cos (u_{k+1}^{l+1}+u_{k+1}^l)\right)\notag\\
	=&\cos^2\frac{u_k^{l+1}+u_k^l}{2}- \cos^2\frac{u_{k+1}^{l+1}+u_{k+1}^l}{2}\notag\\
	=&(\cos\frac{u_k^{l+1}+u_k^l}{2}+ \cos\frac{u_{k+1}^{l+1}+u_{k+1}^l}{2})(\cos\frac{u_k^{l+1}+u_k^l}{2}- \cos\frac{u_{k+1}^{l+1}+u_{k+1}^l}{2}),
	\end{align}
	and
	\begin{align}
	&bc\left(\cos\frac{u_k^{l+1}+u_k^l}{2}+ \cos\frac{u_{k+1}^{l+1}+u_{k+1}^l}{2}\right)\notag\\
	=&\sqrt{b^2c^2-1+(b^2c^2-1)\sinh^2\frac{x_k^{l+1}-x_k^l-2bc+2\chi_1}{2}}+\sqrt{b^2c^2-1+(b^2c^2-1)\sinh^2\frac{x_{k+1}^{l+1}-x_{k+1}^l-2bc+2\chi_1}{2}}\notag\\
	\sim&\frac{\dif}{\dif\tau}(x_k+x_{k+1})b+\mathcal{O}(b^3).
	\end{align}
	If we divide both sides of \eqref{Dis-gsG1} by $b^2\frac{\dif}{\dif\tau}(x_k+x_{k+1})$ and take $b\rightarrow0$, we arrive at
	\begin{align}
	\frac{\dif}{\dif\tau}(x_{k+1}-x_{k})=c(\cos u_{k+1}-\cos u_k),\label{sm-gsg2}
	\end{align}
	Eqs. \eqref{sm-gsg1}, \eqref{sm-gsg3} and \eqref{sm-gsg2} are the semi-discrete analog of the gsG equation we proposed in \cite{gsg-F}. It can be easily verified that the $\tau$-function and the variable transformations converge to those in the semi-discrete gsG equation with $\nu=-1$
	\begin{align}
	&	u_k(\tau) =\mathrm{i} \ln \frac{\bar{f}_k \bar{g}_k}{f_k g_k},\ 	\phi_k(\tau)=\mathrm{i}\ln \frac{g_k\bar{f}_k }{f_k \bar{g}_k},\\
	&x_k(\tau)=2kac^{-1}+c\tau+\ln\frac{\bar{g}_kg_k}{f_k\bar{f}_k},\ t=c\tau,\\
	&	f_k=\tau_{00}(k),\,g_k=\tau_{01}(k),
	\end{align}
	with
	\begin{align*}
	\tau_{n,m}(\frac{\tau}{2},k)=\left|m_{ij}^{n,m}(k)\right|_{N\times N}=\left|\mathrm{i}\delta_{ij}+\frac{1}{p_i+p_j}\left(-\frac{p_i}{p_j}\right)^n\left(
	\frac{1-cp_i}{1+cp_j}
	\right)^m\left(
	\frac{1-ap_i}{1+ap_j}
	\right)^{-k}e^{\xi_i+\eta_j} \right|_{N\times N},
	\end{align*}
	\begin{align*}
	\xi_i=\frac{\tau}{2p_i}+\xi_{i 0}, \quad \eta_j= \frac{\tau}{2p_j}+\eta_{j 0},
	\end{align*}
	or
	\begin{align}
	\tau_{nm}(\frac{\tau}{2},k)=\left|\phi_{(n+j-1)}^{(i)}(k, l,m)\right|_{1 \leqslant i, j \leqslant N},
	\end{align}
	with
	\begin{align*}
	\phi_{(n)}^{(i)}(k, l,m)=p_i^n\left(1-a p_i\right)^{-k}\left(1-cp_i\right)^me^{\frac{1}{2p_i}\tau+\xi_{i0}} +\mathrm{i}(-p_i)^n\left(1+a p_i\right)^{-k}\left(1+cp_i\right)^me^{-\frac{1}{2p_i}\tau+\eta_{i0}}.
	\end{align*}

	\subsubsection{The semi-discrete gsG equation with $\nu=1$}
	With the continuous limit $b\rightarrow0$ and the similar procedure in \ref{semi-1}, one can obtain the following theorem.
	\begin{theorem}\label{semi-dis1}
		An integrable semi-discrete analogue of the gsG equation with $\nu=1$ is of the form
		\begin{align}
		&	\frac{\dif }{\dif \tau}\left(u_{k+1}-u_k\right)=\tilde{\Delta}_k\sin\frac{u_{k+1}+u_k}{2},\\
		&\frac{\dif \delta_k}{\dif \tau}=\lambda(\cos u_k-\cos u_{k+1}),
		\end{align}
		where the lattice parameter $\tilde{\Delta}_k$ is a function depending on $(k, \tau)$ defined by
		\begin{align}
		\tilde{\Delta}_k=\sqrt{4a^2+4\lambda^2\sin^2\frac{u_{k+1}-u_k}{2}}.
		\end{align}
		Moreover, the $N$-soliton solution is given by
		\begin{align}
		u_k(\tau)& =\mathrm{i} \ln \frac{\bar{f}_k \bar{g}_k}{f_k g_k}, \\
		x_k(\tau)&=2ka\lambda^{-1}-\lambda\tau+\mathrm{i}\ln\frac{\bar{f}_kg_k}{f_k\bar{g}_k},\ t=\lambda\tau,\\
		\varphi_k(\tau)&=\ln \frac{g_k\bar{g}_k }{f_k \bar{f}_k},
		\end{align}
		where $f_k, g_k, \bar{f}_k$ and $\bar{g}_k$ are $\tau$-functions defined by
		\begin{align}
		f_k=\tau_{00}(k),\,g_k=\tau_{01}(k),
		\end{align}
		either with Gram-type determinant
		\begin{align*}
		\tau_{n,m}(\frac{\tau}{2},k)=\left|m_{ij}^{n,m}(k)\right|_{N\times N}=\left|\mathrm{i}\sqrt{\frac{1-\lambda\mathrm{i}p_i}{1+\lambda\mathrm{i}q_j}}\delta_{ij}+\frac{1}{p_i+p_j}\left(-\frac{p_i}{p_j}\right)^n\left(
		\frac{1-\lambda\mathrm{i}p_i}{1+\lambda\mathrm{i}p_j}
		\right)^m\left(
		\frac{1-ap_i}{1+ap_j}
		\right)^{-k}e^{\xi_i+\eta_j} \right|_{N\times N},
		\end{align*}
		and
		\begin{align*}
		\xi_i=\frac{\tau}{2p_i}+\xi_{i 0}, \quad \eta_j= \frac{\tau}{2p_j}+\eta_{j 0},
		\end{align*}
		or with Casorati-type determinant
		\begin{align}
		\tau_{nm}(\frac{\tau}{2},k)=\left|\phi_{(n+j-1)}^{(i)}(k,m)\right|_{1 \leqslant i, j \leqslant N},
		\end{align}
		\begin{align*}
		\phi_{(n)}^{(i)}(k,m)=p_i^n\left(1-a p_i\right)^{-k}\left(1-cp_i\right)^me^{\frac{1}{2p_i}\tau+\xi_{i0}} +\mathrm{i}\left({\frac{1-cp_i}{1+cp_i}}\right)^{\frac{1}{2}}(-p_i)^n\left(1+a p_i\right)^{-k}\left(1+cp_i\right)^me^{-\frac{1}{2p_i}\tau+\eta_{i0}}.
		\end{align*}
	\end{theorem}
	
	\begin{remark}
		Semi-discrete analogue of the generalized sG equation with $\nu=-1$ can be transformed into the case with $\nu=1$ we proposed here through
		\begin{align}
		\phi_k=\mathrm{i}\varphi_{k},\ \tilde{x}_k=\mathrm{i}x_k,\ \tilde{t}=-\mathrm{i}t,\
		c=\lambda\mathrm{i}.
		\end{align}
	\end{remark}
	According to a proof similar to that in \cite{gsg-F}, we know that the continuous limit of this semi-discrete analogue we proposed in this paper is the gsG equation with $\nu=1$. However, their solutions are pretty dissimilar which we will illustrate in Section \ref{fig}.
	
	\subsection{From the semi-discrete 2DTL equation to the semi-discrete gsG equation with $\nu=1$}
	We proposed two integrable semi-discrete gsG equations with $\nu=-1$ from the semi-discrete 2DTL equation in \cite{gsg-F}. In this part, we show that the semi-discrete gsG equation with $\nu=1$ can also be obtained from the semi-discrete 2DTL equation through some reductions and appropriate definitions of discrete hodograph transformation.
	
	We start with the following
	semi-discrete 2DTL equations
	\begin{align}
	&\left(\frac{1}{a} D_{x_{-1}}-1\right) \tau_{n,m}(k+1) \cdot \tau_{n,m}(k)+\tau_{n+1,m}(k+1) \tau_{n-1,m}(k)=0, \label{2dtl-dis1}\\
	&\left(\frac{1}{c}D_{x_{-1}}-1\right) \tau_{n,m}(k) \cdot \tau_{n,m+1}(k)+\tau_{n+1,m}(k) \tau_{n-1,m+1}(k)=0 .\label{dis-bt}
	\end{align}
	Here $a$ is a spatial discrete step. Eqs. \eqref{2dtl-dis1} can also be viewed as Bäcklund transformations for the 2DTL equations in the sense that if $\tau_{n,m}(k)$ is a solution to the 2DTL equations \eqref{2dtl}, so is $\tau_{n,m}(k+1)$, while Eq. \eqref{dis-bt} is the BT linking the solution $\tau_{n,m}(k)$ of Eq. \eqref{2dtl-dis1} to $\tau_{n,m+1}(k)$.
	\begin{lemma}\label{semi-dis-sol}
		The bilinear equations \eqref{2dtl-dis1}-\eqref{dis-bt} admit the following Casorati determinant solutions
		\begin{align}
		\tau_{n,m}\left(x_{-1}, k\right)=\left|\begin{array}{cccc}
		\phi_{n,m}^{(1)}(k) & \phi_{n+1,m}^{(1)}(k) & \cdots & \phi_{n+N-1,m}^{(1)}(k) \\
		\phi_{n,m}^{(2)}(k) & \phi_{n+1,m}^{(2)}(k) & \cdots & \phi_{n+N-1,m}^{(2)}(k) \\
		\cdots & \cdots & \cdots & \cdots \\
		\phi_{n,m}^{(N)}(k) & \phi_{n+1,m}^{(N)}(k) & \cdots & \phi_{n+N-1,m}^{(N)}(k)
		\end{array}\right|,
		\end{align}
		where
		\begin{align}
		& \phi_{n,m}^{(i)}(k)=c_ip_i^n\left(1-cp_i\right)^m\left(1-a p_i\right)^{-k} e^{\xi_i}+d_iq_i^n\left(1-cq_i\right)^m\left(1-a q_i\right)^{-k} e^{\eta_i},
		\end{align}
		with
		$$
		\xi_i=p_i{ }^{-1} x_{-1}+\xi_{i 0}, \quad \eta_i=q_i{ }^{-1} x_{-1}+\eta_{i 0},
		$$
		and Gram-type determinant solutions
		\begin{align}
		\tau_{n,m}(x_{-1},k)=\left|m_{ij}^{n,m}(k)\right|_{N\times N}=\left|c_{ij}+\frac{1}{p_i+q_j}\left(-\frac{p_i}{q_j}\right)^n\left(
		\frac{1-cp_i}{1+cq_j}
		\right)^m\left(
		\frac{1-ap_i}{1+aq_j}
		\right)^{-k}e^{\xi_i+\eta_j} \right|_{N\times N},
		\end{align}
		with
		\begin{align}
		\xi_i=p_i^{-1} x_{-1}+\xi_{i 0}, \quad \eta_j= q_j^{-1} x_{-1}+\eta_{j 0}.
		\end{align}
		Here $p_i, q_i, \xi_{i 0}$ and $\eta_{i 0}$ are arbitrary parameters which can take either real or complex values.
	\end{lemma}
	Applying reductions similar to continuous case,
	then we could have each of the $\tau$ functions satisfies the following relations
	\begin{align}
	\tau_{n+1,0}\Bumpeq \bar{\tau}_{n,1},\,\tau_{n+1,1}\Bumpeq \bar{\tau}_{n,0}.
	\end{align}
	By putting $\tau = 2x_{-1},\, \tau_{00}(k) = f_k,\ \tau_{01}(k)=g_k$, \eqref{2dtl-dis1}-\eqref{dis-bt} can be converted into
	\begin{align}
	&\left(\frac{2}{a}D_\tau-1\right)f_{k+1}\cdot f_k+\bar{g}_{k+1}\bar{g}_k=0,\label{dis-bl1}\\
	&\left(\frac{2}{a}D_\tau-1\right)\bar{f}_{k+1}\cdot \bar{f}_k+{g}_{k+1}{g}_k=0,\\
	&\left(\frac{2}{a}D_\tau-1\right){g}_{k+1}\cdot{g}_k+\bar{f}_{k+1}\cdot \bar{f}_k=0,\\
	&\left(\frac{2}{a}D_\tau-1\right)\bar{g}_{k+1}\cdot\bar{g}_k+{f}_{k+1}\cdot {f}_k=0,\\
	&(-2\mathrm{i}\lambda^{-1}D_\tau-1)f_k\cdot g_k+\bar{f}_k\bar{g}_k=0,\\
	&(2\mathrm{i}\lambda^{-1}D_\tau-1)\bar{f}_k\cdot \bar{g}_k+f_kg_k=0.\label{dis-bl6}
	\end{align}
	Now we can rewrite the bilinear equations \eqref{dis-bl1}-\eqref{dis-bl6} as
	\begin{align}
	& \frac{2}{a}\left(\ln \frac{f_{k+1}}{f_k}\right)_\tau-1=-\frac{\bar{g}_{k+1} \bar{g}_k}{f_{k+1} f_k}, \\
	& \frac{2}{a}\left(\ln \frac{\bar{f}_{k+1}}{\bar{f}_k}\right)_\tau-1=-\frac{g_{k+1} g_k}{\bar{f}_{k+1} \bar{f}_k}, \\
	& \frac{2}{a}\left(\ln \frac{g_{k+1}}{g_k}\right)_\tau-1=-\frac{\bar{f}_{k+1} \bar{f}_k}{g_{k+1} g_k},\\
	& \frac{2}{a}\left(\ln \frac{\bar{g}_{k+1}}{\bar{g}_k}\right)_\tau-1=-\frac{f_{k+1} f_k}{\bar{g}_{k+1} \bar{g}_k},
	\end{align}
	\begin{align}
	2\mathrm{i}\lambda^{-1}\left(\ln \frac{f_k}{g_k}\right)_\tau+1&=\frac{\bar{f}_k \bar{g}_k}{f_k g_k}, \\
	2\mathrm{i}\lambda^{-1}\left(\ln \frac{\bar{f}_k}{\bar{g}_k}\right)_\tau-1&=-\frac{f_k g_k}{\bar{f}_k \bar{g}_k} .
	\end{align}
	Introducing two intermediate variables
	\begin{align}
	\sigma_k(\tau) & =2 \mathrm{i} \ln \frac{\bar{g}_k(s)}{f_k(s)}, \\
	\sigma_k^{\prime}(\tau) & =2 \mathrm{i} \ln \frac{\bar{f}_k(s)}{g_k(s)},
	\end{align}
	one arrives at a pair of semi-discrete sG equation
	\begin{align}
	& \frac{1}{2 a}\left(\sigma_{k+1}-\sigma_k\right)_\tau=\sin \left(\frac{\sigma_{k+1}+\sigma_k}{2}\right), \\
	& \frac{1}{2 a}\left(\sigma_{k+1}^{\prime}-\sigma_k^{\prime}\right)_\tau=\sin \left(\frac{\sigma_{k+1}^{\prime}+\sigma_k^{\prime}}{2}\right) .
	\end{align}
	Next we introduce dependent variable transformations
	\begin{align}
	u_k(\tau)=\frac{1}{2}\left(\sigma_k+\sigma_k^{\prime}\right)=\mathrm{i} \ln \frac{\bar{f}_k \bar{g}_k}{f_k g_k},\ \varphi_k(\tau)=\frac{1}{2\mathrm{i}}\left(\sigma_k-\sigma_k^{\prime}\right)=\ln \frac{g_k\bar{g}_k }{f_k \bar{f}_k},
	\end{align}
	\begin{prop}
		\begin{align}
		& \frac{\dif \varphi_k}{\dif \tau}=\lambda \sin u_k \\
		& a \sinh \left(\frac{\varphi_{k+1}+\varphi_k}{2}\right)=\lambda\sin \left(\frac{u_{k+1}-u_k}{2}\right)
		\end{align}
	\end{prop}
	One can find a similar proof of this proposition in \cite{gsg-F} for details, which we omit here. We define a discrete hodograph transformation
	\begin{align}
	x_k=2ka\lambda^{-1}-\lambda\tau+\mathrm{i}\ln\frac{\bar{f}_kg_k}{f_k\bar{g}_k},
	\end{align}
	then the nonuniform mesh can be derived as
	\begin{align}
	\delta_k=x_{k+1}-x_k=2a\lambda^{-1}+\mathrm{i}\ln\frac{\bar{f}_{k+1}g_{k+1}f_k\bar{g}_k}{f_{k+1}\bar{g}_{k+1}\bar{f}_kg_k}.
	\end{align}
	Taking the derivative with respect to $\tau$ results in
	\begin{align}
	\frac{\dif \delta_k}{\dif \tau}&=\mathrm{i}\left(\ln\frac{\bar{f}_{k+1}}{\bar{g}_{k+1}}+\ln\frac{f_k}{g_k}-\ln\frac{f_{k+1}}{g_{k+1}}-\ln\frac{\bar{f}_k}{\bar{g}_k} \right)_\tau \notag\\
	&=\frac{\lambda}{2}\left(
	\frac{f_k g_k}{\bar{f}_k \bar{g}_k}+\frac{\bar{f}_k \bar{g}_k}{f_k g_k}-\frac{f_{k+1} g_{k+1}}{\bar{f}_{k+1} \bar{g}_{k+1}}-\frac{\bar{f}_{k+1} \bar{g}_{k+1}}{f_{k+1} g_{k+1}}
	\right)\notag\\
	&=\lambda(\cos u_k-\cos u_{k+1}).
	\end{align}
	\begin{align}
	\frac{\dif }{\dif \tau}\left(u_{k+1}-u_k\right)&=\frac{1}{2}\frac{\dif }{\dif \tau}\left(\sigma_{k+1}+\sigma'_{k+1}-\sigma_{k}-\sigma'_{k}\right)\notag\\
	&=2a\sin\frac{u_k+u_{k+1}}{2}\cosh \frac{\varphi_k+\varphi_{k+1}}{2}\notag\\
	&=\sin\frac{u_k+u_{k+1}}{2}\sqrt{4a^2+4\lambda^2\sin^2\frac{u_{k+1}-u_k}{2}}.
	\end{align}
	Therefore, we can propose an integrable
	semi-discrete gsG equation with $\nu=1$ which is just the same as the analog we obtained from the fully discrete gsG equation in Theorem \ref{semi-dis1}.
	
	\subsubsection{Another integrable semi-discrete generalized sine-Gordon equation with $\nu=1$}
	If we define a discrete hodograph transformation
	\begin{align}
	\delta_k=2a\lambda^{-1}\cosh\left(\frac{\varphi_k+\varphi_{k+1}}{2}\right).
	\end{align}
	Then, we have
	\begin{align*}
	\frac{\dif }{\dif \tau}(u_{k+1}-u_k)
	&=\frac{1}{2}\frac{\dif }{\dif \tau}(\sigma_{k+1}+\sigma_{k+1}'-\sigma_k-\sigma_k')\\
	&=a\left[\sin\left(\frac{\sigma_k+\sigma_{k+1}}{2}\right)+\sin\left(\frac{\sigma_k'+\sigma_{k+1}'}{2}\right)\right]\\
	&=2a\sin\left(\frac{u_k+u_{k+1}}{2}\right)\cos \frac{\mathrm{i}(\varphi_k+\varphi_{k+1})}{2}\\
	&=\lambda\delta_k\sin\left(\frac{u_k+u_{k+1}}{2}\right).
	\end{align*}
	On the other hand, taking the derivative with respect to $\delta_k$ results in
	\begin{align*}
	\frac{\dif \delta_k}{\dif \tau}&=a\lambda^{-1}\sinh\left(\frac{\varphi_k+\varphi_{k+1}}{2}\right) (\varphi_k+\varphi_{k+1})_\tau\\
	&=\lambda\sin \left(\frac{u_{k+1}-u_k}{2}\right)(\sin u_k+\sin u_{k+1})\\
	&=\frac{\lambda}{2}\cos\left(\frac{3u_k-u_{k+1}}{2}\right)-\frac{\lambda}{2}\cos\left(\frac{3u_{k+1}-u_k}{2}\right)\\
	&=-\lambda\cos\left(\frac{u_{k+1}-u_k}{2}\right)(\cos u_{k+1}-\cos u_k).
	\end{align*}
	Summarizing the above results, we obtain an alternative integrable semi-discrete analogue of the gsG equation \eqref{gsg1} by the following theorem.
	\begin{theorem}\label{semi-dis2}
		An alternative semi-discrete analogue of the gsG equation \eqref{gsg1} is of the form
		\begin{align}
		&	\frac{\dif }{\dif \tau}(u_{k+1}-u_k)=\lambda\delta_k\sin\left(\frac{u_k+u_{k+1}}{2}\right)\\
		&\frac{\dif \delta_k}{\dif \tau}=\lambda\cos\left(\frac{u_{k+1}-u_k}{2}\right)(\cos u_k-\cos u_{k+1})
		\end{align}
		where the lattice parameter $\delta_k$ is a function depending on $(k, \tau)$ defined by
		\begin{align}
		\delta_k=2a\lambda^{-1}\cosh\left(\frac{\varphi_k+\varphi_{k+1}}{2}\right),
		\end{align}
		and $t=\lambda\tau.$
	\end{theorem}	
	
	\subsection{Reduction to the semi-discrete sG and semi-discrete SP equation}
	Similar to the continuous case and the discrete case, we demonstrate that the semi-discrete gsG equations reduce to the semi-discrete sG equation \cite{sm-sg} and the semi-discrete SP equation \cite{dis-sp} with the variable transformations and the scaling limit.
	
	\subsubsection{Reduction to the semi-discrete sG equation}
	{\bf (I) From the semi-discrete gsG equation with $\nu=-1$ to the semi-discrete sG equation.}
	Similar to the continuous case and the discrete case, we rewrite the elements of the $\tau$-function as
	\begin{align}
	\phi_n^{(i)}(k,1)&=p_i^n(1-ap_i)^{-k}(1-cp_i)e^{\frac{1}{2p_i}\tau+\xi_{i0}}+\mathrm{i}(-p_i)^n(1+ap_i)^{-k}\left(1+cp_i\right)e^{-\frac{1}{2p_i}\tau+\eta_{i0}}\nonumber\\
	&= \phi_n^{(i)}(k,0)-c\phi_{n+1}^{(i)}(k,0).
	\end{align}
	Thus we have
	\begin{align}
	&u_k=2\mathrm{i}\ln \frac{\bar{f}_k}{f_k}+\mathnormal{O}(c),\label{sm-u-sg}\\
	&\phi_k=\mathnormal{O}(c),\label{sm-phi-sg}\\
	&\hat{x}_k=cx_k=2ka+\mathnormal{O}(c^2),\label{sm-x-sg}\\
	&\hat{t}=c^{-1}t=\tau.\label{sm-t-sg}
	\end{align}
	We can develop the semi-discrete analogue of the gsG equation with $\nu=-1$ proposed in \cite{gsg-F} to
	\begin{align}
	&\frac{\dif }{\dif \hat{t}}(u_{k+1}-u_{k})=\sqrt{4a^2-4c^2\sin^2\left(\frac{u_{k+1}-u_k}{2}\right)}\sin\left(\frac{u_{k+1}+u_k}{2}\right),\label{s-sg1}\\
	&\frac{\dif }{\dif \hat{t}}(\hat{x}_{k+1}-\hat{x}_k)=c^2(\cos u_{k+1}-\cos u_k).\label{s-sg2}
	\end{align}
	With the scaling limit $c\rightarrow 0$, equation \eqref{s-sg1}-\eqref{s-sg2} lead to
	\begin{align}
	&\frac{\dif }{\dif \hat{t}}(u_{k+1}-u_{k})=2a\sin\left(\frac{u_{k+1}+u_k}{2}\right),\label{sm-sg}\\
	&\frac{\dif }{\dif \hat{t}}(\hat{x}_{k+1}-\hat{x}_k)=0.
	\end{align}
	Here equation \eqref{sm-sg} is just the semi-discrete sG equation \cite{sm-sg}. Additionally, the solutions also reduce to those for the semi-discrete sG equation \cite{sm-sg1,sm-sg2,sm-sg3}.
	\\
	\\
	{\bf (II) From the semi-discrete gsG equation with $\nu=1$ to the semi-discrete sG equation.}
	In this case, the relations between $\tau$-functions are the same as in the semi-discrete gsG equation with $\nu=-1$, then we have
	\begin{align}
	u_k=2\mathrm{i}\ln \frac{\bar{\hat{f}}_k}{\hat{f}_k}+\mathnormal{O}(\lambda),\
	\phi_k=\mathnormal{O}(\lambda),\
	\hat{x}_k=\lambda x_k=2ka+\mathnormal{O}(\lambda^2),\
	\hat{t}=\lambda^{-1}t=\tau.
	\end{align}
	One can rewrite the semi-discrete analogue of the gsG equation with $\nu=1$ as
	\begin{align}
	&	\frac{\dif }{\dif \hat{t}}\left(u_{k+1}-u_k\right)=\sqrt{4a^2+4\lambda^2\sin^2\frac{u_{k+1}-u_k}{2}}\sin\frac{u_{k+1}+u_k}{2},\\
	&\frac{\dif }{\dif \hat{t}}\left(\hat{x}_{k+1}-\hat{x}_k\right)=\lambda^2(\cos u_k-\cos u_{k+1}).
	\end{align}
	Taking $\lambda\rightarrow 0$, one can also arrive at the semi-discrete sG equation and its determinant solutions.
	
	\subsubsection{Reduction to the semi-discrete SP equation}	
	{\bf (I) From the gsG equation with $\nu=-1$ to the SP equation with $\sigma=-1$.} We have
	\begin{align}
	\phi_n^{(i)}(k,1)&=p_i^n(1-ap_i)^{-k}(1-cp_i)e^{\frac{1}{2p_i}\tau+\xi_{i0}}+\mathrm{i}(-p_i)^n(1+ap_i)^{-k}\left(1+cp_i\right)e^{-\frac{1}{2p_i}\tau+\eta_{i0}}\nonumber\\
	&\propto \phi_{n+1}^{(i)}(k,0)-2\epsilon \frac{\dif}{\dif \tau}\phi_{n+1}^{(i)}(k,0),
	\end{align}
	and
	\begin{align}
	&g_k\propto\bar{f}_k-2\epsilon \frac{\dif}{\dif \tau}\bar{f}_k+\mathnormal{O}(\epsilon^2),\ \bar{g}_k\propto {f}_k-2\epsilon \frac{\dif}{\dif \tau} {f}_k+\mathnormal{O}(\epsilon^2),\label{s-sp-tau}\\
	&\hat{u}_k=\frac{u_k}{\epsilon}=2\mathrm{i}\left(\ln\frac{\bar{f}_k}{f_k}\right)_\tau+\mathnormal{O}(\epsilon),\ \phi_k=2\mathrm{i}\ln\frac{\bar{f}_k}{f_k}+\mathnormal{O}(\epsilon),\label{s-sp-u}\\
	& \hat{x}_k=2ka-2(\ln(f_k\bar{f}_k))_\tau+\mathnormal{O}(\epsilon),\ \hat{t}=\tau.\label{s-sp-x}
	\end{align}
	The semi-discrete gsG equation with $\nu=-1$ can be recast to
	\begin{align}
	&{\epsilon}\frac{\dif }{\dif \hat{t}}(\hat{u}_{k+1}-\hat{u}_{k})=\sqrt{4a^2-\frac{4}{\epsilon^2}\sin^2\left(\frac{\epsilon(\hat{u}_{k+1}-\hat{u}_k)}{2}\right)}\sin\left(\frac{\epsilon(\hat{u}_{k+1}+\hat{u}_k)}{2}\right),\label{s-sp1}\\
	&\epsilon\frac{\dif }{\dif \hat{t}}(\hat{x}_{k+1}-\hat{x}_k)=\frac{1}{\epsilon}(\cos (\epsilon \hat{u}_{k+1})-\cos(\epsilon \hat{u}_k)).\label{s-sp2}
	\end{align}
	Dividing both sides of \eqref{s-sp1}-\eqref{s-sp2} by $\epsilon$ and taking $\epsilon\rightarrow0$ in \eqref{s-sp-tau}-\eqref{s-sp2}, we arrive at
	\begin{align}
	&\frac{\dif }{\dif \hat{t}}(\hat{u}_{k+1}-\hat{u}_{k})=\sqrt{4a^2-(\hat{u}_{k+1}-\hat{u}_{k})^2}\frac{\hat{u}_{k+1}+\hat{u}_k}{2},\\
	&\frac{\dif }{\dif \hat{t}}(\hat{x}_{k+1}-\hat{x}_k)=\frac{\hat{u}_{k}^2-\hat{u}_{k+1}^2}{2},
	\end{align}
	which is nothing but the semi-discrete SP equation proposed in \cite{dis-sp} by taking $X_k=\hat{x}_k/2$.
	\\
	\\
	{\bf (II) From the gsG equation with $\nu=1$ to the SP equation with $\sigma=1$.} With the similar analysis in the cotinuous case and the discrete case, we can construct the semi-discrete SP equation with $\sigma=1$. By defining
	\begin{align}
	\hat{\phi}_n^{(i)}(k)=p_i^{n}(1-ap_i)^{-k} e^{\frac{1}{2p_i}\tau+\xi_{i0}}+(-p_i)^{n}(1+ap_i)^{-k}e^{-\frac{1}{2p_i}\tau+\eta_{i0}},
	\end{align}
	and
	\begin{align}
	\hat{f}_k=\left|\begin{array}{cccc}
	\hat{\phi}_{1}^{(1)}(k) & \hat{\phi}_{2}^{(1)}(k) & \cdots & \hat{\phi}_{N}^{(1)}(k) \\
	\hat{\phi}_{1}^{(2)}(k) & \hat{\phi}_{2}^{(2)}(k) & \cdots & \hat{\phi}_{N}^{(2)}(k) \\
	\cdots & \cdots & \cdots & \cdots \\
	\hat{\phi}_{1}^{(N)}(k) & \hat{\phi}_{2}^{(N)}(k) & \cdots & \hat{\phi}_{n+N}^{(N)}(k)
	\end{array}\right|,\ \hat{g}_k=\left|\begin{array}{cccc}
	\hat{\phi}_0^{(1)}(k) & \hat{\phi}_{1}^{(1)}(k) & \cdots & \hat{\phi}_{N-1}^{(1)}(k) \\
	\hat{\phi}_0^{(2)}(k) & \hat{\phi}_{1}^{(2)}(k) & \cdots & \hat{\phi}_{N-1}^{(2)}(k) \\
	\cdots & \cdots & \cdots & \cdots \\
	\hat{\phi}_0^{(N)}(k) & \hat{\phi}_{1}^{(N)}(k) & \cdots & \hat{\phi}_{N-1}^{(N)}(k)
	\end{array}\right|,
	\end{align}
	we know the expansion of the $\tau$-functions for the semi-discrete gsG equation with $\nu=1$ are
	\begin{align}
	&\phi_n^{(i)}(k,0)\propto \hat{\phi}_{n+1}^{(i)}(k)-\mathrm{i}\epsilon\frac{\dif}{\dif\tau}\hat{\phi}_{n+1}^{(i)}(k)+\mathnormal{O}(\epsilon^2),\\
	&\phi_n^{(i)}(k,1)\propto \hat{\phi}_{n}^{(i)}(k)-\mathrm{i}\epsilon\frac{\dif}{\dif\tau}\hat{\phi}_{n}^{(i)}(k)+\mathnormal{O}(\epsilon^2),\\
	&f_k=\hat{f}_k-{\mathrm{i}\epsilon}\frac{\dif}{\dif \tau}\hat{f}_k+\mathnormal{O}(\epsilon^2),\ \bar{f}_k=\hat{f}_k+{\mathrm{i}\epsilon}\frac{\dif}{\dif \tau}\hat{f}_k+\mathnormal{O}(\epsilon^2),\\ &g_k=\hat{g}_k+{\mathrm{i}\epsilon}\frac{\dif}{\dif \tau}\hat{g}_k+\mathnormal{O}(\epsilon^2),\ \bar{g}_k=\hat{g}_k-{\mathrm{i}\epsilon}\frac{\partial}{\partial \tau}\hat{g}_k+\mathnormal{O}(\epsilon^2).
	\end{align}
	The semi-discrete dependent transformation and the corresponding variables can be defined similar to the continuous case as
	\begin{align}
	\hat{u}_k=2\left(\ln\frac{\hat{g}_k}{\hat{f}_k}\right)_\tau,\
	\varphi_k=2\ln\frac{\hat{g}_k}{\hat{f}_k},\ \hat{x}_k=2ka-2(\ln\hat{f}_k\hat{g}_k)_\tau,\ \hat{t}=\tau,
	\end{align}
	under the scaling limit $\epsilon\rightarrow0$. In addition, the semi-discrete SP equation with $\sigma=1$ can be derived from the semi-discrete gsG equation with $\nu=1$
	\begin{align}
	&\frac{\dif }{\dif \hat{t}}(\hat{u}_{k+1}-\hat{u}_{k})=\sqrt{4a^2+(\hat{u}_{k+1}-\hat{u}_{k})^2}\frac{\hat{u}_{k+1}+\hat{u}_k}{2},\\
	&\frac{\dif }{\dif \hat{t}}(\hat{x}_{k+1}-\hat{x}_k)=\frac{\hat{u}_{k+1}^2-\hat{u}_{k}^2}{2}.
	\end{align}
	
	Similarly to the continuous equation \eqref{sp1}, solutions of the semi-discrete SP equation with $\sigma=1$ develops singularities. Here we take the one-soliton solution as an example. The $\tau$-functions for the one-soliton solutions are given by
	\begin{align}
	&\hat{f}_k=p_1(1-ap_1)^{-k} e^{\frac{1}{2p_1}\tau+\xi_{10}}-p_1(1+ap_1)^{-k}e^{-\frac{1}{2p_1}\tau+\eta_{10}}\propto 1-\left(\frac{1-ap_1}{1+ap_1}\right)^{-k}e^{p_1^{-1}\tau+\xi_{10}-\eta_{10}},\\
	&\hat{g}_k=(1-ap_1)^{-k} e^{\frac{1}{2p_1}\tau+\xi_{10}}+(1+ap_1)^{-k}e^{-\frac{1}{2p_1}\tau+\eta_{10}}\propto 1+\left(\frac{1-ap_1}{1+ap_1}\right)^{-k}e^{p_1^{-1}\tau+\xi_{10}-\eta_{10}}.
	\end{align}
	One can express the one-soliton solution as
	\begin{align}
	&\hat{u}_k=-\frac{2}{p_1}\frac{1}{\sinh \left(p_1^{-1}\tau+\xi_{10}-\eta_{10}+k\ln\frac{1+ap_1}{1-ap_1}\right)},\\
	&\hat{x}_k=2ka-\frac{2}{p_1}\frac{1}{\tanh \left(p_1^{-1}\tau+\xi_{10}-\eta_{10}+k\ln\frac{1+ap_1}{1-ap_1}\right)}-\frac{2}{p_1}.
	\end{align}
	When $|\hat{x}_k|$ tends to infinity, $\hat{u}_k$ diverges. The more detailed description of $N$-soliton solutions will be reported elsewhere.
	\section{One- and two-soliton solutions}\label{fig}
	\subsection{One-soliton solution}
	The $\tau$-functions for the one-soliton solutions of the gsG equation \eqref{gsg1} are given by
	\begin{align}
	&f=e^{\xi_1}+\mathrm{i}\sqrt{\frac{1-\mathrm{i}p_1}{1+\mathrm{i}p_1}}e^{\eta_1}\propto 1+\mathrm{i}\sqrt{\frac{1-\mathrm{i}p}{1+\mathrm{i}p}}e^{-\zeta},\\
	&g=\left(1-\mathrm{i}p_1\right) e^{\xi_{1}}+\mathrm{i}\sqrt{\frac{1-\mathrm{i}p_1}{1+\mathrm{i}p_1}}\left(1+\mathrm{i}p_1\right)e^{\eta_{1}}\propto 1+\mathrm{i}\sqrt{\frac{1+\mathrm{i}p}{1-\mathrm{i}p}}e^{-\zeta},
	\end{align}
	with
	\begin{align}
	\zeta=py+\frac{\tau}{p}+\zeta_0.
	\end{align}
	Here we set $p=p_1$ and $\lambda=1$ for simplicity. Thus, we are able to obtain the parametric form of the one-soliton solution
	\begin{align}
	&u=\mathrm{i}\ln\frac{\bar{f}\bar{g}}{fg}=-2\arctan (\sqrt{1+p^2}\sinh\zeta)+\pi,\\
	&X=x+c_1t+2\arctan p =\frac{\zeta}{p}+2\arctan(p\tanh\zeta),
	\end{align}
	with $c_1=1+\frac{1}{p^2}$ stands for the velocity of the soliton. By setting $p\rightarrow-p$, one can see that the one-soliton solution presented above is equivalent to the solution in \cite{gsg1}. This relates to the fact that if $u$ solves equation \eqref{gsg1}, then so do the functions $\pm u+2\pi n$ ($n$: integer). A detailed analysis of the solution can be seen in \cite{gsg1}, which we omit here.
	
	For the semi-discrete gsG equation with $\nu=1$, the $\tau$-functions are
	\begin{align}
	f_k\propto 1+\mathrm{i}\sqrt{\frac{1-\mathrm{i}p}{1+\mathrm{i}p}}\left(\frac{1-ap}{1+ap}\right)^{k}e^{-\theta},\quad g_k\propto 1+\mathrm{i}\sqrt{\frac{1+\mathrm{i}p}{1-\mathrm{i}p}}\left(\frac{1-ap}{1+ap}\right)^{k}e^{-\theta},
	\end{align}
	with $\theta=p^{-1}\tau+\theta_0$. Then the one-soliton solution can be expressed as
	\begin{align}
	&u_k=-2\arctan(\sqrt{1+p^2}\sinh\zeta(k))+\pi,\\
	&X_k=x_k+c_1t+2\arctan p=2ka+\frac{1}{p^2}t+2\arctan(p\tanh \zeta(k)),
	\end{align}
	where $\zeta(k)=k\ln\frac{1+ap}{1-ap}+\theta$. By taking $a=0.5$, Figure \ref{1-soliton-3d} displays kink and anti-kink solutions for the semi-discrete gsG equation and Figure \ref{1-soliton-fig} shows the solutions under different $p$ values. Figure \ref{sd-1-soliton-fig} compares the semi-discrete gsG equation's kink and anti-kink solutions to those for the gsG equation. Similar to the continuous case, if $p<0(>0)$, the solution $u_k$ represents an kink(anti-kink) solution. The value of $|p|$ has a positive correlation with the amplitude of $v_k\equiv\frac{u_{k+1}-u_k}{\delta_k}$.
	\begin{figure}[H]
		\centering
		\subfigure[]
		{
			\label{kink-3d}
			\includegraphics[width=2.2in]{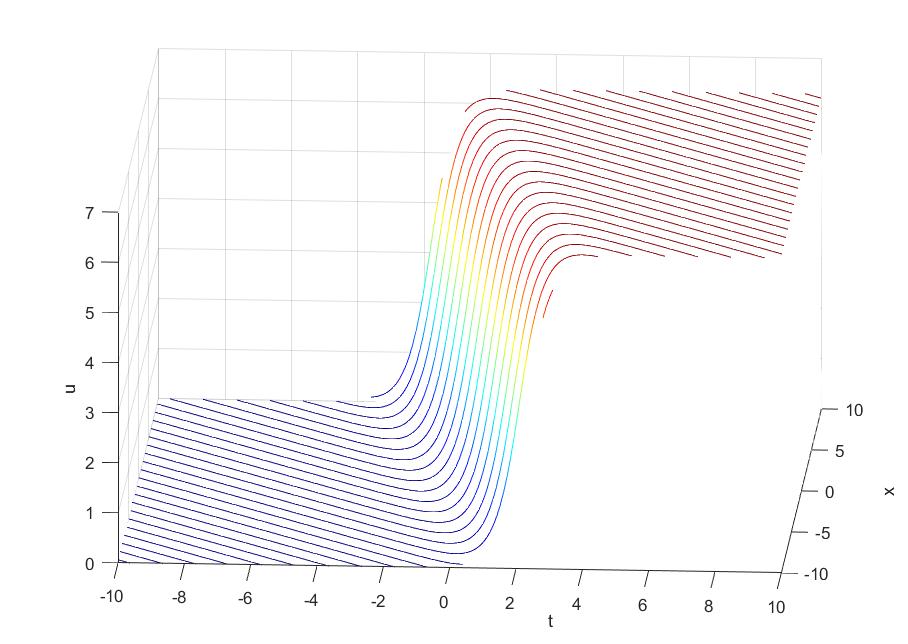}
		}
		\hspace*{3em}
		\subfigure[]
		{\label{antikink-3d}
			\includegraphics[width=2.2in]{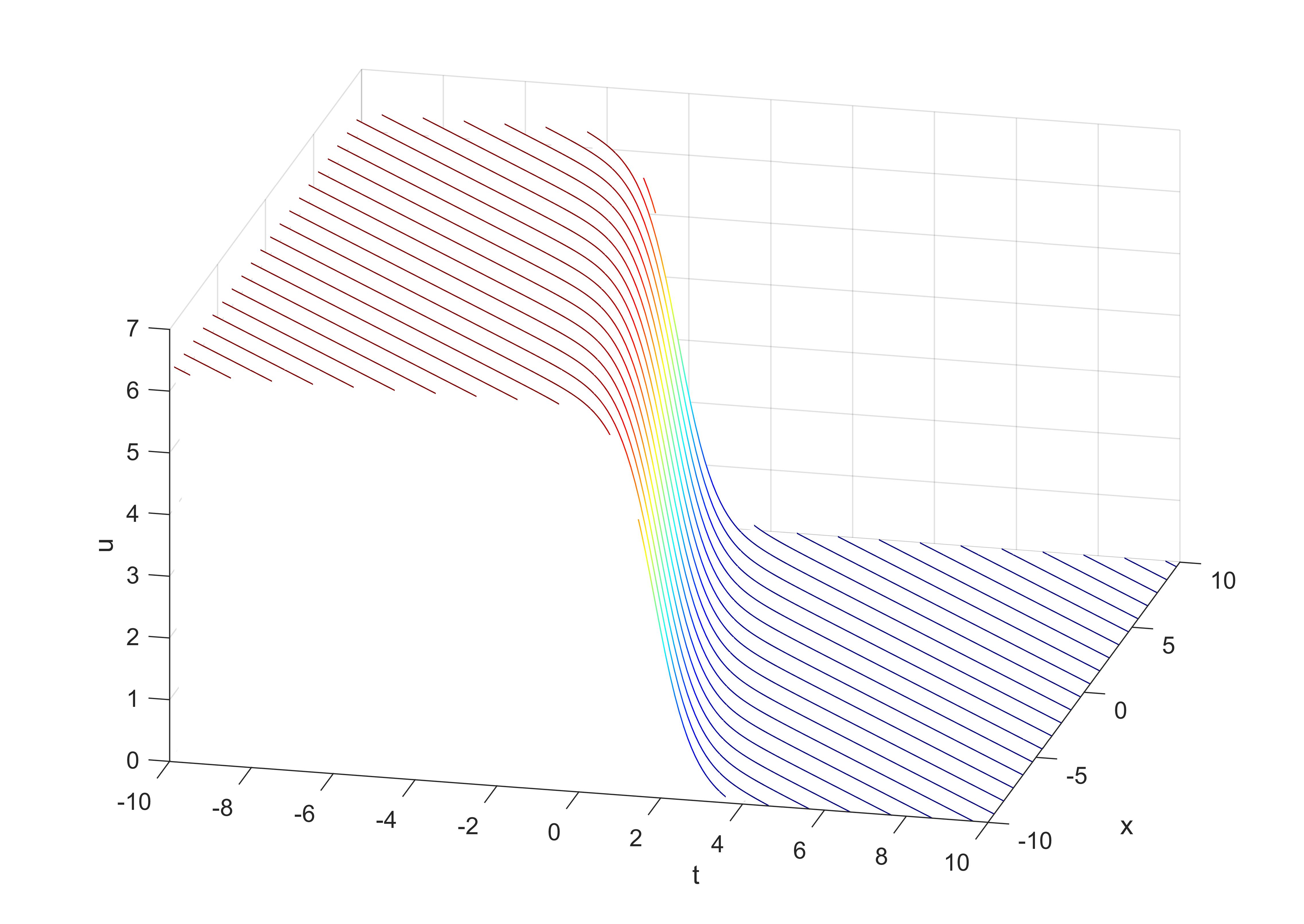}
		}
		\caption{Kink and anti-kink solution $u_k$ for semi-discrete gsG equation. (a)$p=-0.5$, (b)$p=0.5$.}\label{1-soliton-3d}
	\end{figure}
	\begin{figure}[H]
		\centering
		\subfigure[]
		{
			\label{kink-1}
			\includegraphics[width=2.2in]{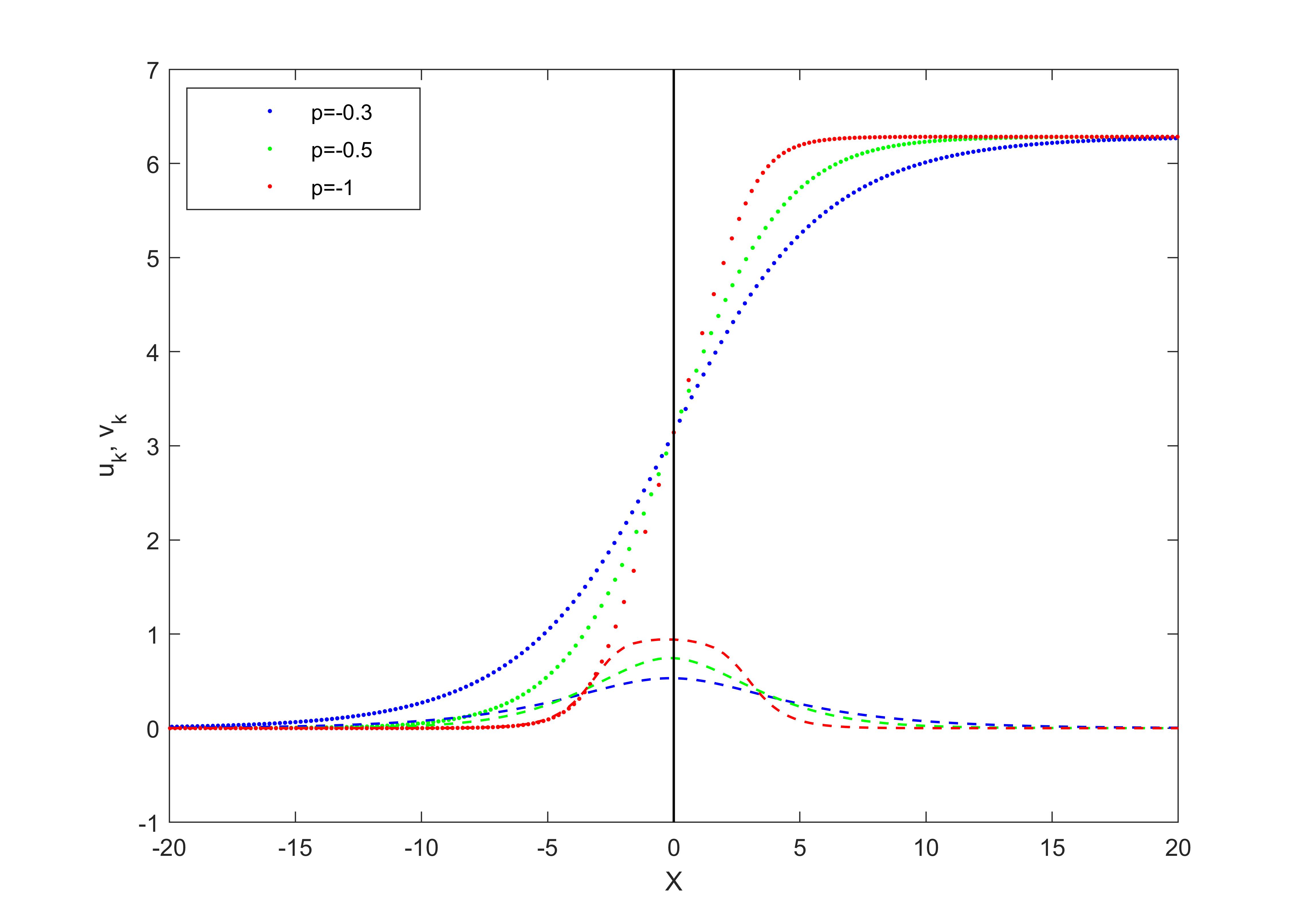}
		}
		\hspace*{3em}
		\subfigure[]
		{\label{antikink-2}
			\includegraphics[width=2.2in]{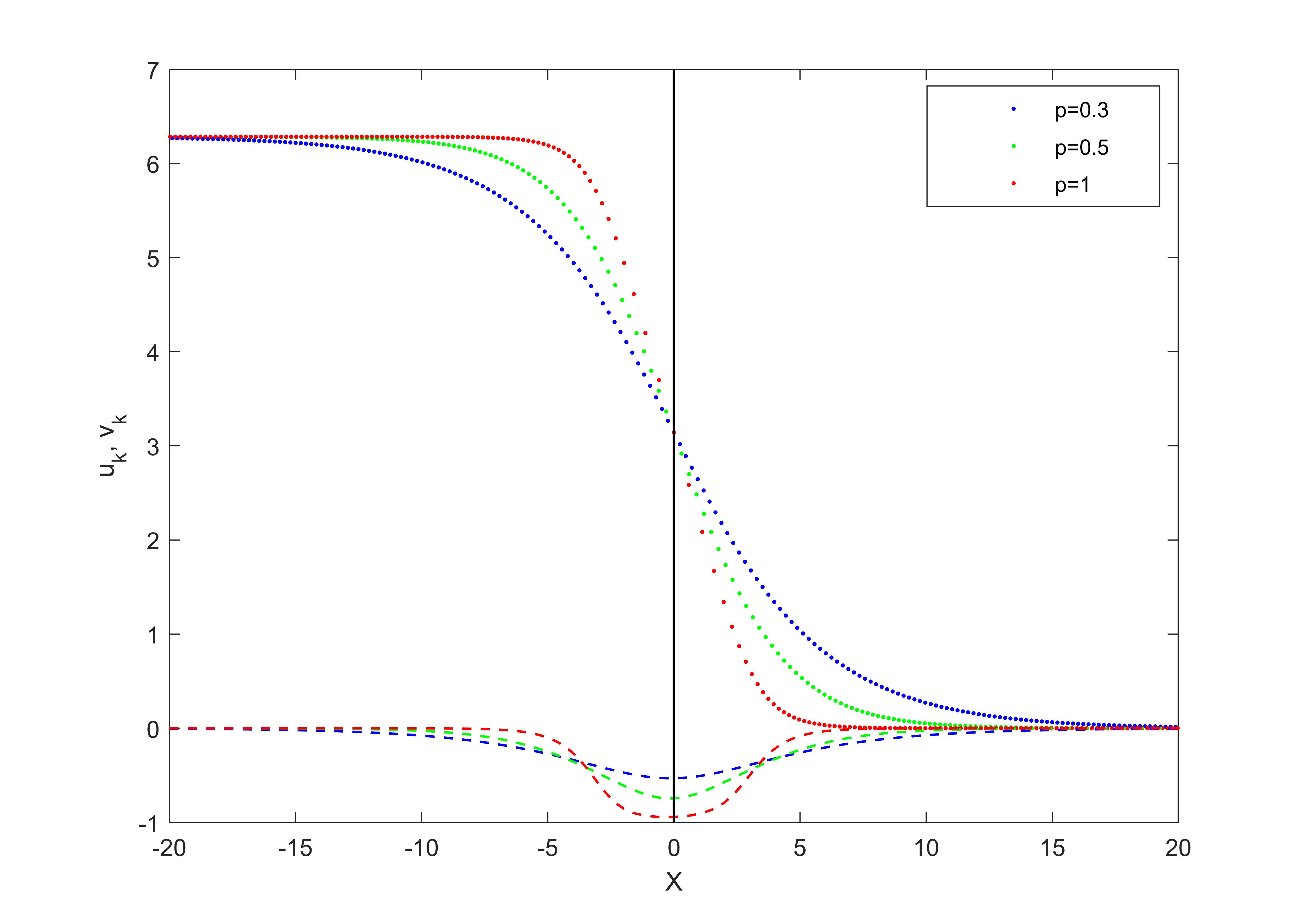}
		}
		\caption{Kink and anti-kink solution $u_k$ for semi-discrete gsG equation and corresponding $v_k\equiv\frac{u_{k+1}-u_k}{\delta_k}$ with different $p$ at $t=0$; dot: $u_k$, dashed line: $v_k$.}\label{1-soliton-fig}
	\end{figure}
	\begin{figure}[H]
		\centering
		\subfigure[]
		{
			\label{sd-kink-1}
			\includegraphics[width=2.2in]{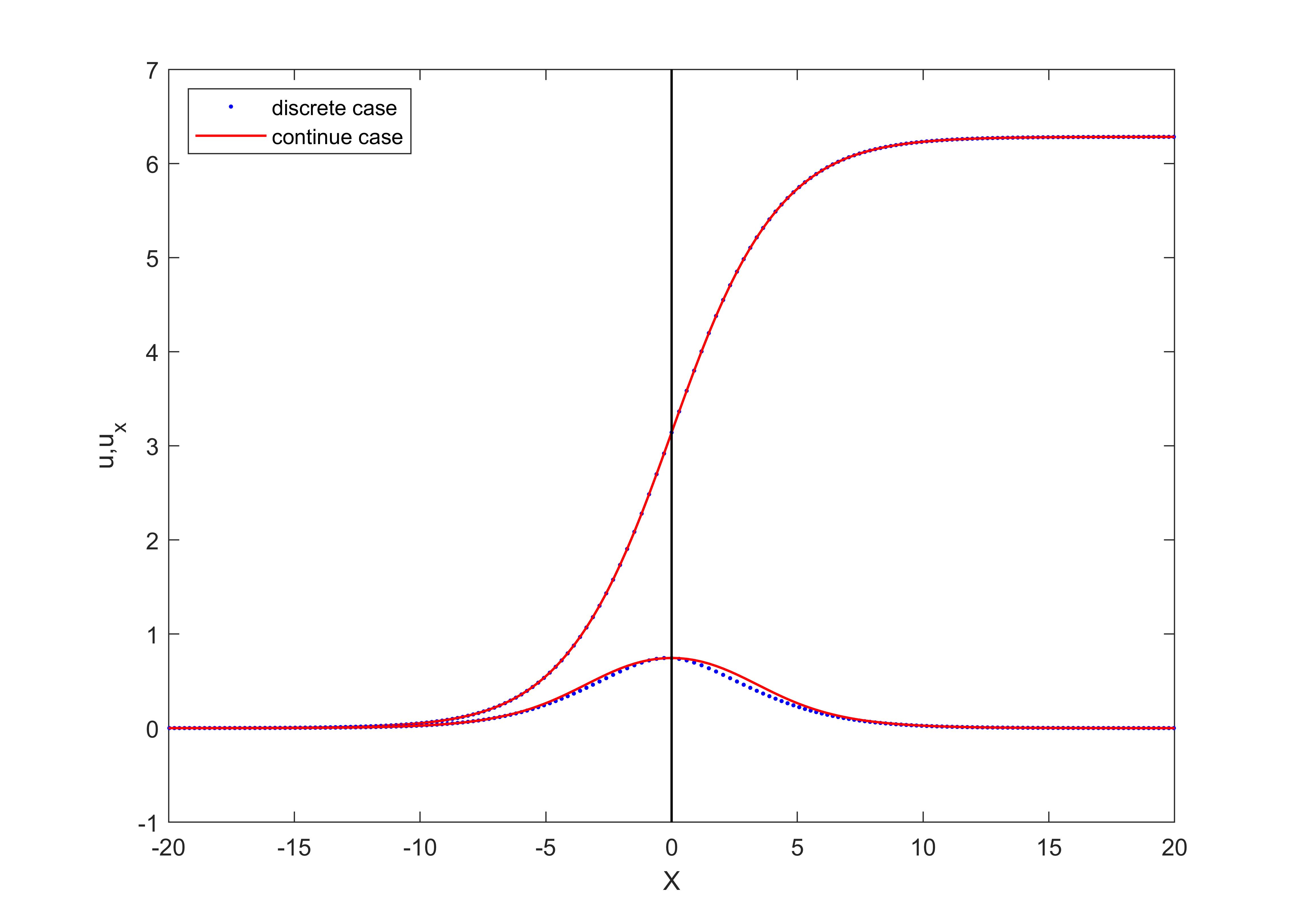}
		}
		\hspace*{3em}
		\subfigure[]
		{\label{sd-antikink-2}
			\includegraphics[width=2.2in]{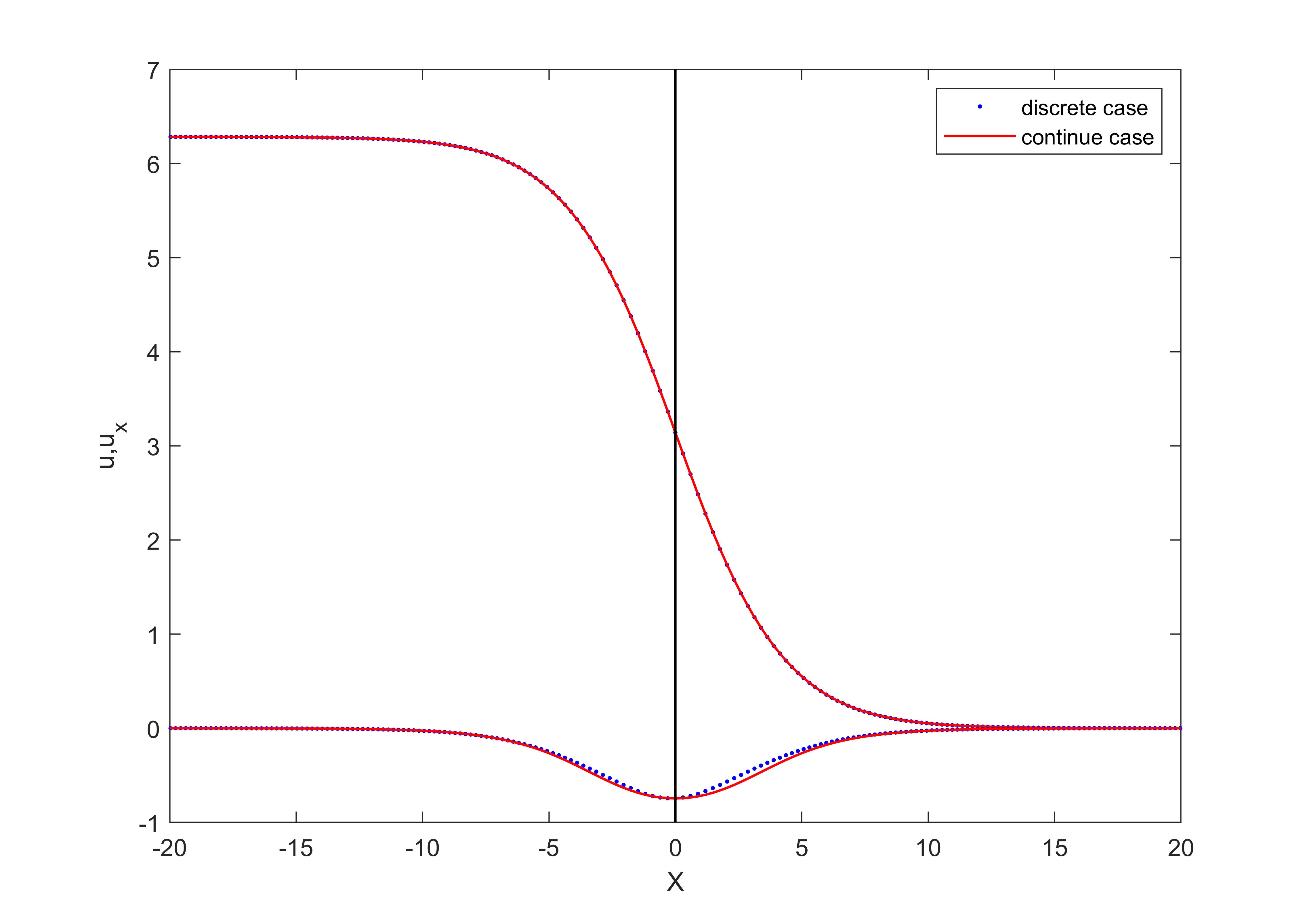}
		}
		\caption{Comparison between the kink and anti-kink solution of the gsG equation(solid line) and the semi-discrete gsG equation(dot) at $t=0$. (a)$p=-0.5$, (b)$p=0.5$.}\label{sd-1-soliton-fig}
	\end{figure}
	
	For the fully discrete gsG equation with $\nu=1$, the $\tau$-functions are
	\begin{align}
	f_k^l\propto 1+\mathrm{i}\sqrt{\frac{1-\mathrm{i}p}{1+\mathrm{i}p}}\left(\frac{1-a p}{1+a p}\right)^{k}\left(\frac{1-b p^{-1}}{1+b p^{-1}}\right)^{l},\ g_k^l\propto 1+\mathrm{i}\sqrt{\frac{1+\mathrm{i}p}{1-\mathrm{i}p}}\left(\frac{1-a p}{1+a p}\right)^{k}\left(\frac{1-b p^{-1}}{1+b p^{-1}}\right)^{l}.
	\end{align}
	Then the one-soliton solution can be expressed as
	\begin{align}
	&u_k^l=-2\arctan(\sqrt{1+p^2}\sinh\zeta(k,l))+\pi,\\
	&x_k=2ka-2lb+2\arctan(p\tanh \zeta(k,l))-2\arctan p,
	\end{align}
	where $\zeta(k,l)=k\ln\frac{1+ap}{1-ap}+l\ln\frac{1+bp^{-1}}{1-bp^{-1}}$. For $a=0.5$ and $b=0.1$, Figure \ref{fd-1-soliton-fig} shows kink and anti-kink solutions for the fully discrete gsG equation with $\nu=1$.
	\begin{figure}[H]
		\centering
		\subfigure[]
		{
			\label{fd-kink-1}
			\includegraphics[width=2.2in]{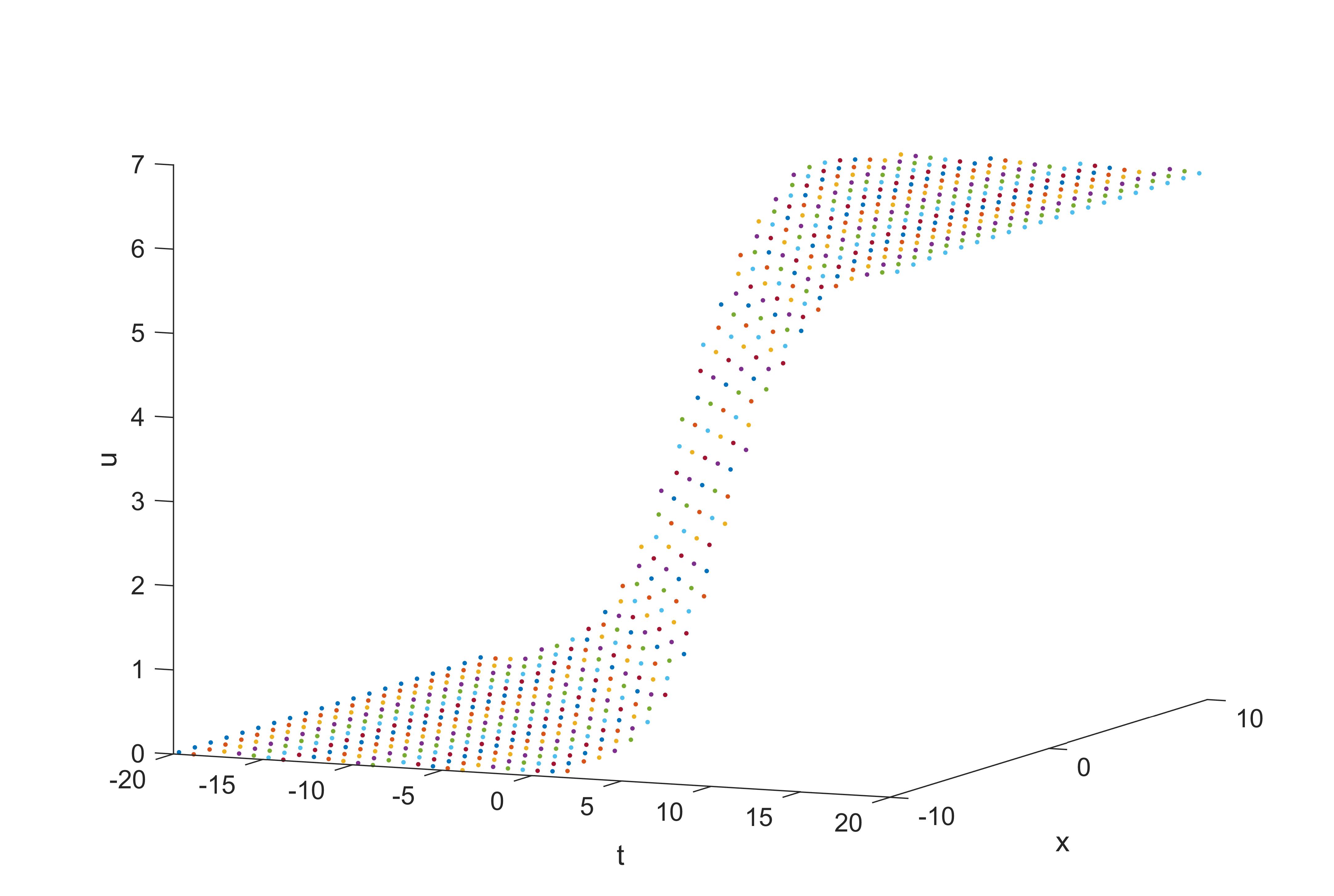}
		}
		\hspace*{3em}
		\subfigure[]
		{\label{fd-antikink-2}
			\includegraphics[width=2.2in]{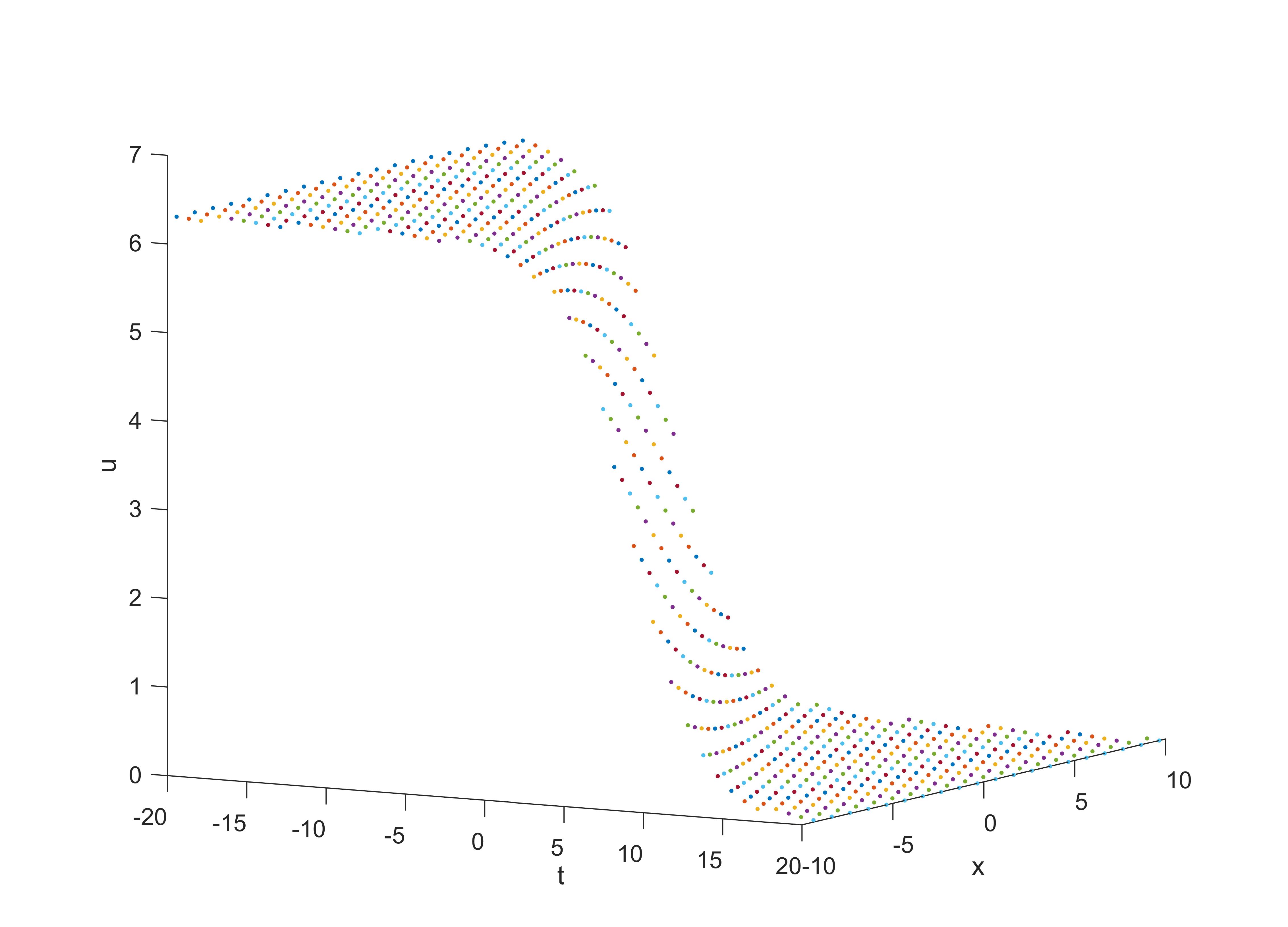}
		}
		\caption{Kink and anti-kink solution $u_k^l$ for fully discrete gsG equation with $\nu=1$. (a)$p=-0.5$, (b)$p=0.5$.}\label{fd-1-soliton-fig}
	\end{figure}
	For $\nu=-1$, one-soliton solutions of the fully discrete gsG equation may be multi-valued. The $\tau$-functions can be expressed as
	\begin{align}
	f_k^l\propto 1+\mathrm{i}\left(\frac{1-a p}{1+a p}\right)^{k}\left(\frac{1-b p^{-1}}{1+b p^{-1}}\right)^{l},\ g_k^l\propto 1+\mathrm{i}\frac{1+p}{1-p}\left(\frac{1-a p}{1+a p}\right)^{k}\left(\frac{1-b p^{-1}}{1+b p^{-1}}\right)^{l}.
	\end{align}
	The one-soliton solution can therefore be written as
	\begin{align}
	&u_k^l=-2\arctan (\sinh\zeta(k,l)-p\cosh\zeta(k,l))+\pi,\\
	&x_k^l=2ka+2lb+\ln \left(\frac{1+p^2}{\left(1-p\right)^2}-\frac{2 p}{\left(1-p\right)^2} \tanh \zeta(k,l)\right),
	\end{align}
	where $\zeta(k,l)=k\ln\frac{1+ap}{1-ap}+l\ln\frac{1+bp^{-1}}{1-bp^{-1}}$. Figure \ref{fd-1-soliton-m1} shows one-soliton solutions to the fully discrete gsG equation with $\nu=-1$ for $a=0.1$ and $b=0.1$. Figure \ref{fd-rkink-m1} illustrates that when $p = -0.4$, $u$ is a single-valued kink solution because $x_{k+1}^l-x_k^l>0$. Figure \ref{fd-irrkink-m1} presents irregular kink solution from \cite{gsg1}, which exhibits a three-valued characteristic at $p= -0.9$, whereas for $p = -1.3$, $u$ becomes loop soliton(see Figure \ref{fd-loop-m1}).
	
	\begin{figure}[H]
		\centering
		\subfigure[regular kink]
		{
			\label{fd-rkink-m1}
			\includegraphics[width=1.5in]{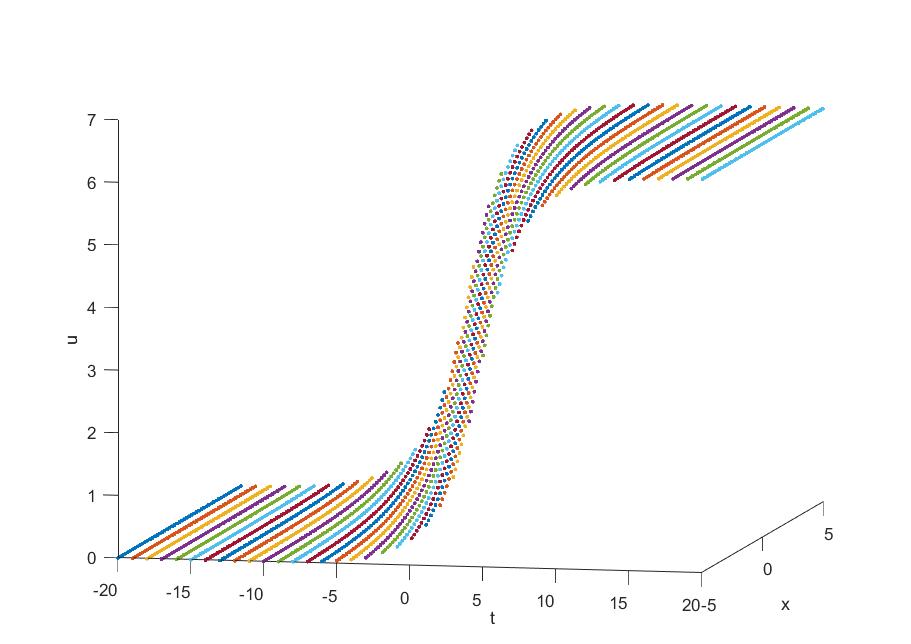}
		}
		\hspace*{3em}
		\subfigure[irregular kink]
		{\label{fd-irrkink-m1}
			\includegraphics[width=1.5in]{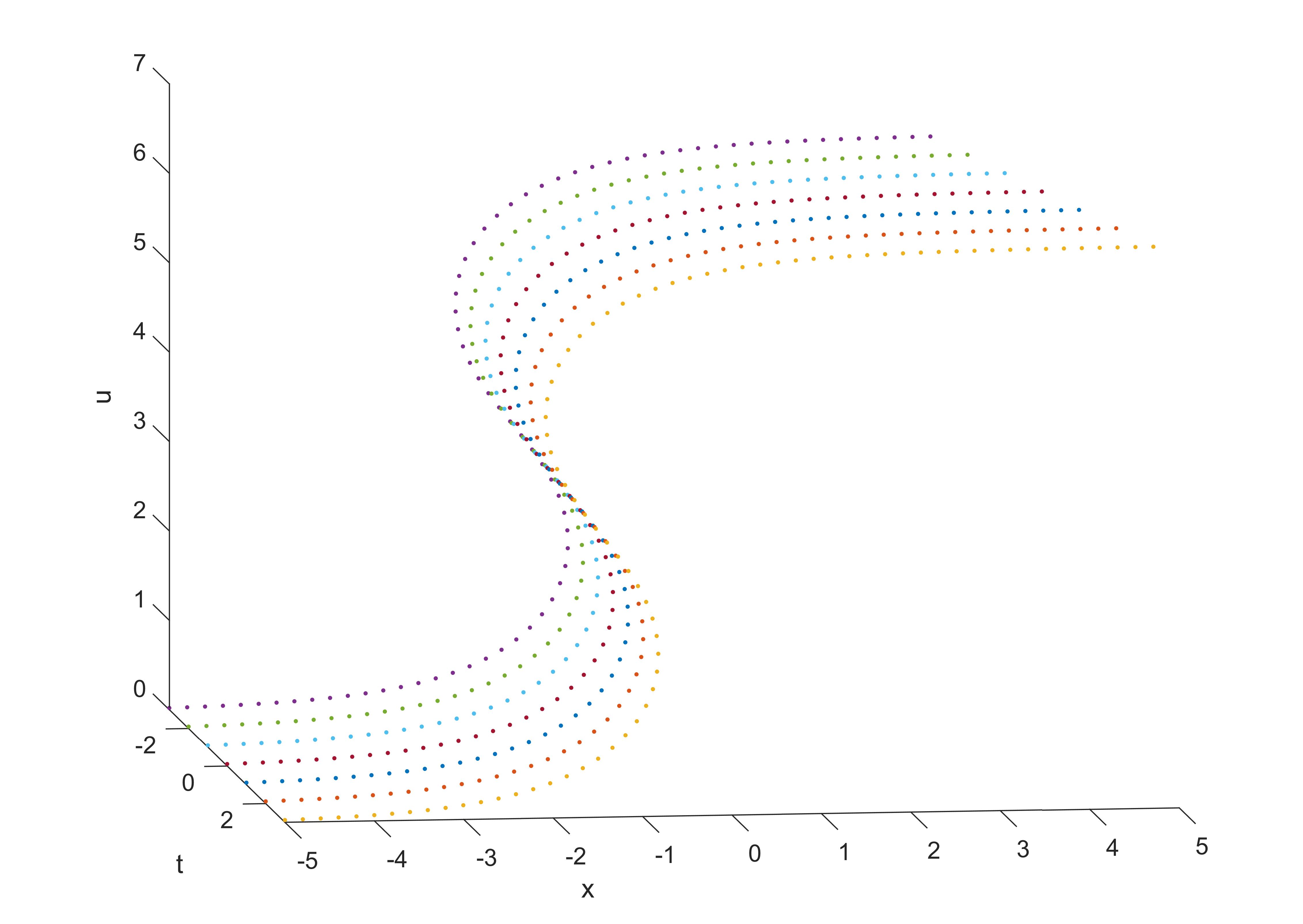}
		}
		\hspace*{3em}
		\subfigure[loop soliton]
		{\label{fd-loop-m1}
			\includegraphics[width=1.5in]{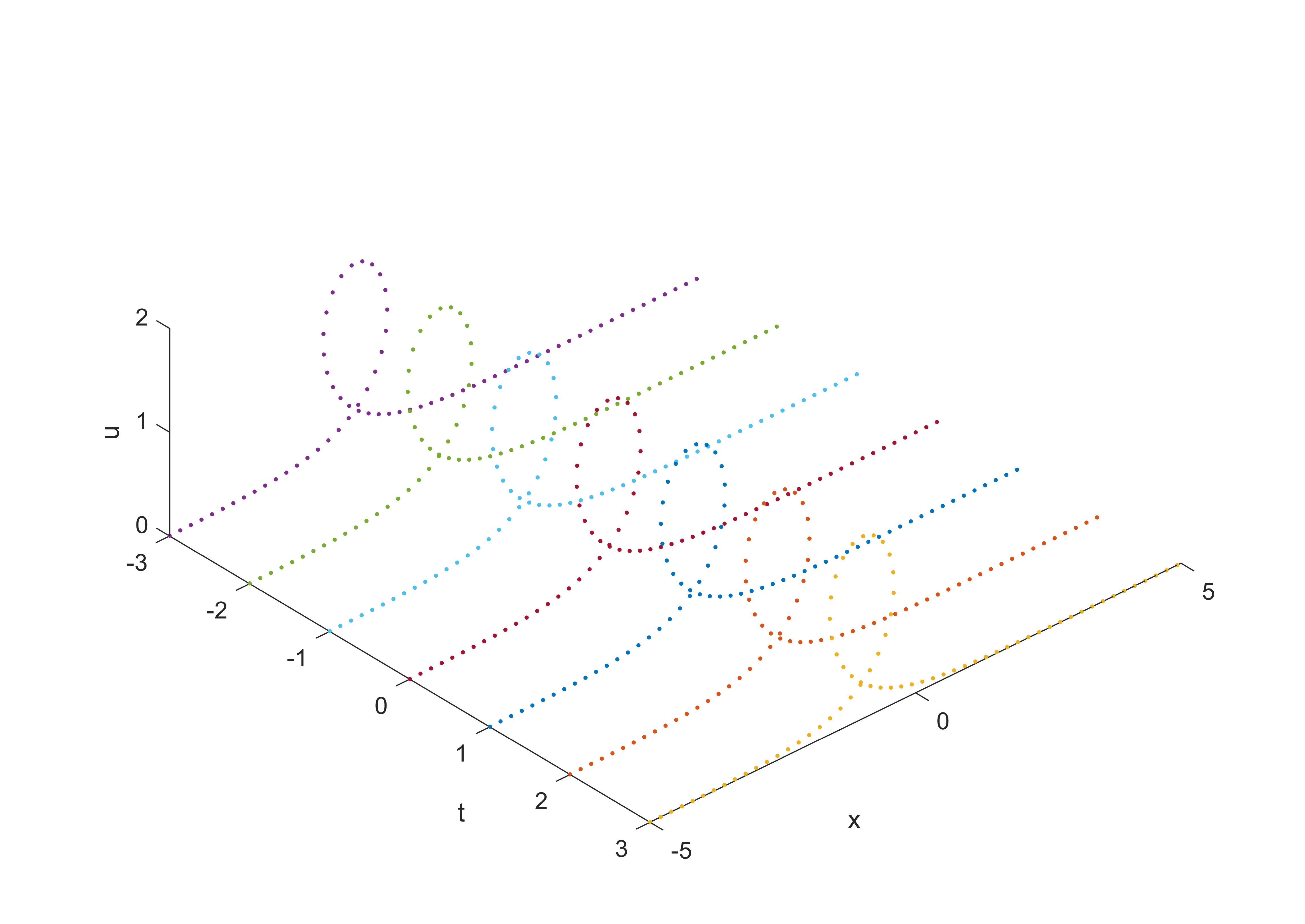}
		}
		\caption{Regular kink, irregular kink and loop soliton solution $u_k^l$ for fully discrete gsG equation with $\nu=-1$. (a)$p=-0.4$, (b)$p=-0.9$, (c)$p=-1.3$.}\label{fd-1-soliton-m1}
	\end{figure}
	
	\subsection{Two-soliton solutions}
	The tau-functions for the two-soliton solutions of the continuous and discrete gsG equations are given by
	
	{\bf(I) the gsG equation with $\nu=1$}
	\begin{align}
	&f\propto 1+\mathrm{i} \frac{p_1+p_2}{p_2-p_1}\sqrt{\frac{1-\mathrm{i}p_1}{1+\mathrm{i}p_1}}e^{-\zeta_1}-\mathrm{i} \frac{p_1+p_2}{p_2-p_1}\sqrt{\frac{1-\mathrm{i}p_2}{1+\mathrm{i}p_2}}e^{-\zeta_2}+\sqrt{\frac{1-\mathrm{i}p_1}{1+\mathrm{i}p_1}}\sqrt{\frac{1-\mathrm{i}p_2}{1+\mathrm{i}p_2}}e^{-\zeta_1-\zeta_2},\\
	&g\propto 1+\mathrm{i} \frac{p_1+p_2}{p_2-p_1}\sqrt{\frac{1+\mathrm{i}p_1}{1-\mathrm{i}p_1}}e^{-\zeta_1}-\mathrm{i} \frac{p_1+p_2}{p_2-p_1}\sqrt{\frac{1+\mathrm{i}p_2}{1-\mathrm{i}p_2}}e^{-\zeta_2}+\sqrt{\frac{1+\mathrm{i}p_1}{1-\mathrm{i}p_1}}\sqrt{\frac{1+\mathrm{i}p_2}{1-\mathrm{i}p_2}}e^{-\zeta_1-\zeta_2},
	\end{align}
	with
	\begin{align}
	\zeta_i=p_iy+\frac{\tau}{p_i}+\zeta_{i0},\ i=1,2,
	\end{align}
	
	{\bf(II) the semi-discrete gsG equation with $\nu=1$}
	\begin{align}
	&f_k\propto 1+\mathrm{i} \frac{p_1+p_2}{p_2-p_1}\sqrt{\frac{1-\mathrm{i}p_1}{1+\mathrm{i}p_1}}z_1^{k}e^{-\theta_1}-\mathrm{i} \frac{p_1+p_2}{p_2-p_1}\sqrt{\frac{1-\mathrm{i}p_2}{1+\mathrm{i}p_2}}z_2^{k}e^{-\theta_2}+\sqrt{\frac{1-\mathrm{i}p_1}{1+\mathrm{i}p_1}}\sqrt{\frac{1-\mathrm{i}p_2}{1+\mathrm{i}p_2}}(z_1z_2)^{k}e^{-\theta_1-\theta_2},\\
	&g_k\propto 1+\mathrm{i} \frac{p_1+p_2}{p_2-p_1}\sqrt{\frac{1+\mathrm{i}p_1}{1-\mathrm{i}p_1}}z_1^{k}e^{-\theta_1}-\mathrm{i} \frac{p_1+p_2}{p_2-p_1}\sqrt{\frac{1+\mathrm{i}p_2}{1-\mathrm{i}p_2}}z_2^{k}e^{-\theta_2}+\sqrt{\frac{1+\mathrm{i}p_1}{1-\mathrm{i}p_1}}\sqrt{\frac{1+\mathrm{i}p_2}{1-\mathrm{i}p_2}}(z_1z_2)^{k}e^{-\theta_1-\zeta_2},
	\end{align}
	with
	\begin{align}
	z_i=\frac{1-ap_i}{1+ap_i},\
	\theta_i=\frac{\tau}{p_i}+\theta_{i0},\ i=1,2,
	\end{align}
	
	{\bf(III) the fully discrete gsG equation with $\nu=1$}
	\begin{align}
	&f_k^l\propto 1+\mathrm{i} \frac{p_1+p_2}{p_2-p_1}\sqrt{\frac{1-\mathrm{i}p_1}{1+\mathrm{i}p_1}}z_1^{k}w_1^l-\mathrm{i} \frac{p_1+p_2}{p_2-p_1}\sqrt{\frac{1-\mathrm{i}p_2}{1+\mathrm{i}p_2}}z_2^{k}w_2^l+\sqrt{\frac{1-\mathrm{i}p_1}{1+\mathrm{i}p_1}}\sqrt{\frac{1-\mathrm{i}p_2}{1+\mathrm{i}p_2}}(z_1z_2)^{k}(w_1w_2)^l,\\
	&g_k^l\propto 1+\mathrm{i} \frac{p_1+p_2}{p_2-p_1}\sqrt{\frac{1+\mathrm{i}p_1}{1-\mathrm{i}p_1}}z_1^{k}w_1^l-\mathrm{i} \frac{p_1+p_2}{p_2-p_1}\sqrt{\frac{1+\mathrm{i}p_2}{1-\mathrm{i}p_2}}z_2^{k}w_2^l+\sqrt{\frac{1+\mathrm{i}p_1}{1-\mathrm{i}p_1}}\sqrt{\frac{1+\mathrm{i}p_2}{1-\mathrm{i}p_2}}(z_1z_2)^{k}(w_1w_2)^l,
	\end{align}
	with
	\begin{align}
	z_i=\frac{1-ap_i}{1+ap_i},\
	w_i=\frac{1-bp_i^{-1}}{1+bp_i^{-1}},\ i=1,2,
	\end{align}
	
	{\bf(4) the fully discrete gsG equation with $\nu=-1$}
	\begin{align}
	&f_k^l\propto 1+\mathrm{i} \frac{p_1+p_2}{p_2-p_1}z_1^{k}w_1^l-\mathrm{i} \frac{p_1+p_2}{p_2-p_1}z_2^{k}w_2^l+(z_1z_2)^{k}(w_1w_2)^l,\\
	&g_k^l\propto 1+\mathrm{i} \frac{p_1+p_2}{p_2-p_1}{\frac{1+p_1}{1-p_1}}z_1^{k}w_1^l-\mathrm{i} \frac{p_1+p_2}{p_2-p_1}{\frac{1+p_2}{1-p_2}}z_2^{k}w_2^l+{\frac{1+p_1}{1-p_1}}{\frac{1+p_2}{1-p_2}}(z_1z_2)^{k}(w_1w_2)^l,
	\end{align}
	with
	\begin{align}
	z_i=\frac{1-ap_i}{1+ap_i},\
	w_i=\frac{1-bp_i^{-1}}{1+bp_i^{-1}},\ i=1,2.
	\end{align}
	
	In \cite{gsg-1,gsg1}, collisions between several types of one-soliton solutions for the continuous gsG equation with $\nu = \pm1$ are shown. Collisions of solutions are similar in the discrete case, thus we omit here. And we demonstrate a different type of solution known as breathers. As pointed out in \cite{gsg-1,gsg1} and Theorem \ref{det-sol}, if we set $p_1=\bar{p}_2$, one can obtain breather solutions. Figure \ref{sd-bre} displays the comparison between the breather solutions of the gsG and the semi-discrete gsG equation with $\nu=1$ for $a=0.1,\ p_1=-0.3+0.5\mathrm{i},\ p_2=-0.3-0.5\mathrm{i}$, in which one can find that  the breather solution of the semi-discrete gsG equation agrees with that of the gsG equation very well. And Figure \ref{fd-bre-fig} shows such kind of breather solutions also appear in fully discrete gsG equations with $\nu=\pm1$.
	
	\begin{figure}[H]
		\centering
		\subfigure[$t=0$]
		{
			\label{sd-bre-1}
			\includegraphics[width=1.5in]{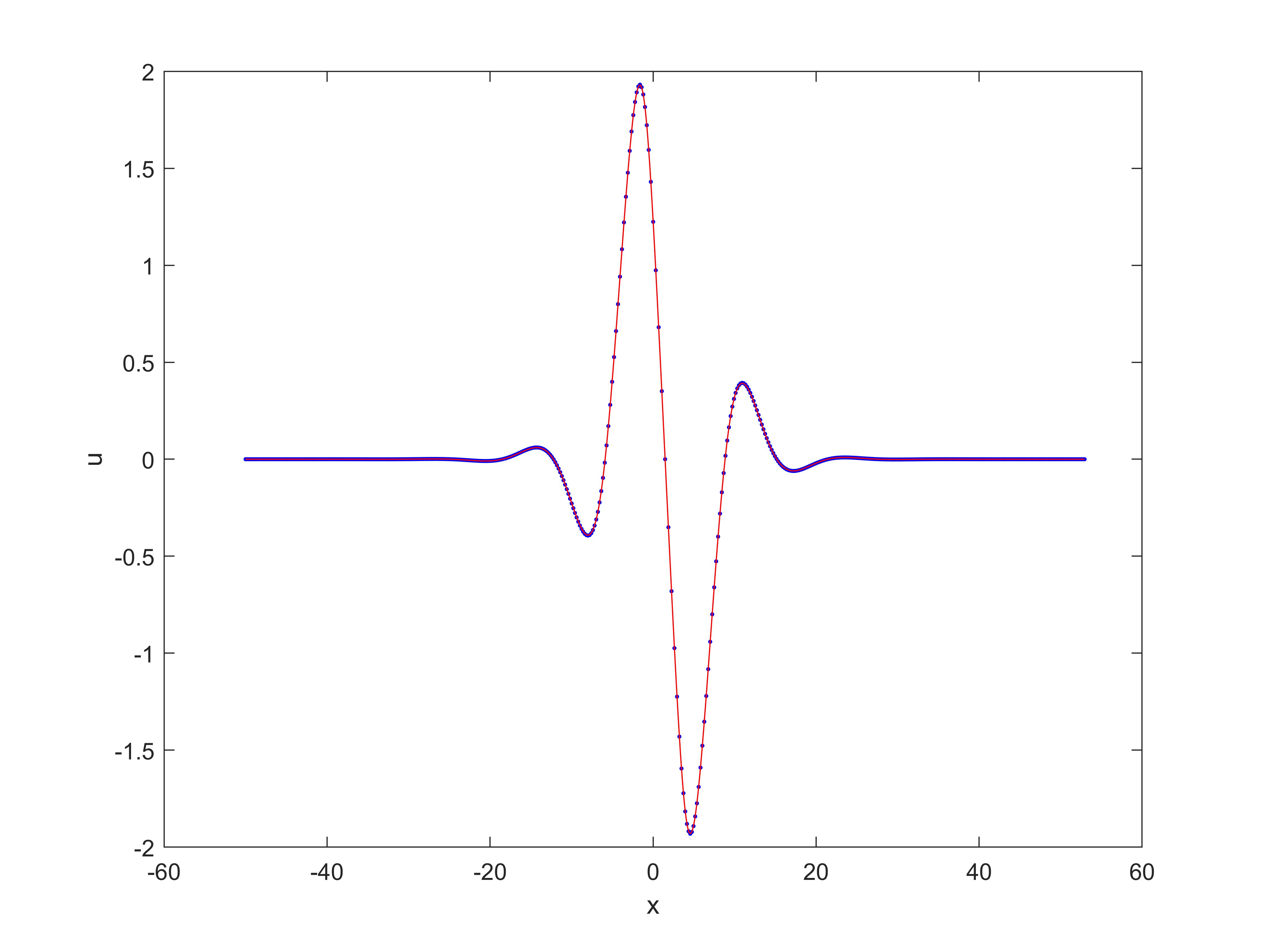}
		}
		\hspace*{3em}
		\subfigure[$t=5$]
		{\label{sd-bre-2}
			\includegraphics[width=1.5in]{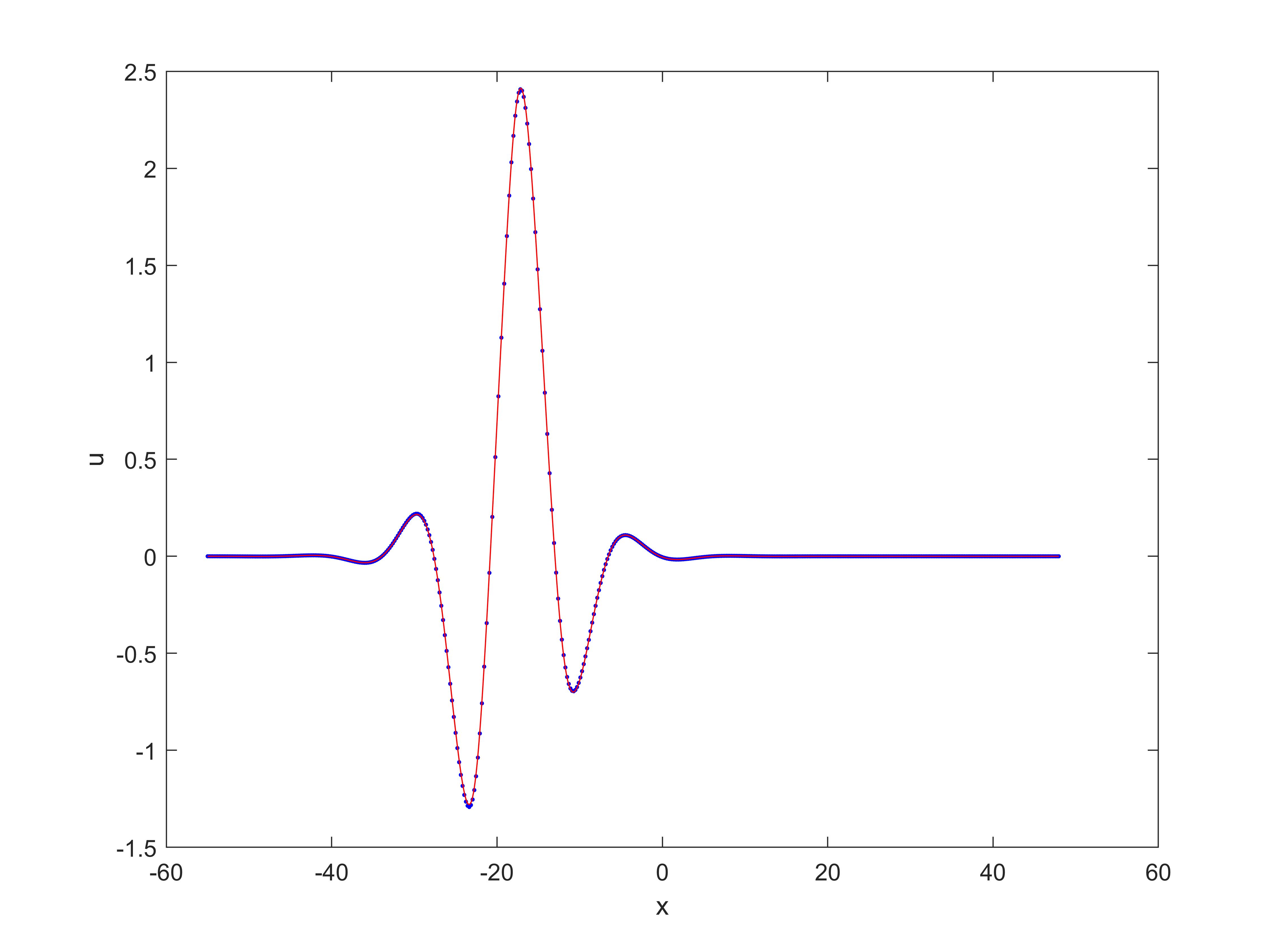}
		}
		\hspace*{3em}
		\subfigure[$t=10$]
		{\label{sd-bre-3}
			\includegraphics[width=1.5in]{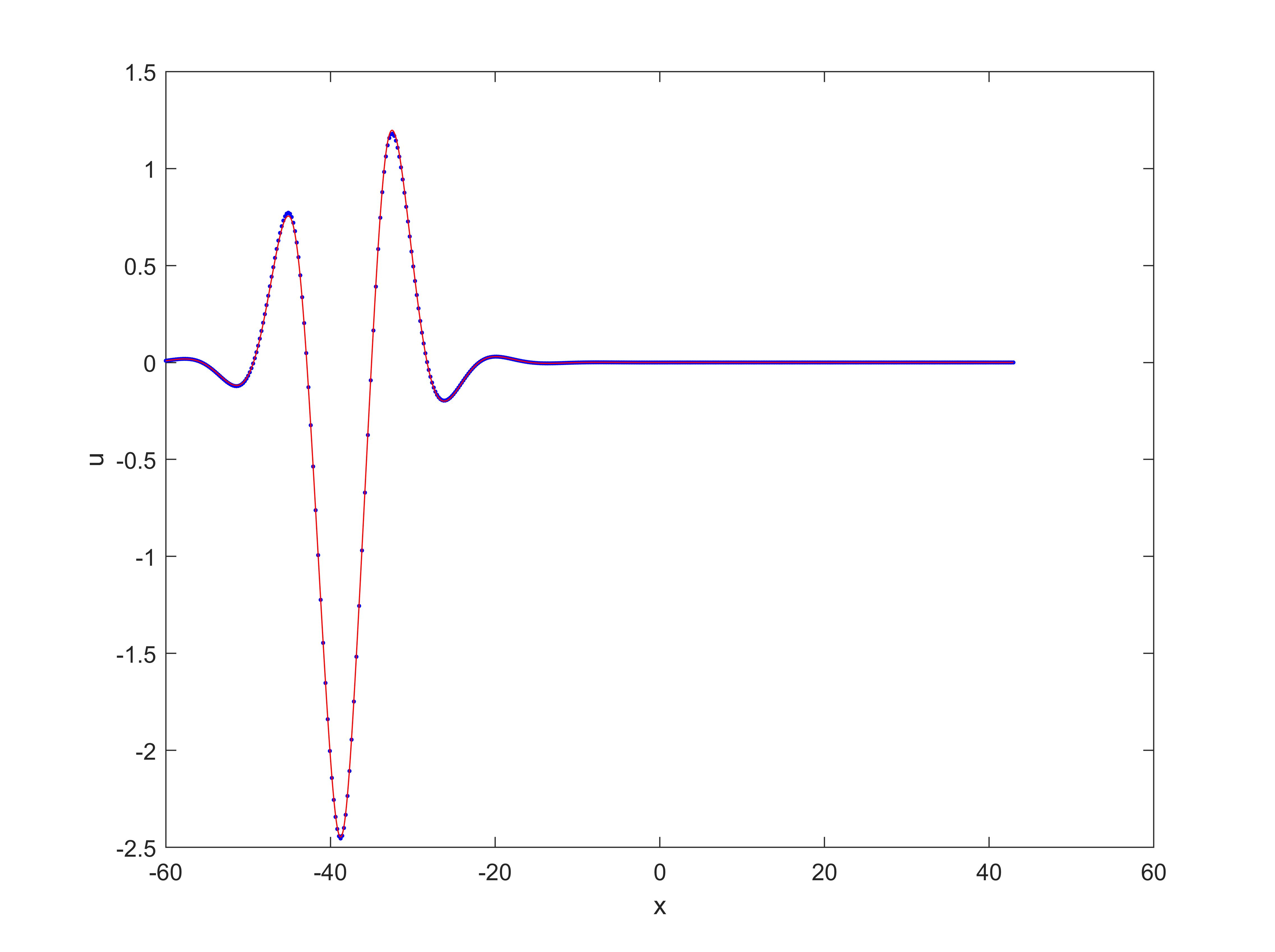}
		}
		\caption{Comparison between the breather solutions of the gsG(solid line) and the semi-discrete gsG equation(dot) with $\nu=1$ for $a=0.1,\ p_1=-0.3+0.5\mathrm{i},\ p_2=-0.3-0.5\mathrm{i}$.}\label{sd-bre}
	\end{figure}
	
	\begin{figure}[H]
		\centering
		\subfigure[]
		{
			\label{fd-bre-1}
			\includegraphics[width=2.2in]{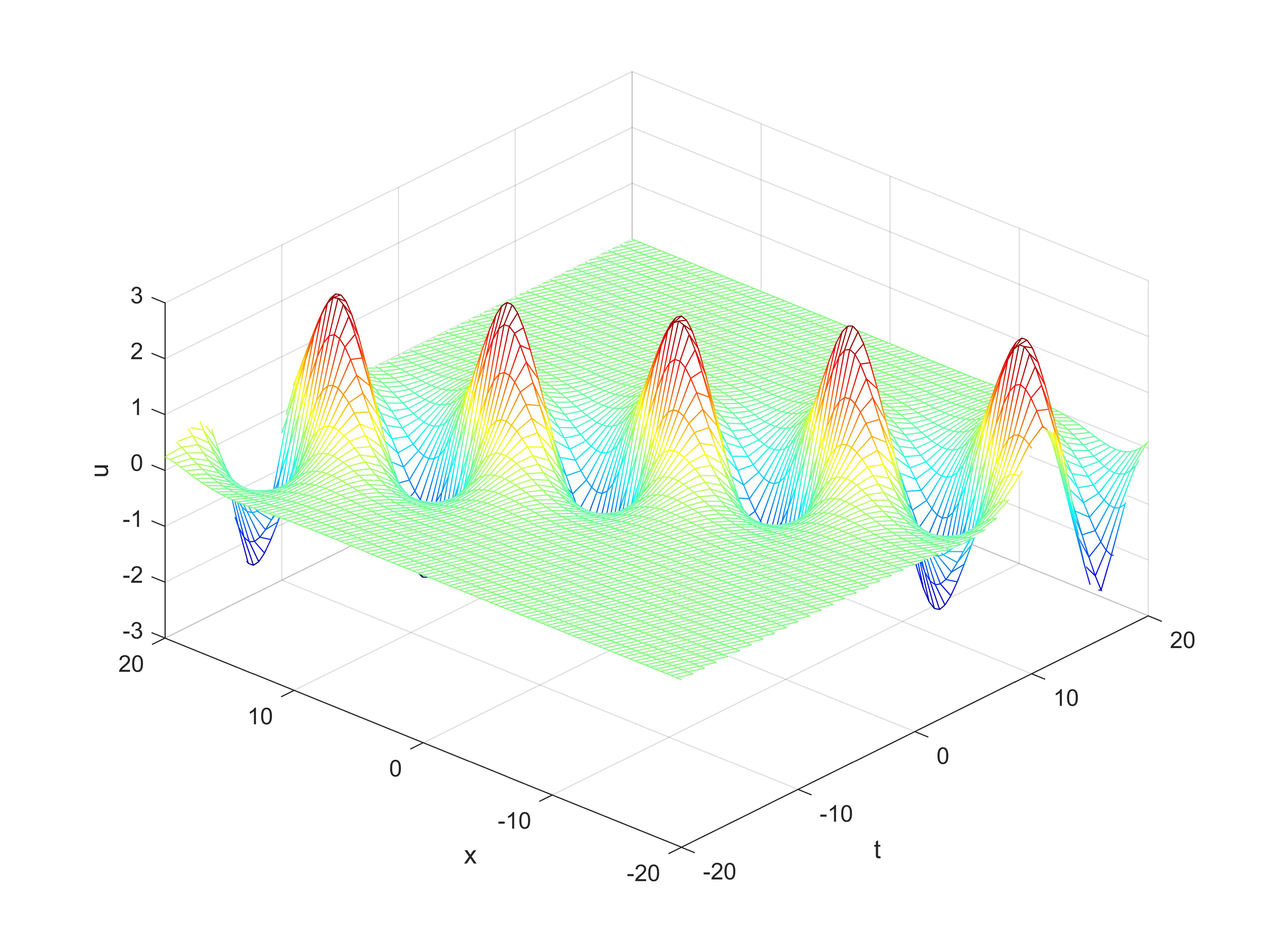}
		}
		\hspace*{3em}
		\subfigure[]
		{\label{fd-bre-2}
			\includegraphics[width=2.2in]{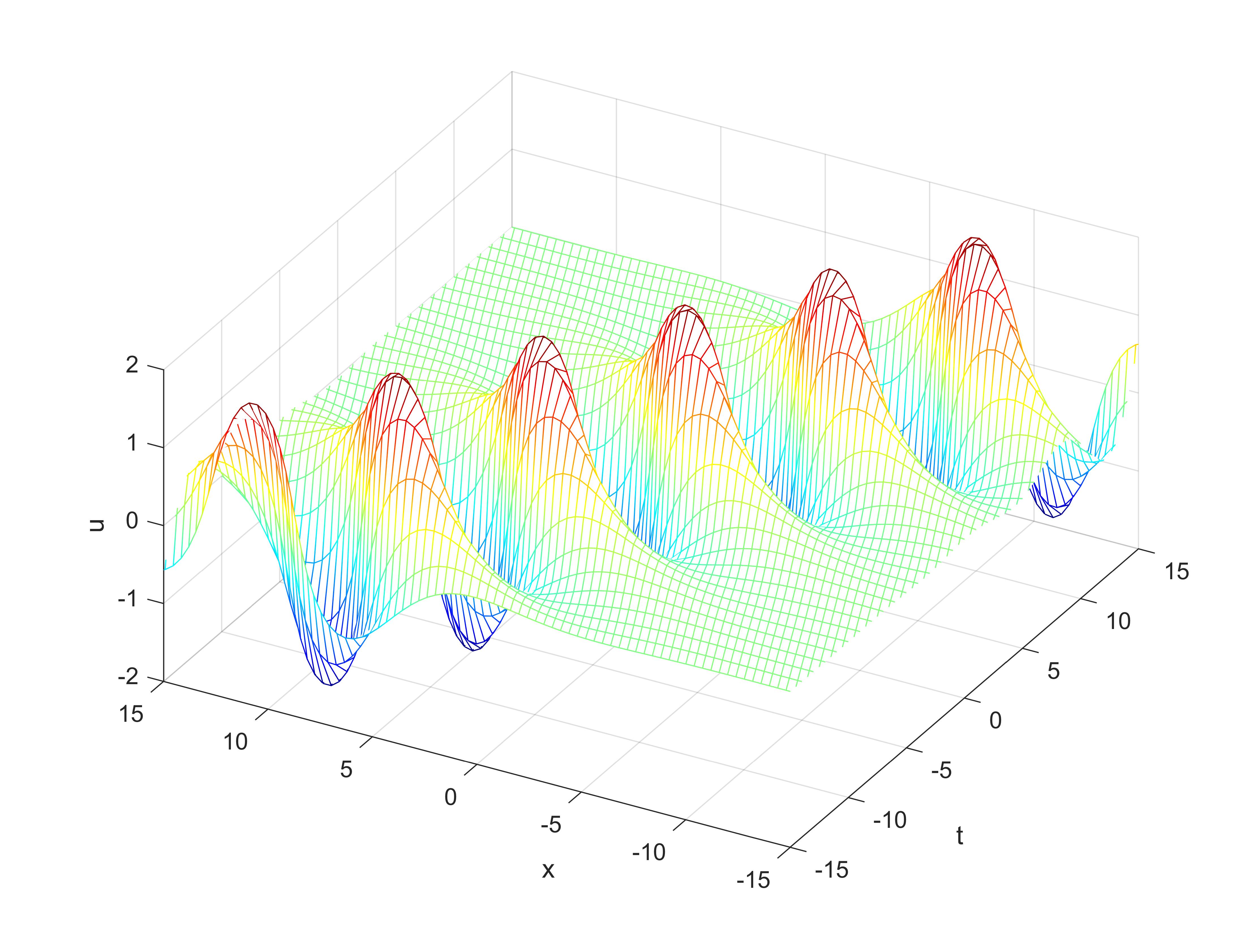}
		}
		\caption{Breather solutions $u_k^l$ for fully discrete gsG equation with $a=b=0.2,\ p_1=-0.3+0.5\mathrm{i},\ p_2=-0.3-0.5\mathrm{i}$. (a)$\nu=1$, (b)$\nu=-1$.}\label{fd-bre-fig}
	\end{figure}
	
	\section{Conclusion}\label{sec6}
	In this paper, we have successfully proposed integrable semi-discrete and full-discrete analogues of a generalized sine-Gordon equation. Determinant formulations for the $N$-soliton solutions, encompassing multi-kink solitons and multi-breather solutions, have been derived for both the continuous and discrete versions of the gsG equations. We have also investigated reductions from the gsG equation to the sG equation and the SP equation, both in continuous and discrete cases. Notably, we have demonstrated the essential role of the B{\"a}cklund transformation of bilinear equations and its parameters in the construction and reduction processes.
However, certain aspects still remain unknown. Firstly, it is crucial to determine the Lax pairs associated with the semi-discrete and full-discrete gsG equations presented in this study. Given the close relation between the bilinear forms of these equations and those of the 2DTL equation, derived from the discrete KP equation, it is natural to explore their connections within the context of Lax pairs. Identification of Lax pairs would enable further investigations into these discrete gsG equations. Recent attention has been focused on multi-component integrable systems, such as the integrable vector sine-Gordon equation \cite{vectorsg,vectorsg1,vectorsg2,vectorsg3} and the multi-component short pulse equation \cite{multisp,Feng1}. Given that the gsG equation lies between the SP equation and the sG equation, it is reasonable to propose multi-component gsG equations by establishing connections with the SP equation and the sG equation. These intriguing questions will be addressed in future studies.

	\section*{Acknowledgement}
	G. Yu is supported by National Natural Science Foundation of
	China (Grant no. 12175155), Shanghai Frontier Research Institute for Modern Analysis and the Fundamental Research Funds for the Central Universities. B.F. Feng's work is supported by the U.S.
	Department of Defense (DoD), Air Force for Scientific Research
	(AFOSR) under grant No. W911NF2010276.
	
	\appendix
	\section{Proof of Theorem \ref{fdis-nu1}}\label{theo-fdis-nu1}
	\begin{proof}
		We can rewrite \eqref{dis-gsg2} and \eqref{dis-gsg3}-\eqref{dis-gsg5} as
		\begin{align}
		\frac{1}{b}\sinh\frac{\varphi_k^{l+1}-\varphi_{k}^l}{2}&=\lambda\sin\frac{u_k^{l+1}+u_{k}^l}{2},\label{Dis-1gsg2}\\
		\frac{1}{b }\cosh\frac{\varphi_k^{l+1}-\varphi_k^l}{2}&=\frac{\sqrt{b^2\lambda^2+1}}{b}\sin\frac{\tilde{x}_k^{l+1}-\tilde{x}_k^l+2b\lambda+2\omega_1}{2},\ \sin\omega_1=\frac{1}{\sqrt{b^2\lambda^2+1}},\label{Dis-1gsg3}\\
		a\sinh\frac{\varphi_k^l+\varphi_{k+1}^l}{2}&=d\sin\frac{u_{k+1}^l-u_k^l}{2},\label{Dis-1gsg4}\\
		a\cosh\frac{\varphi_k^l+\varphi_{k+1}^l}{2}&=\sqrt{a^2+\lambda^2}\sin\frac{\tilde{x}_{k+1}^l-\tilde{x}_k^l-2a\lambda^{-1}+2\omega_2}{2},\ \sin\omega_2=\frac{a}{\sqrt{a^2+\lambda^2}}.\label{Dis-1gsg5}
		\end{align}
		By making a shift of $k\rightarrow k+1$ in \eqref{Dis-1gsg2}, then adding and subtracting it and \eqref{Dis-1gsg2}, one obtains
		\begin{align}
		&\sinh \frac{\varphi_{k+1}^{l+1}-\varphi_{k+1}^l+\varphi_k^{l+1}-\varphi_{k}^l}{4}\cosh \frac{\varphi_{k+1}^{l+1}-\varphi_{k+1}^l-\varphi_k^{l+1}+\varphi_{k}^l}{4}\notag\\
		&=bd\sin \frac{u_{k+1}^{l+1}+u_{k+1}^l+u_k^{l+1}+u_{k}^l}{4}
		\cos \frac{u_{k+1}^{l+1}+u_{k+1}^l-u_k^{l+1}-u_{k}^l}{4},\label{dis-1gsg9}
		\end{align}
		\begin{align}
		&\cosh \frac{\varphi_{k+1}^{l+1}-\varphi_{k+1}^l+\varphi_k^{l+1}-\varphi_{k}^l}{4}\sinh \frac{\varphi_{k+1}^{l+1}-\varphi_{k+1}^l-\varphi_k^{l+1}+\varphi_{k}^l}{4}\notag\\
		&=bd\cos \frac{u_{k+1}^{l+1}+u_{k+1}^l+u_k^{l+1}+u_{k}^l}{4}
		\sin \frac{u_{k+1}^{l+1}+u_{k+1}^l-u_k^{l+1}-u_{k}^l}{4}.\label{Dis-1gsg9}
		\end{align}
		Similarly, from equations \eqref{Dis-1gsg3}-\eqref{Dis-1gsg5}, we arrive at
		\begin{align}
		&\sqrt{1+b^2\lambda^2}\sin \frac{\tilde{x}_{k+1}^{l+1}-\tilde{x}_{k+1}^l+\tilde{x}_k^{l+1}-\tilde{x}_k^l+4b\lambda+4\omega_1}{4}\cos\frac{\tilde{x}_{k+1}^{l+1}-\tilde{x}_{k+1}^l-\tilde{x}_k^{l+1}+\tilde{x}_k^l}{4}\notag\\
		&=\cosh\frac{\varphi_{k+1}^{l+1}-\varphi_{k+1}^l+\varphi_k^{l+1}-\varphi_k^l}{4}\cosh\frac{\varphi_{k+1}^{l+1}-\varphi_{k+1}^l-\varphi_k^{l+1}+\varphi_k^l}{4},\label{dis-1gsg6}
		\end{align}
		\begin{align}
		&\sqrt{1+b^2\lambda^2}\cos \frac{\tilde{x}_{k+1}^{l+1}-\tilde{x}_{k+1}^l+\tilde{x}_k^{l+1}-\tilde{x}_k^l+4b\lambda+4\omega_1}{4}\sin\frac{\tilde{x}_{k+1}^{l+1}-\tilde{x}_{k+1}^l-\tilde{x}_k^{l+1}+\tilde{x}_k^l}{4}\notag\\
		&=\sinh\frac{\varphi_{k+1}^{l+1}-\varphi_{k+1}^l+\varphi_k^{l+1}-\varphi_k^l}{4}\sinh\frac{\varphi_{k+1}^{l+1}-\varphi_{k+1}^l-\varphi_k^{l+1}+\varphi_k^l}{4},\label{Dis-1gsg6}
		\end{align}
		\begin{align}
		&a\sinh \frac{\varphi_{k+1}^{l+1}+\varphi_{k+1}^l+\varphi_k^{l+1}+\varphi_k^l}{4}\cosh \frac{\varphi_{k+1}^{l+1}-\varphi_{k+1}^l+\varphi_k^{l+1}-\varphi_k^l}{4}\notag\\
		&=d\sin \frac{u_{k+1}^{l+1}+u_{k+1}^l-u_k^{l+1}-u_{k}^l}{4}\cos \frac{u_{k+1}^{l+1}-u_{k+1}^l-u_k^{l+1}+u_{k}^l}{4}, \label{dis-1gsg10}
		\end{align}
		\begin{align}
		&a\cosh \frac{\varphi_{k+1}^{l+1}+\varphi_{k+1}^l+\varphi_k^{l+1}+\varphi_k^l}{4}\sinh \frac{\varphi_{k+1}^{l+1}-\varphi_{k+1}^l+\varphi_k^{l+1}-\varphi_k^l}{4}\notag\\
		&=d\cos \frac{u_{k+1}^{l+1}+u_{k+1}^l-u_k^{l+1}-u_{k}^l}{4}\sin \frac{u_{k+1}^{l+1}-u_{k+1}^l-u_k^{l+1}+u_{k}^l}{4}, \label{Dis-1gsg10}
		\end{align}
		and
		\begin{align}
		&\sqrt{a^2+\lambda^2}\sin\frac{\tilde{x}_{k+1}^{l+1}-\tilde{x}_k^{l+1}+\tilde{x}_{k+1}^l-\tilde{x}_k^l-4a\lambda^{-1}+4\omega_2}{4}\cos\frac{\tilde{x}_{k+1}^{l+1}-\tilde{x}_{k+1}^l-\tilde{x}_k^{l+1}+\tilde{x}_k^l}{4}\notag\\
		&=a\cosh\frac{\varphi_{k+1}^{l+1}+\varphi_{k+1}^l+\varphi_k^{l+1}+\varphi_k^l}{4}\cosh\frac{\varphi_{k+1}^{l+1}-\varphi_{k+1}^l+\varphi_k^{l+1}-\varphi_k^l}{4},\label{dis-1gsg7}
		\end{align}
		\begin{align}
		&\sqrt{a^2+\lambda^2}\cos\frac{\tilde{x}_{k+1}^{l+1}-\tilde{x}_k^{l+1}+\tilde{x}_{k+1}^l-\tilde{x}_k^l-4a\lambda^{-1}+4\omega_2}{4}\sin\frac{\tilde{x}_{k+1}^{l+1}-\tilde{x}_{k+1}^l-\tilde{x}_k^{l+1}+\tilde{x}_k^l}{4}\notag\\
		&=a\sinh\frac{\varphi_{k+1}^{l+1}+\varphi_{k+1}^l+\varphi_k^{l+1}+\varphi_k^l}{4}\sinh\frac{\varphi_{k+1}^{l+1}-\varphi_{k+1}^l+\varphi_k^{l+1}-\varphi_k^l}{4},\label{Dis-1gsg7}
		\end{align}
		respectively. Note that equation \eqref{dis-1gsg9} and \eqref{Dis-1gsg10} give
		\begin{align}
		&\frac{1}{ab} \sin \frac{u_{k+1}^{l+1}-u_{k+1}^l-u_k^{l+1}+u_k^l}{4}\cosh\frac{\varphi_{k+1}^{l+1}-\varphi_{k+1}^l-\varphi_k^{l+1}+\varphi_k^l}{4}\notag \\
		&=\sin \frac{u_{k+1}^{l+1}+u_{k+1}^l+u_k^{l+1}+u_k^l}{4}\cosh\frac{\varphi_{k+1}^{l+1}+\varphi_{k+1}^l+\varphi_k^{l+1}+\varphi_k^l}{4} .\label{dis-1gsg1}
		\end{align}
		Equations \eqref{dis-1gsg6} and \eqref{dis-1gsg7} lead to
		\begin{align}
		&a\sqrt{1+b^2\lambda^2}\sin \frac{\tilde{x}_{k+1}^{l+1}-\tilde{x}_{k+1}^l+\tilde{x}_k^{l+1}-\tilde{x}_k^l+4b\lambda+4\omega_1}{4}\cosh\frac{\varphi_{k+1}^{l+1}+\varphi_{k+1}^l+\varphi_k^{l+1}+\varphi_k^l}{4}\notag\\
		&=\sqrt{a^2+\lambda^2}\sin\frac{\tilde{x}_{k+1}^{l+1}-\tilde{x}_k^{l+1}+\tilde{x}_{k+1}^l-\tilde{x}_k^l-4a\lambda^{-1}+4\omega_2}{4}\cosh\frac{\varphi_{k+1}^{l+1}-\varphi_{k+1}^l-\varphi_k^{l+1}+\varphi_k^l}{4},\label{dis-1gsg8}
		\end{align}
		Substituting \eqref{dis-1gsg8} into \eqref{dis-1gsg1}, we have
		\begin{align}
		\frac{1}{b} \sin \frac{u_{k+1}^{l+1}-u_{k+1}^l-u_k^{l+1}+u_k^l}{4}
		=\tilde{\Delta}_k^l\sin \frac{u_{k+1}^{l+1}+u_{k+1}^l+u_k^{l+1}+u_k^l}{4},\label{fdis-1gsG1}
		\end{align}
		with
		\begin{align}
		\tilde{\Delta}_k^l=\frac{\sqrt{a^2+\lambda^2}\sin\frac{\tilde{x}_{k+1}^{l+1}-\tilde{x}_k^{l+1}+\tilde{x}_{k+1}^l-\tilde{x}_k^l-4a\lambda^{-1}+4\omega_2}{4}}{\sqrt{1+b^2\lambda^2}\sin \frac{\tilde{x}_{k+1}^{l+1}-\tilde{x}_{k+1}^l+\tilde{x}_k^{l+1}-\tilde{x}_k^l+4b\lambda+4\omega_1}{4}}.\label{fdis-1gsG2}
		\end{align}
		Subsequently,  by multiplying equation \eqref{dis-1gsg9} and \eqref{Dis-1gsg9}, \eqref{dis-1gsg6} and \eqref{Dis-1gsg6}, we obtain
		\begin{align}
		&b^2\lambda^2\sin \frac{u_{k+1}^{l+1}+u_{k+1}^l+u_k^{l+1}+u_{k}^l}{2}
		\sin \frac{u_{k+1}^{l+1}+u_{k+1}^l-u_k^{l+1}-u_{k}^l}{2}\notag\\
		&=\sinh \frac{\varphi_{k+1}^{l+1}-\varphi_{k+1}^l+\varphi_k^{l+1}-\varphi_{k}^l}{2}\sinh \frac{\varphi_{k+1}^{l+1}-\varphi_{k+1}^l-\varphi_k^{l+1}+\varphi_{k}^l}{2},
		\end{align}
		\begin{align}
		&({1+b^2\lambda^2})\sin \frac{\tilde{x}_{k+1}^{l+1}-\tilde{x}_{k+1}^l+\tilde{x}_k^{l+1}-\tilde{x}_k^l+4b\lambda+4\omega_1}{2}\sin\frac{\tilde{x}_{k+1}^{l+1}-\tilde{x}_{k+1}^l-\tilde{x}_k^{l+1}+\tilde{x}_k^l}{2}\notag\\
		&=\sinh\frac{\varphi_{k+1}^{l+1}-\varphi_{k+1}^l+\varphi_k^{l+1}-\varphi_k^l}{2}\sinh\frac{\varphi_{k+1}^{l+1}-\varphi_{k+1}^l-\varphi_k^{l+1}+\varphi_k^l}{2},
		\end{align}
		which can be recast to
		\begin{align}
		&({1+b^2\lambda^2})\sin \frac{\tilde{x}_{k+1}^{l+1}-\tilde{x}_{k+1}^l+\tilde{x}_k^{l+1}-\tilde{x}_k^l+4b\lambda+4\omega_1}{2}\sin\frac{\tilde{x}_{k+1}^{l+1}-\tilde{x}_{k+1}^l-\tilde{x}_k^{l+1}+\tilde{x}_k^l}{2}\notag\\
		&=b^2\lambda^2\sin \frac{u_{k+1}^{l+1}+u_{k+1}^l+u_k^{l+1}+u_{k}^l}{2}
		\sin \frac{u_{k+1}^{l+1}+u_{k+1}^l-u_k^{l+1}-u_{k}^l}{2}.
		\end{align}
		Then we obtain the  fully discrete gsG equation with $\nu=1$. In addition, from \eqref{Dis-1gsg2}-\eqref{Dis-1gsg5}, we know
		\begin{align}
		&\tilde{J}_k^l=\left(\frac{1}{b}\cos\frac{\tilde{x}_k^{l+1}-\tilde{x}_k^l+2b\lambda}{2}+d\sin\frac{\tilde{x}_k^{l+1}-\tilde{x}_k^l+2b\lambda}{2}\right)^2-\lambda^2\sin^2\frac{u_k^{l+1}+u_{k}^l}{2}=\frac{1}{b^2},\label{dis-1gsG3}\\
		&\tilde{I}_k^l=\left(a\cos\frac{\tilde{x}_{k+1}^l-\tilde{x}_k^l-2a\lambda^{-1}}{2}+d\sin\frac{\tilde{x}_{k+1}^l-\tilde{x}_k^l-2a\lambda^{-1}}{2}\right)^2-\lambda^2\sin^2\frac{u_{k+1}^l-u_k^l}{2}={a^2}.\label{dis-1gsG4}
		\end{align}
		$\tilde{I}_k^l$ and $\tilde{J}_k^l$ are conserved quantities because $a^2$ and $\frac{1}{b^2}$ are constants.
	\end{proof}

\end{document}